\documentclass[conference]{IEEEtran}
\IEEEoverridecommandlockouts
\newif\ifiscameraready
\iscamerareadyfalse
\ifiscameraready
\IEEEpubid{\makebox[\columnwidth]{979-8-3503-3587-3/23/\$31.00~
\copyright2023 IEEE \hfill} \hspace{\columnsep}\makebox[\columnwidth]{ }}
\fi

\def\BibTeX{{\rm B\kern-.05em{\sc i\kern-.025em b}\kern-.08em
T\kern-.1667em\lower.7ex\hbox{E}\kern-.125emX}}

\usepackage{cite}

\usepackage{hyperref}
\usepackage{nameref}
\newcounter{mylabelcounter}

\makeatletter
\newcommand{\labeltext}[2]{%
    #1\refstepcounter{mylabelcounter}%
    \immediate\write\@auxout{%
        \string\newlabel{#2}{{1}{\thepage}{{\unexpanded{#1}}}{mylabelcounter.\number\value{mylabelcounter}}{}}%
    }%
}
\makeatother
\usepackage{amsmath}
\usepackage{amsfonts}
\usepackage{amssymb}
\usepackage{amsthm}

\usepackage{cleveref}
\theoremstyle{plain}
\newtheorem{theorem}{Theorem}
\crefname{theorem}{Thm.}{Thms.}
\newtheorem{proposition}[theorem]{Proposition}
\crefname{proposition}{Prop.}{Props.}
\newtheorem{lemma}[theorem]{Lemma}
\crefname{lemma}{Lem.}{Lems.}
\newtheorem*{sublemma*}{Sublemma}

\crefname{corollary}{Cor.}{Cors.}

\theoremstyle{definition}
\newtheorem{definition}[theorem]{Definition}
\crefname{definition}{Def.}{Defs.}
\newtheorem{notation}[theorem]{Notation}

\theoremstyle{remark}
\newtheorem{example}[theorem]{Example}
\newtheorem{remark}[theorem]{Remark}

\crefname{figure}{Fig.}{Figs.}

\crefname{section}{Sect.}{Sects.}
\crefname{appendix}{Appendix}{}

\usepackage{tikz}
\usetikzlibrary{shapes, arrows, arrows.spaced, arrows.meta, chains,positioning, cd,
    decorations, decorations.pathmorphing, decorations.markings, decorations.pathreplacing, automata, backgrounds, petri, matrix, fit, calc, graphs, quotes, bending}

\usepackage{ebproof}
\ebproofnewstyle{small}{separation = .5em, label separation= .1em, right label template= \footnotesize \inserttext}
\usepackage{stackengine}

\usepackage{empheq}

\usepackage{textcomp}
\usepackage{stmaryrd} %
\usepackage{stackengine} %

\usepackage{tablefootnote}
\usepackage{multirow}
\usepackage{diagbox}

\usepackage{algorithm, algpseudocode}

\usepackage{pict2e,picture}
\usepackage{xparse}
\usepackage{accents}
\usepackage{xcolor}
\definecolor[named]{ACMBlue}{cmyk}{1,0.1,0,0.1}
\definecolor[named]{ACMYellow}{cmyk}{0,0.16,1,0}
\definecolor[named]{ACMOrange}{cmyk}{0,0.42,1,0.01}
\definecolor[named]{ACMRed}{cmyk}{0,0.90,0.86,0}
\definecolor[named]{ACMLightBlue}{cmyk}{0.49,0.01,0,0}
\definecolor[named]{ACMGreen}{cmyk}{0.20,0,1,0.19}
\definecolor[named]{ACMPurple}{cmyk}{0.55,1,0,0.15}
\definecolor[named]{ACMDarkBlue}{cmyk}{1,0.58,0,0.21}
\hypersetup{colorlinks,
    linkcolor=ACMPurple,
    citecolor=ACMPurple,
    urlcolor=ACMDarkBlue,
    filecolor=ACMDarkBlue}

\usepackage{enumitem}
\makeatletter
\let\orgdescriptionlabel\descriptionlabel
\renewcommand*{\descriptionlabel}[1]{%
    \let\orglabel\label
    \let\label\@gobble
    \phantomsection
    \edef\@currentlabel{#1\unskip}%
    \let\label\orglabel
    \orgdescriptionlabel{#1}%
}
\makeatother

\usepackage{etoolbox}%
\patchcmd{\footnotemark}{\stepcounter{footnote}}{\refstepcounter{footnote}}{}{}
\tikzset{
    every edge/.append style = {
            line width = .3pt,
        },
    vert/.style={draw, circle, fill = black!10, inner sep = .15em, font = \footnotesize},
    plab/.style={line width = 0.1pt, fill=#1, inner sep = .025cm, anchor=center, font = \fontsize{6pt}{0}},
    plab/.default= white,
    elab/.style={draw, rectangle, line width = 0.1pt, fill=#1, inner sep = .035cm, anchor=center, font = \footnotesize},
    elab/.default= white,
    tlab/.style={line width = 0.1pt, fill=#1, inner sep = .025cm, anchor=center, font = \fontsize{6pt}{0}\selectfont},
    tlab/.default= white
}
\tikzset{
    png export/.style={
            external/system call/.add={}{; convert -density 300 -transparent white "\image.pdf" "\image.png"},
            /pgf/images/external info,
            /pgf/images/include external/.code={
                    \includegraphics[width=\pgfexternalwidth,height=\pgfexternalheight]{##1.png}
                },
        }
}

\tikzstyle{mynode} = [inner sep = 1.5pt, fill= gray!20]
\tikzstyle{mysmallnode} = [inner sep = 1.pt, fill= gray!20]

\tikzset{earrow/.style={>={{[flex] Latex[length=.1cm, width=2.5pt]}}}}

\tikzset{homoarrow/.style={earrow,  line width = 1pt, dotted, color = green, opacity=0.6}}

\makeatletter
\def\widebreve{\mathpalette\wide@breve}
\def\wide@breve#1#2{\sbox\z@{$#1#2$}%
    \mathop{\vbox{\m@th\ialign{##\crcr
                \kern0.08em\brevefill#1{0.8\wd\z@}\crcr\noalign{\nointerlineskip}%
                $\hss#1#2\hss$\crcr}}}\limits}
\def\brevefill#1#2{$\m@th\sbox\tw@{$#1($}%
        \hss\resizebox{#2}{\wd\tw@}{\rotatebox[origin=c]{90}{\upshape(}}\hss$}
\makeatletter 
\begin{document}

\title{Existential Calculi of Relations with Transitive Closure: Complexity and Edge Saturations\\
    \thanks{This work was supported by JSPS KAKENHI Grant Number JP21K13828.}
}

\author{\IEEEauthorblockN{Yoshiki Nakamura}
    \IEEEauthorblockA{Tokyo Institute of Technology\\
        Email: nakamura.yoshiki.ny@gmail.com
    }
}

\maketitle

\begin{abstract}
We study the decidability and complexity of equational theories of \emph{the existential calculus of relations with transitive closure} (ECoR*) and its fragments,
where ECoR* is the positive calculus of relations with transitive closure extended with complements of term variables and constants.
We give characterizations of these equational theories by using \emph{edge saturations} and we show that the equational theory is
1) coNP-complete for ECoR* without transitive closure;
2) in coNEXP for ECoR* without intersection and PSPACE-complete for two smaller fragments;
3) $\Pi_{1}^{0}$-complete for ECoR*.
The second result gives PSPACE-upper bounds for some extensions of Kleene algebra, including Kleene algebra with top w.r.t.\ binary relations.
 \end{abstract}


\newcommand{\set}[1]{\{ #1 \}}
\newcommand{\tuple}[1]{\langle #1 \rangle}
\newcommand{\const}[1]{\mathsf{#1}}
\newcommand{\bl}{\_}
\NewDocumentCommand\struc{O{1}}{%
    \ifcase#1
        undefined
    \or \mathfrak{A}
    \or \mathfrak{B}
    \else undefined
        
    \fi
}
\newcommand{\ppstruc}{\ddot{\struc}}

\newcommand{\defeq}{\mathrel{\ensurestackMath{\stackon[1pt]{=}{\scriptscriptstyle\Delta}}}}
\newcommand{\defiff}{\mathrel{\ensurestackMath{\stackon[1pt]{\iff}{\scriptscriptstyle\Delta}}}}

\newcommand{\Term}{\mathrm{Term}}
\NewDocumentCommand\term{O{1}}{%
    \ifcase#1
        undefined
    \or t
    \or s
    \or {u}
    \else undefined
        
    \fi
}
\NewDocumentCommand\tset{O{1}}{%
    \ifcase#1
        undefined
    \or T
    \or S
    \or {U}
    \else undefined
        
    \fi
}
\NewDocumentCommand\aterm{O{1}}{%
    \ifcase#1
        undefined
    \or a
    \or b
    \or c
    \else undefined
        
    \fi
}
\NewDocumentCommand\Ccompo{m m}{#1 \cdot #2}
\NewDocumentCommand\Cdag{m m}{#1 \dagger #2}
\NewDocumentCommand\Ccup{m m}{#1 \cup #2}
\NewDocumentCommand\Ccap{m m}{#1 \cap #2}
\NewDocumentCommand\Ccompl{m}{#1^{-}}

\NewDocumentCommand\Cid{}{\mathbf{1}}

\newcommand{\CoR}{\mathrm{CoR}}
\newcommand{\CoRFull}{\CoR({\Ccompo{}{}}, \Ccup{}{}, \Ccompl{\bullet}, \Cdag{}{})}
\newcommand{\CoRCC}{\CoR({\Ccompo{}{}}, \Ccompl{\bullet})}
\newcommand{\CoRUC}{\CoR({\Ccup{}{}, \Ccompl{\bullet}})}
\newcommand{\CoRCU}{\CoR({\Ccompo{}{}}, \Ccup{}{})}
\newcommand{\CoRCUC}{\CoR({\Ccompo{}{}}, \Ccup{}{}, \Ccompl{\bullet})}
\newcommand{\CoRCDC}{\CoR({\Ccompo{}{}}, \Cdag{}{}, -)}
\newcommand{\CoRCDUC}{\CoR({\Ccompo{}{}}, \Cdag{}{}, \bullet^{(-)})}

\newcommand{\PCoR}{\textrm{PCoR}}
\newcommand{\PCoRTC}{{\PCoR}*}

\newcommand{\ECoR}{\textrm{ECoR}}
\newcommand{\ECoRTC}{{\ECoR}*}

\newcommand{\CoRTC}{{\CoR}*}

\newcommand{\KA}{\textrm{KA}}
\newcommand{\KAC}{\textrm{KAC}}
\newcommand{\ExKA}{\textrm{ExKA}}

\newcommand{\Termset}{\mathcal{T}}
\newcommand{\SIG}{\Sigma}
\newcommand{\sig}{S}

\NewDocumentCommand\model{O{1}}{%
    \ifcase#1
        undefined
    \or \mathfrak{A}
    \or \mathfrak{B}
    \else undefined
        
    \fi
}
\NewDocumentCommand\ppmodel{O{1}}{\ddot{\model[#1]}}

\newcommand{\EqT}{\mathrm{EqT}}
\section{Introduction}\label{section: introduction}
\emph{The calculus of relations} (CoR, for short) \cite{Tarski1941} is an algebraic system on binary relations.
As binary relations appear everywhere in computer science, CoR and relation algebras (including their transitive closure extensions) can be applied to various areas, such as databases and program development and verification.
However, the equational theory of CoR is \emph{undecidable} \cite{Tarski1987}.
The undecidability holds even over the signature $\set{\cdot, \cup, \bl^{-}}$
of composition ($\cdot$), union ($\cup$), and complement ($\bl^{-}$) \cite[Thm.\ 1]{Hirsch2018}, moreover even when the number of term variables is one \cite{Nakamura2019}.

One approach to avoid the undecidability of the equational theory of CoR is to consider its \emph{positive} fragments, by excluding complements.
The terms of \emph{the positive calculus of relations} ($\PCoR$, for short) \cite{Andreka1995, Pous2018}
is the set of terms over the signature $\set{\cdot, \cup, \bl^{\smile}, \bot, \top, \const{I}, \cap}$
of composition ($\cdot$), union ($\cup$), converse ($\bl^{\smile}$), the empty relation ($\bot$), the universal relation ($\top$), the identity relation ($\const{I}$), and intersection ($\cap$).
The equational theory of $\PCoR$ is \emph{decidable} \cite{Andreka1995}.
This decidability result also holds when adding the reflexive transitive operator, thus arriving at Kleene allegory terms \cite{Brunet2015, Brunet2017, Nakamura2017}.

Then, it is natural to ask about the decidability and complexity of $\PCoR$ when the complement operator is added in restricted ways (cf.\ e.g., \emph{antidomain} \cite{Hollenberg1997, Desharnais2009} and \emph{tests} in Kleene algebra with tests (KAT) \cite{Kozen1997a}) to extend the expressive power of $\PCoR$ without drastically increasing the complexity.
In this paper, inspired by the negations of atomic programs in the context of propositional dynamic logic (PDL) \cite{Lutz2002, Lutz2005_neg, Goller2009},
we consider \emph{existential calculi of relations}---positive calculi of relations extended with complements of term variables and constants (i.e., the complement operator only applies to term variables or constants)---and study the decidability and complexity of their equational theories.
We denote by $\PCoRTC$ the $\PCoR$ with the reflexive transitive closure operator ($\bl^{*}$).
We denote by $\ECoR$ (resp.\ $\ECoRTC$) the $\PCoR$ (resp.\ $\PCoRTC$) with complements of term variables and constants.
For example, whereas $(a \cdot a)^{-}$ and $(a^{*})^{-}$ are not $\ECoR$ terms, $\ECoR$ contains terms such as:
\begin{center}
    \begin{tabular}{cr}
        \multirow{2}{*}{$a^{-}$}                                      & \multirow{2}{*}{$\left(\begin{array}{l}
                                                                                                           \mbox{there is \emph{not} an ($a$-labeled) edge} \\
                                                                                                           \mbox{from the source to the target}
                                                                                                       \end{array}\right)$} \\
                                                                      &                                                                                      \\
        \multirow{2}{*}{$ (a \cdot a) \cap a^{-} \cap \const{I}^{-}$} & \multirow{2}{*}{$\left(\begin{array}{l}
                                                                                                           \mbox{the distance from the source} \\
                                                                                                           \mbox{to the target is two}
                                                                                                       \end{array}\right)$}                         \\
                                                                      & 
    \end{tabular}
\end{center}
They are not expressible in $\PCoR$ because they are not preserved under homomorphisms; thus, $\ECoR$ is strictly more expressive than $\PCoR$.
More precisely, w.r.t.\  binary relations,
while $\PCoR$ has the same expressive power as the three-variable fragment of \emph{existential positive} logic with equality \cite[Cor.\ 3.2]{nakamuraExpressivePowerSuccinctness2022},
$\ECoR$ has the same expressive power as the three-variable fragment of \emph{existential} logic with equality \cite[Cor.\ 3.14]{nakamuraExpressivePowerSuccinctness2022}\footnote{
    $\ECoR$ has the same expressive power as the level $\Sigma_1^{\CoR}$ of the \emph{dot-dagger alternation hierarchy} \cite{nakamuraExpressivePowerSuccinctness2020}\cite[Def.\ 3.13]{nakamuraExpressivePowerSuccinctness2022} of $\CoR$, because every $\ECoR$ term is a $\Sigma_1^{\CoR}$ term,
    and conversely, every $\Sigma_1^{\CoR}$ term has an equivalent $\ECoR$ term by pushing complement and projection deeper into the term to the extent possible.
    Here, $\Sigma_1^{\CoR}$ has the same expressive power as the three-variable fragment of existential logic (i.e., the level $\Sigma_1$ of the quantifier alternation hierarchy of first-order logic) \cite[Cor.\ 3.14]{nakamuraExpressivePowerSuccinctness2022}.}.
(The name ``\emph{existential} calculi of relations'' comes from this fact.)

On the equational theory of $\ECoRTC$, while it subsumes that of $\PCoRTC$ (including axioms of Kleene algebra and allegory \cite{Pous2018}),
some non-trivial valid (in)equations are as follows\ifiscameraready :\else{} (see \Cref{section: equations} for the proofs):\fi\\
\begin{minipage}[!t]{0.45\linewidth}\leavevmode
    \begin{align}
        \label{equation: I- top 0}      \top            & = a \cup a^{-}                               \\
        \label{equation: I- top}              a b a^{-} & \le \const{I}^{-} a^{-} \cup a \const{I}^{-}
    \end{align}
\end{minipage} %
\hfill %
\begin{minipage}[!t]{0.45\linewidth}%
    \begin{align}
        \label{equation: I- ID}               a                & \le \const{I}^{-} \cup a a \\
        \label{equation: I- De Morgan}        a^{\smile} a^{-} & \le \const{I}^{-}
    \end{align}
\end{minipage}
\subsection*{Contribution and related work}
We show that the equational theory is
\begin{description}
    \item[1)\label{result: coNP}] coNP-complete for $\ECoR$ (\Cref{theorem: PCoR complexity});
    \item[2)\label{result: coNEXP}] in coNEXP for $\ECoR$ without $\cap$ (\Cref{theorem: KACC complexity});
        PSPACE-complete for $\ECoR$ without $\cap$ such that $\const{I}^{-}$ does not occur (\Cref{theorem: KACC without I-}) 
        and for $\ECoR$ without $\cap$ such that $a^{-}$ does not occur for any term variable $a$ (\Cref{theorem: KACC without a-});
    \item[3)\label{result: PI}] $\Pi_{1}^{0}$-complete for $\ECoRTC$ (\Cref{theorem: PCoRTC complexity}).
\end{description}
\Cref{table: result} summarizes our results and related results.
\ref{result: coNEXP} gives PSPACE-upper bounds for some %
\begin{table*}[tbh]
    \caption{Complexity of the equational theory (w.r.t.\ binary relations) over $\sig \cup \sig'$
        where $\sig \subseteq \set{\cdot, \cup, \bl^{\smile}, \cap, \bl^{*}, \bot, \top, \const{I}} = \sig_{\mathrm{PCoR*}}$ and $\sig' \subseteq \set{a^{-}, \const{I}^{-}, \bl^{-}}$.
        Here, $a^{-}$, $\const{I}^{-}$, $\bl^{-}$ denote the complements of term variables, those of constants, and the (full) complement.}\label{table: result}
    \centering
    \begin{tabular}{|c||c|c|c|c|c|}
        \hline
        \diagbox[height=1.2\line]{$\sig$\quad}{\quad$\sig'$}                                                              & $\emptyset$                                                                      & 
        $\set{\const{I}^{-}}$                                                                                             & $\set{a^{-}}$                                                                    & $\set{a^{-},\const{I}^{-}}$                               & $\set{\bl^{-}}$                                                                                             \\
        \hline
        \hline
        \multirow{2}{*}{$\set{\cdot, \cup} \subseteq \sig \subseteq \sig_{\mathrm{PCoR*}} \setminus \set{{\bl^{*}}}$}     & coNP-c                                                                           & 
        \multicolumn{3}{|c|}{coNP-c}                                                                                      & $\Sigma_{1}^{0}$-c                                                                                                                                                                                                                                         \\
                                                                                                                          & (\hspace{1sp}\cite{Meyer1973}, \cite{Hunt1976}, \cite[Lem.\ 8]{Hirsch2018})      & \multicolumn{3}{|c|}{(\Cref{theorem: PCoR complexity})}   & (\hspace{1sp}\cite{Tarski1941,Tarski1987,Hirsch2018})                                                       \\
        \hline
        \multirow{2}{*}{$\set{\cdot, \cup, \bl^{*}} \subseteq \sig \subseteq \sig_{\mathrm{PCoR*}} \setminus \set{\cap}$} & PSPACE-c                                                                         & PSPACE-c                                                  & PSPACE-c                                              & in coNEXP                         & 
        \\
                                                                                                                          & (\hspace{1sp}\cite{Meyer1973}, \cite{Hunt1976}, \Cref{theorem: KACC without I-}) & (\Cref{theorem: KACC without a-})                         & (\Cref{theorem: KACC without I-})                     & (\Cref{theorem: KACC complexity}) & $\Pi_{1}^{1}$-c
        \\
        \cline{1-5}
        \multirow{2}{*}{$\set{\cdot, \cup, \cap, \bl^{*}} \subseteq \sig  \subseteq \sig_{\mathrm{PCoR*}}$}               & EXPSPACE-c                                                                       & EXPSPACE-hard                                             & \multicolumn{2}{|c|}{$\Pi_1^{0}$-c}                   & 
        (\hspace{1sp}\cite[p.332]{Gradel1999}, \cite{Goller2009}\footnotemark\label{footnote: lowerbound})
        \\
                                                                                                                          & (\hspace{1sp}\cite{Brunet2015, Brunet2017, Nakamura2017})
                                                                                                                          & \cite{Brunet2017, Nakamura2017}                                                  & \multicolumn{2}{|c|}{(\Cref{theorem: PCoRTC complexity})} &                                                                                                             %
        \\
        \hline
    \end{tabular}
\end{table*}
\footnotetext{On the lower bound, the validity problem for PDL with intersection and negation of atomic programs, which is $\Pi_{1}^{1}$-complete \cite{Goller2009}, is recursively reducible to the equational theory over the signature $\set{\cup, \cdot, \bl^*, \bl^-}$ by the standard translation from modal logics to $\CoR$ \cite[p.\ 95]{Orlowska1997} while eliminating the identity using a fresh variable preserving the validity \cite[Lem.\ 9]{Nakamura2019}.\label{footnote: lowerbound}}%
extensions of \emph{Kleene algebra} (KA) w.r.t.\ binary relations, as $\ECoR$ without $\cap$ includes KA terms (i.e., terms over the signature $\set{\cdot, \cup, \bl^{*}, \bot, \const{I}}$).
Notably, it includes KA with converse ($\smile$) \cite{Brunet2014, Brunet2016} and KA with top ($\top$) \cite{Pous2022} w.r.t.\ binary relations; thus, the result gives a positive answer to a question posed by Pous and Wagemaker \cite[p.\ 14]{Pous2022} as follows:
the equational theory of KA with top w.r.t.\ binary relations is still in PSPACE.\footnote{Very recently, this result is also given in \cite{pousCompletenessTheoremsKleene2023}.}
Additionally, \ref{result: PI} negatively answers a question posed by Nakamura \cite[p.\ 12]{Nakamura2017} as follows:
the equational theory of Kleene allegories with complements of term variables (w.r.t.\ binary relations) is undecidable.

To show them, %
we give a \emph{graph theoretical characterization} of the equational theory of $\ECoRTC$ (\Cref{section: graph}).
The characterization using graph languages and graph homomorphisms is based on that for $\PCoRTC$, by Brunet and Pous {\cite[Thm.\ 3.9]{Brunet2017}}\hspace{1sp}{\cite[Thm.\ 16]{Pous2018}},
but the identical characterization fails for $\ECoRTC$.
Nevertheless, we can extend such a characterization even for $\ECoRTC$, by extending their graph languages, using \emph{edge saturations} (cf.\ saturations of graphs \cite{joostenFindingModelsGraph2018, Doumane2021}; in our saturation, the vertex set is fixed).
Using this graph characterization, we can show the upper bounds of \ref{result: coNP} and \ref{result: PI}.

For the upper bound of \ref{result: coNEXP}: the equational theory of $\ECoRTC$ without intersection,
we refine the graph characterization above by using path graphs with some additional information, called \emph{saturable paths} (\Cref{section: automata}).
This notion is inspired by Hintikka-trees for Boolean modal logics \cite{Lutz2002} and PDL with the negation of atomic programs \cite{Lutz2005_neg},
where we consider paths instead of infinite trees and introduce the converse and the difference relation $\const{I}^{-}$.
This characterization gives the coNEXP upper bound for the equational theory of $\ECoRTC$ without intersection.
Moreover, for some fragments, we can give \emph{word automata} using saturable paths, which shows that 
the equational theory of $\ECoRTC$ without intersection is decidable in PSPACE if $\const{I}^{-}$ does not occur (\Cref{theorem: KACC without I-})
or if $a^{-}$ does not occur (for any term variable $a$) (\Cref{theorem: KACC without a-}).
(Our automata construction cannot apply to the full case (\Cref{remark: full KACC}).)

For the lower bounds, \ref{result: coNP} and \ref{result: coNEXP} are immediate from the results known in regular expressions \cite{Meyer1973,Hunt1976}.
For \ref{result: PI}, we give a reduction from the universality problem for context-free grammars (\Cref{section: undecidability}), via KA with hypotheses \cite{cohenHypothesesKleeneAlgebra1994, kozenComplexityReasoningKleene2002, Doumane2019} of the form $\term \le a$ (where $\term$ is a term and $a$ is a term variable) w.r.t.\ binary relations.

\subsection*{Outline}
\Cref{section: preliminaries} briefly gives basic definitions.
\Cref{section: ECoR} defines $\ECoRTC$ (\Cref{section: ECoR def}) and recalls known results w.r.t.\ word and graph languages for smaller fragments of $\ECoRTC$ (\Cref{subsection: languages,subsection: graph languages}).
In \Cref{section: graph}, we give a graph theoretical characterization for $\ECoRTC$, using \emph{edge saturations}.
By using this, we prove \ref{result: coNP} and the upper bound of \ref{result: PI}.
In \Cref{section: automata}, we introduce \emph{saturable paths}, which refine the characterization of graph saturations, for $\ECoRTC$ without intersection.
Moreover, we give automata using saturable paths for two smaller fragments.
By using them, we prove \ref{result: coNEXP}.
In \Cref{section: undecidability}, we prove the lower bound of \ref{result: PI}.
\Cref{section: conclusion} concludes this paper.
\newcommand{\dom}{\mathord{\const{dom}}}
\newcommand\pfun{\mathrel{\ooalign{\hfil$\mapstochar\mkern5mu$\hfil\cr$\to$\cr}}}
\newcommand{\nat}{\mathbb{N}}
\newcommand{\pnat}{\mathbb{N}_{+}}
\newcommand{\znat}{\mathbb{Z}}
\newcommand{\Occ}{\mathop{\mathrm{Occ}}}
\newcommand{\card}{\#}

\newcommand{\jump}[1]{\llbracket #1 \rrbracket}

\newcommand{\Fml}{\mathrm{Fml}}
\NewDocumentCommand\fml{O{1}}{%
    \ifcase#1
        undefined
    \or \varphi
    \or \psi
    \or \rho
    \else undefined
        
    \fi
}
\NewDocumentCommand\afml{O{1}}{%
    \ifcase#1
        undefined
    \or p
    \or q
    \or r
    \else undefined
        
    \fi
}
\NewDocumentCommand\var{O{1}}{%
    \ifcase#1
        undefined
    \or x
    \or y
    \or z
    \else undefined
        
    \fi
}

\newcommand{\len}[1]{\|#1\|}
\newcommand{\atlen}[1]{\|#1\|_{\mathrm{at}}}

\newcommand{\comnf}[1]{\overline{#1}}

\newcommand{\pt}[1]{\mathsf{#1}}
\newcommand{\lv}{\pt{1}}
\newcommand{\rv}{\pt{2}}
\newcommand{\src}{\lv}
\newcommand{\tgt}{\rv}

\NewDocumentCommand\graph{O{1}}{%
    \ifcase#1
        undefined
    \or G
    \or H
    \else undefined
        
    \fi
}
\NewDocumentCommand\glang{O{1}}{%
    \ifcase#1
        undefined
    \or \mathcal{G}
    \or \mathcal{H}
    \else undefined
        
    \fi
}
\newcommand{\homo}{\longrightarrow}

\newcommand{\REL}{\mathrm{REL}}

\section{Preliminaries}\label{section: preliminaries}
We write $\nat$ for the set of all non-negative integers.
For $l, r \in \nat$, we write $[l, r]$ for the set $\{i \in \nat \mid l \le i \le r\}$.
For $n \in \nat$, we write $[n]$ for $[1, n]$.
For a set $A$, we write $\card(A)$ for the cardinality of $A$ and $\wp(A)$ for the power set of $A$.

\subsection{Graphs}\label{section: preliminary graphs}
For $k \in \nat$, a \emph{$k$-pointed graph} $\graph$ over a set $A$ is a tuple $\tuple{|\graph|, \set{a^{\graph}}_{a \in A}, \pt{1}^{\graph}, \dots, k^{\graph}}$, where
\begin{itemize}
    \item $|\graph|$ is a non-empty set of \emph{vertices};
    \item $\aterm^{\graph} \subseteq |\graph|^2$ is a binary relation for each $\aterm \in A$ ($\tuple{x, y} \in \aterm^{\graph}$ denotes that there is an \emph{$\aterm$-labeled edge} from $x$ to $y$);
    \item $\pt{1}^{\graph}, \dots, k^{\graph} \in |\graph|$ are the vertices pointed by $\pt{1}, \dots, k$.
\end{itemize}
We say that $2$-pointed graphs are \emph{graphs}, here; we mainly use them.
For $2$-pointed graph $\graph$, we say that $\lv^{\graph}$ and $\rv^{\graph}$ are the \emph{source} and \emph{target}, respectively.

A \emph{(graph) homomorphism} from a graph $\graph[1]$ to a graph $\graph[2]$ is a map $h$ from $|\graph[1]|$ to $|\graph[2]|$ such that
\begin{itemize}
    \item $\tuple{h(\lv^{\graph[1]}), h(\rv^{\graph[1]})} = \tuple{\lv^{\graph[2]}, \rv^{\graph[2]}}$;
    \item $\tuple{h(x), h(y)} \in a^{H}$ for every $a \in A$ and $\tuple{x, y} \in a^{G}$.
\end{itemize}
In particular, we say that $h$ is \emph{isomorphism} if $h$ is bijective and $\tuple{h(x), h(y)} \not \in a^{\graph[2]}$ for every $a \in A$ and $\tuple{x, y} \not\in a^{\graph[1]}$.
We write $h \colon \graph[1] \homo \graph[2]$ if $h$ is a homomorphism from $\graph[1]$ to $\graph[2]$
and write $\graph[1] \homo \graph[2]$ if $h \colon \graph[1] \homo \graph[2]$ for some $h$.
The relation $\homo$ is a preorder.
We display graphs in a standard way,
where the node having an ingoing (resp.\ outgoing) unlabeled arrow denotes the source (resp.\ the target).
For example, the following are two connected graphs, and dotted arrows induce a homomorphism between them:
$
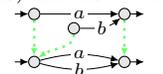
\begin{tikzpicture}[baseline = -2.2ex]
        \graph[grow right = 1.2cm, branch down = 1.2ex, nodes={mynode}]{
        {0/{}[draw, circle], 2/{}[draw, circle, xshift = 1.5em]}-!-{1/{}[draw, circle]}
        };
        \node[left = .5em of 0](l){};
        \node[right = .5em of 1](r){};
        \graph[use existing nodes, edges={color=black, pos = .5, earrow}, edge quotes={fill=white, inner sep=1.pt,font= \scriptsize}]{
        0 ->["$a$"] 1;
        2 ->["$b$", out = 0, in = -135] 1;
        l -> 0; 1 -> r;
        };
        \node[below = 3ex of 0, draw, circle, mynode](0'){};
        \node[below = 3ex of 1, draw, circle, mynode](1'){};
        \node[left = .5em of 0'](l'){};
        \node[right = .5em of 1'](r'){};
        \graph[use existing nodes, edges={color=black, pos = .5, earrow}, edge quotes={fill=white, inner sep=1.pt,font= \scriptsize}]{
        0' ->["$a$", bend left = 15] 1';
        0' ->["$b$", bend right = 15] 1';
        l' -> 0'; 1' -> r';
        };
        \path (0) edge [homoarrow,->] (0');
        \path (2) edge [homoarrow,->, out=-90, in=45] (0');
        \path (1) edge [homoarrow,->] (1');
    \end{tikzpicture}$.

\section{Existential calculi of relations}\label{section: ECoR}
\subsection{The existential calculus of relations with transitive closure (\ECoRTC): syntax and semantics}\label{section: ECoR def}
\subsubsection{Syntax}
We fix a finite set $\SIG$ of \emph{variables}.
The set of \emph{$\ECoRTC$ terms} over $\SIG$ is defined as follows:\footnote{For simplicity, in the term set, the converse only applies to terms of the form $a$ or $a^{-}$ and $\bot^{-}$ and $\top^{-}$ does not occur, but we can give a polynomial-time transformation to the term set \ifiscameraready \else (\Cref{section: connf}) \fi by taking the \emph{converse normal form}
    and that $\models_{\REL} \bot^{-} = \top$ and $\models_{\REL} \top^{-} = \bot$ hold.
    Thus, our complexity upper bounds hold (\Cref{theorem: PCoRTC complexity,theorem: PCoR complexity,theorem: KACC without a-,theorem: KACC without I-}) even if we exclude these restrictions.\label{footnote: term set}}
\begin{align*}
    \ECoRTC \ni \term[1], \term[2], \term[3] & ::= \aterm \mid \aterm^{-} \mid \aterm^{\smile} \mid (\aterm^{-})^{\smile} \mid \const{I} \mid \const{I}^{-} \mid \bot \mid \top   \\
                                             & \mid \term[1] \cdot \term[2] \mid \term[1] \cup \term[2] \mid \term[1] \cap \term[2] \mid \term[1]^{*} \tag*{($\aterm \in \SIG$).}
\end{align*}
We use parentheses in ambiguous situations.
We often abbreviate $t \cdot s$ to $t s$.
For $n \in \nat$, we write $\term^{n}$ for $\begin{cases}
        \term \cdot \term^{n-1} & (n \ge 1) \\
        \const{I}               & (n = 0)
    \end{cases}$.

Let $\sig_{\ECoRTC} \defeq \set{\cdot, \cup, \bl^{\smile}, \cap, \bl^{*}, \const{\bot}, \const{\top}, \const{I}} \cup \set{a^{-}, \const{I}^{-}}$.
Here, ``$\bl^{\smile}$'' only applies to $a$ or $a^{-}$ (for simplicity)
and we use ``$a^{-}$'' to denote the complement of term variables and ``$\const{I}^{-}$'' to denote the complement of constants (the important complemented constant is only $\const{I}^{-}$ in this setting).

For $\sig \subseteq \sig_{\ECoRTC}$,
we write $\Termset_{\sig}$ for the set of all terms $\term$ in $\ECoRTC$ s.t.\ every operator occurring in $\term$ matches one of $\sig$.
We use the following acronyms for some signatures (recall the acronyms in \Cref{section: introduction}; here, \emph{extended KA} ($\ExKA$) terms are used to denote $\ECoRTC$ terms without $\cap$, in this paper):
\begin{table}[h]
    \centering
    \begin{tabular}{c|l}
        $\sig_{\mathop{acronym}}$ & operator set                                                                                                    \\
        \hline \hline
        ${\ECoRTC}$               & ${\set{\cdot, \cup, \bl^{\smile}, \cap, \bl^{*}, \bot, \top, \const{I}, a^{-}, \const{I}^{-}}}$                 \\
        ${\ECoR}$                 & $\sig_{\ECoRTC} \setminus \set{\bl^{*}}$                                                                        \\
        ${\ExKA}$                 & $\sig_{\ECoRTC} \setminus \set{\cap} \quad (= \sig_{\KA} \cup \set{\bl^{\smile}, \top, a^{-}, \const{I}^{-}})$  \\
        \hline
        ${\PCoRTC}$               & $\sig_{\ECoRTC} \setminus \set{a^{-}, \const{I}^{-}}$                                                           \\
        ${\PCoR}$                 & $\sig_{\PCoRTC} \setminus \set{\bl^{*}}$                                                                        \\
        ${\KA}$                   & $\sig_{\PCoRTC} \setminus \set{\bl^{\smile}, \cap, \top} \quad (= \set{\cdot, \cup, \bl^{*}, \bot, \const{I}})$ %
    \end{tabular}
\end{table}

Let $\SIG_{\const{I}} \defeq \SIG \cup \set{\const{I}}$ and let
\begin{align*}
    \SIG^{(-)} & \defeq \set{a, a^{-} \mid a \in \SIG}; & \SIG_{\const{I}}^{(-)} & \defeq \set{a, a^{-} \mid a \in \SIG_{\const{I}}}.
\end{align*}
For each term $\term \in \SIG_{\const{I}}^{(-)} \cup \set{\bot, \top}$, $\comnf{\term}$ denotes the following term, where $a \in \SIG$:
\begin{align*}
    \comnf{a} & \defeq a^{-}; & \comnf{a^{-}} & \defeq a; & \comnf{\const{I}} & \defeq \const{I}^{-}; & \comnf{\const{I}^{-}} & \defeq \const{I}; & \comnf{\bot} & \defeq \top; & \comnf{\top} & \defeq \bot \tag*{.}
\end{align*}

An \emph{equation} $\term[1] = \term[2]$ is a pair of terms.
An \emph{inequation} $\term[1] \le \term[2]$ is an abbreviation of the equation $\term[1] \cup \term[2] = \term[2]$.\ifiscameraready\else\footnote{Note that $\jump{\term[1]}_{\model} \subseteq \jump{\term[2]}_{\model} \Longleftrightarrow \jump{\term[1] \cup \term[2]}_{\model} = \jump{\term[2]}_{\model}$.}\fi

The \emph{size} $\len{\term}$ of a term $\term$ is the number of symbols occurring in $\term$.
Also, let $\len{\term[1] = \term[2]} \defeq \len{\term[1]} + \len{\term[2]}$.

\subsubsection{Relational semantics}
For binary relations $R, Q$ on a set $W$, the \emph{relational converse} $R^{\smile}$, the \emph{relational composition} $R \cdot Q$, the \emph{$n$-th iteration} $R^{n}$ (where $n \in \nat$), and the \emph{reflexive transitive closure} $R^*$ are defined by:
\begin{align*}
    R^{\smile} & \defeq \set{\tuple{y, x} \mid \tuple{x, y} \in R}                                                                 \\
    R \cdot Q  & \defeq \set{\tuple{x, z} \mid \exists y \in W.\ \tuple{x, y} \in R \ \land \ \tuple{y, z} \in Q}\span \span       \\
    R^n        & \defeq \begin{cases}
                            R \cdot R^{n-1}                 & (n \ge 1) \\
                            \set{\tuple{x, x} \mid x \in W} & (n = 0)
                        \end{cases};                                                      & R^* & \defeq \bigcup_{n \in \nat} R^n.
\end{align*}

\begin{definition}
    We say that $\graph$ is a \emph{$k$-pointed structure} if
    $\graph$ is a $k$-pointed graph over $\SIG_{\const{I}}^{(-)}$ such that 
    \begin{itemize}
        \item $\const{I}^{\graph[1]} = \set{\tuple{x, x} \mid x \in |\graph[1]|}$;
        \item $\comnf{a}^{\graph[1]} = |\graph[1]|^2 \setminus a^{\graph[1]}$ for $a \in \SIG_{\const{I}}$.\footnote{$a^{\graph[1]}$ is redundant for each $a \in \SIG_{\const{I}}^{(-)} \setminus \SIG$ because it is determined by the other relations. However, this definition is compatible with the later.}
    \end{itemize}
\end{definition}
We say that $0$-pointed structures are \emph{structures}.
We use $\model$ and $\ppmodel$ to denote $0$- and $2$-pointed structures, respectively.
For a structure $\model$ and two vertices $x, y \in |\model|$,
we write ${\model}[x, y]$ for the $2$-pointed structure $\tuple{|\model|, \set{a^{\model}}_{a \in \SIG_{\const{I}}^{(-)}}, x, y}$.

The \emph{binary relation} $\jump{\term}_{\model} \subseteq |\model|^2$ of an $\ECoRTC$ term $\term$ on a structure $\model$ is defined as follows (where $a \in \SIG$):
\begin{flalign*}
    \jump{a}_{\model}                       & \defeq a^{\model}                                              & 
    \jump{a^{-}}_{\model}                   & \defeq (a^{-})^{\model}                                           \\
    \jump{\bot}_{\model}                    & \defeq \emptyset                                               & 
    \jump{\top}_{\model}                    & \defeq |\model|^2                                                 \\
    \jump{\term[1] \cup \term[2]}_{\model}  & \defeq \jump{\term[1]}_{\model} \cup \jump{\term[2]}_{\model}  & 
    \jump{\term[1] \cap \term[2]}_{\model}  & \defeq \jump{\term[1]}_{\model} \cap \jump{\term[2]}_{\model}     \\
    \jump{\const{I}}_{\model}               & \defeq \const{I}^{\model}                                      & 
    \jump{\const{I}^{-}}_{\model}           & \defeq (\const{I}^{-})^{\model}                                   \\
    \jump{\term[1] \cdot \term[2]}_{\model} & \defeq \jump{\term[1]}_{\model} \cdot \jump{\term[2]}_{\model}    \\
    \jump{\term[1]^{*}}_{\model}            & \defeq \bigcup_{n \in \nat} \jump{\term[1]^{n}}_{\model}       & 
    \jump{\term^{\smile}}_{\model}          & \defeq \jump{\term}_{\model}^{\smile}.
\end{flalign*}
We write $\models_{\REL} \term[1] = \term[2]$ if $\jump{\term[1]}_{\model} = \jump{\term[2]}_{\model}$ for every structure $\model$.
The \emph{equational theory over $\sig$ w.r.t.\ binary relations} is defined as the set of all pairs $\term[1] = \term[2]$ of terms in $\Termset_{\sig}$ s.t.\ $\models_{\REL} \term[1] = \term[2]$.
\begin{notation}\label{notation: models}
    Based on $\jump{\bl}_{\model}$, we define the following notations:
    \begin{align*}
        {\model}[x, y] \models \term         & \defiff \tuple{x, y} \in \jump{\term}_{\model}                                  \\
        \ppmodel \models \term[1] = \term[2] & \defiff (\ppmodel \models \term[1]) \leftrightarrow (\ppmodel \models \term[2]) \\
        \model \models \term[1] = \term[2]   & \defiff \jump{\term[1]}_{\model} = \jump{\term[2]}_{\model}.
    \end{align*}
    (Note that $\model \models \term[1] = \term[2] \iff \forall x, y \in |\model|.\  {\model}[x,y] \models \term[1] = \term[2]$.)
\end{notation}

\newcommand{\lang}[1]{[#1]}
\subsection{KA terms w.r.t.\ binary relations and word languages}\label{subsection: languages}
For a set $X$, we write $X^*$ for the set of all finite sequences (i.e., \emph{words}) over $X$.
We write $\const{I}$ for the empty word.
We write $w v$ for the concatenation of words $w$ and $v$.
For set $L, K \subseteq \SIG^*$, the \emph{composition} $L \cdot K$ is defined by:
\[L \cdot K \defeq \set{w v \mid w \in L \ \land \  v \in K}.\]
The \emph{(word) language} $\lang{\term} \subseteq \SIG^*$ of a KA term $\term$ is defined by:
\begin{flalign*}
    \lang{a}                       & \defeq \set{a}                                 &  &  \\
    \lang{\bot}                    & \defeq \emptyset                               &  &  \\
    \lang{\term[1] \cup \term[2]}  & \defeq \lang{\term[1]} \cup \lang{\term[2]}    &  &  \\
    \lang{\const{I}}               & \defeq \set{\const{I}}                         &  &  \\
    \lang{\term[1] \cdot \term[2]} & \defeq \lang{\term[1]} \cdot \lang{\term[2]}   &  &  \\
    \lang{\term[1]^{*}}            & \defeq \bigcup_{n \in \nat} \lang{\term[1]^n}. &  & 
\end{flalign*}
Interestingly, for KA terms, it is well-known that the equational theory w.r.t.\ binary relations coincides with that under the (single) word language interpretation (see, e.g., \cite[Thm.\ 4]{Pous2018}):
for every KA terms $\term[1], \term[2]$, we have
\begin{equation}
    \label{equation: rel and lang}
    {} \models_{\REL} \term[1] \le \term[2] \quad \iff \quad \lang{\term[1]} \subseteq \lang{\term[2]}. \tag{$\dagger$}
\end{equation}
The following follows from (\ref{equation: rel and lang}) and the known results in regular expressions \cite[p.\ 3]{Meyer1973} (\hspace{1sp}\cite{Hunt1976}, for precise constructions):
\begin{proposition}[\hspace{1sp}{\cite[Thm.\ 2.7]{Hunt1976}, \cite{Meyer1973}}]\label{proposition: PCoR hardness 1}
    The equational theory of $\Termset_{\set{\cdot, \cup}}$ (w.r.t.\ binary relations) is coNP-complete.
\end{proposition}
\begin{proposition}[\hspace{1sp}{\cite[Prop.\ 2.4]{Hunt1976}, \cite{Meyer1973}}]\label{proposition: PCoR hardness 2}
    The equational theory of $\Termset_{\set{\cdot, \cup, \bl^{*}}}$ (w.r.t.\ binary relations) is PSPACE-complete.
\end{proposition}
They can show the lower bounds of \Cref{theorem: PCoR complexity,theorem: KACC without I-,theorem: KACC without a-}.
\begin{remark}\label{remark: rel and lang}
    The equivalence (\ref{equation: rel and lang}) breaks if we add $\cap$, $\bl^{\smile}$, or $\top$ by a standard word language interpretation \cite[Sect.\ 4]{Pous2018}.
    Note that (\ref{equation: rel and lang}) also breaks if we add $a^{-}$ or $\const{I}^{-}$ with $\lang{a^{-}} \defeq \SIG^* \setminus \set{a}$
    and $\lang{\const{I}^{-}} \defeq \SIG^* \setminus \set{\const{I}}$, respectively.
    For example, $\models_{\REL} \top a \top b \top = \top b \top a \top$ holds \cite[p.\ 13]{Pous2018}, but their languages are not the same.
    Here, $\top$ abbreviates the term $c \cup \comnf{c}$ (cf.\ \Cref{equation: I- top 0}) where $c$ is any.
    Conversely, for example, $\lang{a^{-}} = \lang{\const{I} \cup (\SIG \setminus \set{a}) \cup (\SIG \cdot \SIG \cdot \SIG^*)}$ holds,
    but such an equation does not hold in relational semantics (where a finite set $L = \set{\term[1]_1, \dots, \term[1]_n}$ of terms abbreviates the term $\term[1]_1 \cup \dots \cup \term[1]_n$, here).
\end{remark}

\newcommand{\Fill}{\mathcal{S}}
\subsection{{\PCoRTC} and graph languages}\label{subsection: graph languages}
Recall \Cref{section: preliminary graphs}.
Using graph languages rather than word languages, we can give a characterization of the equational theory w.r.t.\ binary relations for more general terms (e.g., \cite{Pous2018}).
We use the following three operations on graphs,
\emph{series-composition} ($\cdot$), \emph{parallel-composition} ($\cap$), and \emph{converse} ($\bl^{\smile}$):
\begin{align*}
    G \cdot H
     & \defeq \begin{tikzpicture}[baseline = -.5ex]
                  \graph[grow right = 1.cm, branch down = 2.5ex, nodes={mynode, font = \scriptsize}]{
                  {s1/{}[draw, circle]}
                  -!- {c/{}[draw, circle]}
                  -!- {t1/{}[draw, circle]}
                  };
                  \node[left = 4pt of s1](s1l){} edge[earrow, ->] (s1);
                  \node[right = 4pt of t1](t1l){}; \path (t1) edge[earrow, ->] (t1l);
                  \graph[use existing nodes, edges={color=black, pos = .5, earrow}, edge quotes={fill=white, inner sep=1pt,font= \scriptsize}]{
                  s1 ->["$G$"] c;
                  c ->["$H$"] t1;
                  };
              \end{tikzpicture}
     & G \cap H
     & \defeq \begin{tikzpicture}[baseline = -.5ex]
                  \graph[grow right = 1.cm, branch down = 2.5ex, nodes={mynode, font = \scriptsize}]{
                  {s1/{}[draw, circle]}
                  -!- {t1/{}[draw, circle]}
                  };
                  \node[left = 4pt of s1](s1l){} edge[earrow, ->] (s1);
                  \node[right = 4pt of t1](t1l){}; \path (t1) edge[earrow, ->] (t1l);
                  \graph[use existing nodes, edges={color=black, pos = .5, earrow}, edge quotes={fill=white, inner sep=1pt,font= \scriptsize}]{
                  s1 ->["$G$", bend left = 25] t1;
                  s1 ->["$H$", bend right = 25] t1;
                  };
              \end{tikzpicture} \\
    G^{\smile}
     & \defeq \begin{tikzpicture}[baseline = -.5ex]
                  \graph[grow right = 1.cm, branch down = 2.5ex, nodes={mynode, font = \scriptsize}]{
                  {s1/{}[draw, circle]}
                  -!- {t1/{}[draw, circle]}
                  };
                  \node[left = 4pt of s1](s1l){} edge[earrow, ->] (s1);
                  \node[right = 4pt of t1](t1l){}; \path (t1) edge[earrow, ->] (t1l);
                  \graph[use existing nodes, edges={color=black, pos = .5, earrow}, edge quotes={fill=white, inner sep=1pt,font= \scriptsize}]{
                  t1 ->["$G$"] s1;
                  };
              \end{tikzpicture}.
\end{align*}
The following is inspired by \cite{Brunet2017}, \cite[Def.\ 15]{Pous2018} for $\PCoRTC$, where we extend the definition of $a^{-}$ and $\const{I}^{-}$ for $\ECoRTC$.
\begin{definition}[cf.\ {\cite{Brunet2017}},{\cite[Def.\ 15]{Pous2018}}]\label{definition: graph language ECoR}
    The \emph{graph language} $\glang(t)$ of an $\ECoRTC$ term $t$ is a set of graphs over $\SIG_{\const{I}}^{(-)}$, defined by:
    \begin{align*}
        \glang(a)                               & \defeq \set{
            \begin{tikzpicture}[baseline = -.5ex]
                \graph[grow right = 1.cm, branch down = 6ex, nodes={mynode}]{
                {0/{}[draw, circle]}-!-{1/{}[draw, circle]}
                };
                \node[left = .5em of 0](l){};
                \node[right = .5em of 1](r){};
                \graph[use existing nodes, edges={color=black, pos = .5, earrow}, edge quotes={fill=white, inner sep=1pt,font= \scriptsize}]{
                0 ->["$a$"] 1;
                l -> 0; 1 -> r;
                };
            \end{tikzpicture}
        }  \mbox{ where $a \in \Sigma^{(-)}$}                   \span\span                                                                                                                 \\
        \glang(\bot)                            & \defeq \emptyset                                                                                                              & 
        \glang(\top)                            & \defeq \set{\begin{tikzpicture}[baseline = -.5ex]
                                                                      \graph[grow right = 1.cm, branch down = 6ex, nodes={mynode}]{
                                                                      {0/{}[draw, circle]}-!-{1/{}[draw, circle]}
                                                                      };
                                                                      \node[left = .5em of 0](l){};
                                                                      \node[right = .5em of 1](r){};
                                                                      \graph[use existing nodes, edges={color=black, pos = .5, earrow}, edge quotes={fill=white, inner sep=1pt,font= \scriptsize}]{
                                                                          l -> 0; 1 -> r;
                                                                      };
                                                                  \end{tikzpicture}}                                 \\
        \glang(\term[1] \cup \term[2])          & \defeq \glang(\term[1]) \cup \glang(\term[2])                                                                                            \\
                                                & \hspace{1.5em} \glang(\term[1] \cap \term[2])  \defeq
        \set{G \cap H \mid G \in \glang(\term[1]) \ \land \  H \in \glang(\term[2])}                                                                                \span \span            \\
        \glang(\const{I})                       & \defeq \set{\begin{tikzpicture}[baseline = -.5ex]
                                                                      \graph[grow right = 1.cm, branch down = 6ex, nodes={mynode}]{
                                                                      {0/{}[draw, circle]}
                                                                      };
                                                                      \node[left = .5em of 0](l){};
                                                                      \node[right = .5em of 0](r){};
                                                                      \graph[use existing nodes, edges={color=black, pos = .5, earrow}, edge quotes={fill=white, inner sep=1pt,font= \scriptsize}]{
                                                                          l -> 0; 0 -> r;
                                                                      };
                                                                  \end{tikzpicture}}                                 & 
        \glang(\const{I}^{-})                   & \defeq \set{\begin{tikzpicture}[baseline = -.5ex]
                                                                      \graph[grow right = 1.cm, branch down = 6ex, nodes={mynode}]{
                                                                      {0/{}[draw, circle]}-!-{1/{}[draw, circle]}
                                                                      };
                                                                      \node[left = .5em of 0](l){};
                                                                      \node[right = .5em of 1](r){};
                                                                      \path (0) edge [draw = white, opacity = 0] node[pos= 0.5, font = \scriptsize, opacity = 1](a1){$\comnf{\const{I}}$}(1);
                                                                      \graph[use existing nodes, edges={color=black, pos = .5, earrow}, edge quotes={fill=white, inner sep=1pt,font= \scriptsize}]{
                                                                          0 -- a1 -> 1;
                                                                          l -> 0; 1 -> r;
                                                                      };
                                                                  \end{tikzpicture}} \\
        \glang(\term[1] \cdot \term[2])         & \defeq \set{G \cdot H \mid G \in \glang(\term[1]) \ \land \  H \in \glang(\term[2])}                        \span \span                  \\
        \glang(\term[1]^{*})                    & \defeq
        \bigcup_{n \in \nat} \glang(\term[1]^n) & 
        \glang(\term[1]^{\smile})               & \defeq \set{ G^{\smile} \mid G \in \glang(\term[1])}.
    \end{align*}
\end{definition}
\begin{definition}\label{definition: semantics for graphs}
    For a graph $G$ over $\SIG_{\const{I}}^{(-)}$, the \emph{binary relation} $\jump{G}_{\model} \subseteq |\model|^2$ on a structure $\model$ is defined by:
    \[\tuple{x, y} \in \jump{G}_{\model} \quad \defiff \quad (G \homo {\model}[x, y]).\]
    ($\graph[1] \homo \graph[2]$ denotes that there is a homomorphism from $\graph[1]$ to $\graph[2]$.)
    For each graph language $\glang$, let
    \[\jump{\glang}_{\model} \defeq \bigcup_{G \in \glang} \jump{G}_{\model}.\]
\end{definition}
Based on $\jump{G}_{\model}$ and $\jump{\glang}_{\model}$, we use $\models$ (\Cref{notation: models}) also for graphs and graph languages.
Note that:
\begin{align*}
    \ppmodel \models G & \iff G \homo \ppmodel; & \ppmodel \models \glang & \iff \exists G \in \glang. G \homo \ppmodel.
\end{align*}
The graph languages above characterize the relational semantics and the equational theory for $\PCoRTC$, as follows
(\Cref{proposition: semantics for graphs} is shown by induction on $\term$ and \Cref{proposition: graph characterization} is shown by using \Cref{proposition: semantics for graphs}; see \ifiscameraready the full version \cite{nakamuraExistentialCalculusRelations2023}\else \Cref{section: proposition: semantics for graphs}\fi, for more details):
\begin{proposition}\label{proposition: semantics for graphs}
    For every structure $\model$ and {\PCoRTC} term $\term$,
    \[\jump{\term}_{\model} = \jump{\glang(\term)}_{\model}.\]
\end{proposition}
\begin{proposition}[\hspace{1sp}{\cite[Thm.\ 3.9]{Brunet2017}}, {\cite[Thm.\ 16]{Pous2018}}]\label{proposition: graph characterization}
    For every $\PCoRTC$ terms $\term[1], \term[2]$, we have
    \[{} \models_{\REL} \term[1] \le \term[2] \iff \forall G \in \glang(\term[1]). \exists H \in \glang(\term[2]).\  H \homo G. \]
\end{proposition}

\begin{example}
    We can prove ${} \models_{\REL} a \cap b \le a \cap (\top b)$ by the homomorphism $\begin{tikzpicture}[baseline = -1.ex]
            \graph[grow right = 1.2cm, branch down = 1.2ex, nodes={mynode}]{
            {0/{}[draw, circle], 2/{}[draw, circle, xshift = 1.5em]}-!-{1/{}[draw, circle]}
            };
            \node[left = .5em of 0](l){};
            \node[right = .5em of 1](r){};
            \graph[use existing nodes, edges={color=black, pos = .5, earrow}, edge quotes={fill=white, inner sep=1.pt,font= \scriptsize}]{
            0 ->["$a$"] 1;
            2 ->["$b$", out = 0, in = -135] 1;
            l -> 0; 1 -> r;
            };
        \end{tikzpicture} \homo \begin{tikzpicture}[baseline = -.5ex]
            \graph[grow right = 1.cm, branch down = 2ex, nodes={mynode}]{
            {0/{}[draw, circle]}-!-{1/{}[draw, circle]}
            };
            \node[left = .5em of 0](l){};
            \node[right = .5em of 1](r){};
            \graph[use existing nodes, edges={color=black, pos = .5, earrow}, edge quotes={fill=white, inner sep=1.pt,font= \scriptsize}]{
            0 ->["$a$", bend left = 20] 1;
            0 ->["$b$", bend right = 20] 1;
            l -> 0; 1 -> r;
            };
        \end{tikzpicture}$. Here,
    \begin{align*}
        \mathcal{G}(a \cap b)        & = \set{\hspace{-.1em}\begin{tikzpicture}[baseline = -.5ex]
                                                                    \graph[grow right = 1.cm, branch down = 2ex, nodes={mynode}]{
                                                                    {0/{}[draw, circle]}-!-{1/{}[draw, circle]}
                                                                    };
                                                                    \node[left = .5em of 0](l){};
                                                                    \node[right = .5em of 1](r){};
                                                                    \graph[use existing nodes, edges={color=black, pos = .5, earrow}, edge quotes={fill=white, inner sep=1.pt,font= \scriptsize}]{
                                                                    0 ->["$a$", bend left = 20] 1;
                                                                    0 ->["$b$", bend right = 20] 1;
                                                                    l -> 0; 1 -> r;
                                                                    };
                                                                \end{tikzpicture}\hspace{-.1em}}; & \hspace{-.8em}
        \mathcal{G}(a \cap (\top b)) & = \set{\hspace{-.1em}\begin{tikzpicture}[baseline = -1.ex]
                                                                    \graph[grow right = 1.2cm, branch down = 1.2ex, nodes={mynode}]{
                                                                    {0/{}[draw, circle], 2/{}[draw, circle, xshift = 1.5em]}-!-{1/{}[draw, circle]}
                                                                    };
                                                                    \node[left = .5em of 0](l){};
                                                                    \node[right = .5em of 1](r){};
                                                                    \graph[use existing nodes, edges={color=black, pos = .5, earrow}, edge quotes={fill=white, inner sep=1.pt,font= \scriptsize}]{
                                                                    0 ->["$a$"] 1;
                                                                    2 ->["$b$", out = 0, in = -135] 1;
                                                                    l -> 0; 1 -> r;
                                                                    };
                                                                \end{tikzpicture}\hspace{-.1em}}.
    \end{align*}
\end{example}
\section{Graph characterization for \ECoRTC}\label{section: graph}
We consider extending \Cref{proposition: semantics for graphs,proposition: graph characterization} for $\ECoRTC$.
We can straightforwardly extend \Cref{proposition: semantics for graphs}.
\begin{proposition}[cf.\ \Cref{proposition: semantics for graphs}]\label{proposition: semantics for graphs ECoR}
    For every structure $\model$ and $\ECoRTC$ term $\term$,
    we have $\jump{t}_{\model} = \jump{\glang(t)}_{\model}$.
\end{proposition}
\begin{proof}
    Similar to \Cref{proposition: semantics for graphs} (\ifiscameraready see the full version \cite{nakamuraExistentialCalculusRelations2023}\else \Cref{section: proof: proposition: semantics for graphs ECoR}\fi).
\end{proof}
However, we cannot extend \Cref{proposition: graph characterization}, immediately.
\begin{example}\label{example: proposition: graph characterization ECoR counter}
    ${} \models_{\REL} \top \le a \cup \comnf{a}$ holds (cf.\ \Cref{equation: I- top 0}), but the right-hand side formula of \Cref{proposition: graph characterization} fails because
    there does not exist any homomorphism from any graphs in $\glang(a \cup \comnf{a})$:
    \begin{align*}
        \glang(a \cup \comnf{a}) & = \quad \set{\begin{tikzpicture}[baseline = -.5ex]
                                                        \graph[grow right = 1cm, branch down = 6ex, nodes={mynode}]{
                                                        {0/{}[draw, circle]}-!-{1/{}[draw, circle]}
                                                        };
                                                        \node[left = .5em of 0](l){};
                                                        \node[right = .5em of 1](r){};
                                                        \path (0) edge [draw = white, opacity = 0] node[pos= 0.5,inner sep = 1.5pt, font = \scriptsize, opacity = 1](a1){$a$}(1);
                                                        \graph[use existing nodes, edges={color=black, pos = .5, earrow}, edge quotes={fill=white, inner sep=1pt,font= \scriptsize}]{
                                                            0 -- a1 -> 1;
                                                            l -> 0; 1 -> r;
                                                        };
                                                    \end{tikzpicture}, \begin{tikzpicture}[baseline = -.5ex]
                                                                           \graph[grow right = 1cm, branch down = 6ex, nodes={mynode}]{
                                                                           {0/{}[draw, circle]}-!-{1/{}[draw, circle]}
                                                                           };
                                                                           \node[left = .5em of 0](l){};
                                                                           \node[right = .5em of 1](r){};
                                                                           \path (0) edge [draw = white, opacity = 0] node[pos= 0.5, inner sep = 1.5pt,font = \scriptsize, opacity = 1](a1){$\comnf{a}$}(1);
                                                                           \graph[use existing nodes, edges={color=black, pos = .5, earrow}, edge quotes={fill=white, inner sep=1pt,font= \scriptsize}]{
                                                                               0 -- a1 -> 1;
                                                                               l -> 0; 1 -> r;
                                                                           };
                                                                       \end{tikzpicture}} \\
        \glang(\top)             & \ni \qquad \begin{tikzpicture}[baseline = -.5ex]
                                                  \graph[grow right = 1cm, branch down = 6ex, nodes={mynode}]{
                                                  {0/{}[draw, circle]}-!-{1/{}[draw, circle]}
                                                  };
                                                  \node[left = .5em of 0](l){};
                                                  \node[right = .5em of 1](r){};
                                                  \graph[use existing nodes, edges={color=black, pos = .5, earrow}, edge quotes={fill=white, inner sep=1pt,font= \scriptsize}]{
                                                      l -> 0; 1 -> r;
                                                  };
                                              \end{tikzpicture}.
    \end{align*}
    (The same problem occurs even without $\top$; consider ${} \models_{\REL} b \cup \comnf{b} \le a \cup \comnf{a}$ where $b \neq a$.)
\end{example}
To avoid the problem above, we consider modifying the graph languages using \emph{edge saturations}.

\newcommand{\eqcl}[1]{#1^{\mathcal{E}}}
\newcommand{\Quo}{\mathcal{Q}}
\newcommand{\QuoFill}{\mathcal{QS}}
\newcommand{\gquo}[2]{[#1]_{#2}}
\subsection{Edge-saturated graphs and 2-pointed structures}\label{section: saturated graphs}
For a binary relation $R$, we write $\eqcl{R}$ for the \emph{equivalence closure} of $R$ (the minimal equivalence relation subsuming $R$).

For a graph $\graph$ over $\SIG_{\const{I}}^{(-)}$, we write $\graph^{\Quo}$ for the quotient graph{\ifiscameraready\else\footnote{Precisely, $\graph^{\Quo}$ is the graph over $\SIG_{\const{I}}^{(-)}$, defined by
                $|\graph^{\Quo}| = \set{\gquo{x}{\graph} \mid x \in \graph}$;
                $a^{\graph^{\Quo}} = \set{\tuple{X, Y} \in |\graph^{\Quo}|^2 \mid \exists x \in X.\ \exists y \in Y.\   \tuple{x, y} \in a^{\graph}}$ for $a \in \SIG_{\const{I}}^{(-)}$;
                $\tuple{\lv^{\graph^{\Quo}}, \rv^{\graph^{\Quo}}} = \tuple{\gquo{\lv^{\graph}}{\graph}, \gquo{\rv^{\graph}}{\graph}}$.}\fi} of $G$ w.r.t.\ the equivalence relation $\eqcl{(\const{I}^{G})}$;
e.g., if $G = \begin{tikzpicture}[baseline = -1.5ex]
        \graph[grow right = 1.8cm, branch down = 2ex, nodes={mysmallnode, font=\scriptsize}]{
        {0/{$0$}[draw, circle], 2/{$2$}[draw, circle, xshift = 2.2em]}-!-{1/{$1$}[draw, circle]}
        };
        \node[left = .5em of 0](l){};
        \node[right = .5em of 1](r){};
        \path (0) edge [draw = white, opacity = 0] node[pos= 0.5,  inner sep = 1.5pt, font = \scriptsize, opacity = 1](a1){$a$}(1);
        \path (2) edge [bend right= 15, draw = white, opacity = 0] node[pos= 0.4,  inner sep = 1.5pt, font = \scriptsize, opacity = 1](a2){$b$}(1);
        \path (0) edge [bend right = 15, draw = white, opacity = 0] node[pos= 0.6,  inner sep = 1.5pt, font = \scriptsize, opacity = 1](a3){$\const{I}$}(2);
        \graph[use existing nodes, edges={color=black, pos = .5, earrow}, edge quotes={fill=white, inner sep=1pt,font= \scriptsize}]{
        {0,2} -- {a1,a2} -> 1;
        {2} --[bend left = 10] {a3} ->[bend left = 10] 0;
        l -> 0; 1 -> r;
        };
    \end{tikzpicture}$, $G^{\Quo} = \begin{tikzpicture}[baseline = -.5ex]
        \graph[grow right = 1.5cm, branch down = 2ex, nodes={mysmallnode, font=\tiny, minimum width=.7cm}]{
        {0/{$\set{0,2}$}[draw, circle]}-!-{1/{$\set{1}$}[draw, circle]}
        };
        \node[left = .5em of 0](l){};
        \node[right = .5em of 1](r){};
        \path (0) edge [draw = white, opacity = 0, bend left = 25] node[pos= 0.5,  inner sep = 1.5pt, font = \scriptsize, opacity = 1](a1){$a$}(1);
        \path (0) edge [bend right = 25, draw = white, opacity = 0] node[pos= 0.5,  inner sep = 1.5pt, font = \scriptsize, opacity = 1](a2){$b$}(1);
        \node[right = .5em of 0, inner sep = 1.5pt, font = \scriptsize, opacity = 1](I00){$\const{I}$};
        \graph[use existing nodes, edges={color=black, pos = .5, earrow}, edge quotes={fill=white, inner sep=1pt,font= \scriptsize}]{
        0 --[bend left = 10] I00 ->[bend left = 10] 0;
        {0} --[bend left = 10] {a1} ->[bend left = 10] 1;
        {0} --[bend right = 10] {a2} ->[bend right = 10] 1;
        l -> 0; 1 -> r;
        };
    \end{tikzpicture}$.
We use $\gquo{x}{\graph}$ to denote the quotient set of a vertex $x \in |\graph|$ w.r.t.\  the equivalence relation $\eqcl{(\const{I}^{\graph})}$.
\begin{definition}\label{definition: consistent}
    Let $\graph$ be a graph over $\SIG_{\const{I}}^{(-)}$.
    We say that $G$ is \emph{consistent} if for every $a \in \SIG_{\const{I}}$,  the following holds:
    \begin{description}
        \item[($a$-consistent)\label{definition: consistent: consistent}] $(\eqcl{(\const{I}^{G})} \cdot a^{G} \cdot \eqcl{(\const{I}^{G})}) \cap (\eqcl{(\const{I}^{G})} \cdot \comnf{a}^{G} \cdot \eqcl{(\const{I}^{G})}) = \emptyset$.
    \end{description}
    We say that $G$ is \emph{(consistently) edge-saturated} if $G$ is consistent and the following hold:
    \begin{description}
        \item[($a$-saturated)\label{definition: consistent: saturated}]  $a^{G} \cup \comnf{a}^{G} = |G|^2$, for every $a \in \SIG_{\const{I}}$;
        \item[(I-equivalence)\label{definition: consistent: I}] $\const{I}^{G}$ is an equivalence relation.
    \end{description}
\end{definition}
Each edge-saturated graph induces a $2$-pointed structure:
\begin{proposition}\label{proposition: structure}
    If a graph $G$ over $\SIG_{\const{I}}^{(-)}$ is edge-saturated, then $G^{\Quo}$ is a $2$-pointed structure.
\end{proposition}
\begin{proof}
    $\const{I}^{G^{\Quo}}$ is the identity relation because $\const{I}^{G}$ is an equivalence relation \ref{definition: consistent: I}.
    For $a \in \SIG_{\const{I}}$, $\bar{a}^{G^{\Quo}} = |G^{\Quo}|^2 \setminus a^{G^{\Quo}}$ holds because
    $a^{G^{\Quo}} \cup \bar{a}^{G^{\Quo}} = |G^{\Quo}|^2$ \ref{definition: consistent: saturated}
    and $a^{G^{\Quo}} \cap \bar{a}^{G^{\Quo}} = \emptyset$ \ref{definition: consistent: consistent}.
    Thus, $G^{\Quo}$ is a $2$-pointed structure.
\end{proof}

\subsection{Graph characterization via edge saturations}\label{section: graph characterization}
\begin{definition}
    For graphs $\graph[1], \graph[2]$ over $\SIG_{\const{I}}^{(-)}$,
    we say that $\graph[2]$ is an \emph{edge-extension} of $\graph[1]$, if $|\graph[2]| = |\graph[1]|$ and $a^{\graph[2]} \supseteq a^{\graph[1]}$ for every $a \in \SIG_{\const{I}}^{(-)}$.
    We say that $\graph[2]$ is an \emph{(edge-)saturation} of $\graph[1]$ if
    $\graph[2]$ is an edge-extension of $\graph[1]$ and is edge-saturated.
\end{definition}
Let $\Fill(\graph[1])$ be the set of all saturations of $\graph[1]$.

For a graph language $\glang$,
let $\Fill(\glang) \defeq \bigcup_{\graph \in \glang} \Fill(\graph)$ and $\Quo(\glang) \defeq \set{\graph^{\Quo} \mid \graph \in \glang}$.
We abbreviate $\Quo \circ \Fill$ to $\QuoFill$.

\begin{example}\label{example: quofill}
    When $\SIG = \set{a}$,
    $\QuoFill(\begin{tikzpicture}[baseline = -.5ex]
                \graph[grow right = .7cm, branch down = 6ex, nodes={mynode}]{
                {0/{}[draw, circle]}-!-{1/{}[draw, circle]}
                };
                \node[left = .5em of 0](l){};
                \node[right = .5em of 1](r){};
                \graph[use existing nodes, edges={color=black, pos = .5, earrow}, edge quotes={fill=white, inner sep=1pt,font= \scriptsize}]{
                    l -> 0; 1 -> r;
                };
            \end{tikzpicture})$ is the set:
    \[\left\{ \hspace{-.1em}\begin{tikzpicture}[remember picture, baseline = -.5ex]
            \graph[grow right = 1.2cm, branch down = 6ex, nodes={mynode}]{
            {C0/{}[draw, circle]}
            };
            \node[left = .5em of C0](l){};
            \node[right = .5em of C0](r){};
            \node[below left = 2ex and .0em of C0, inner sep = 1.5pt, font = \scriptsize, opacity = 1](a00){$\comnf{a}$};
            \node[below right = 2ex and .0em of C0, inner sep = 1.5pt, font = \scriptsize, opacity = 1](I00){$\const{I}$};
            \graph[use existing nodes, edges={color=black, pos = .5, earrow}, edge quotes={fill=white, inner sep=1pt,font= \scriptsize}]{
            C0 --[bend left = 10] a00 ->[bend left = 10] C0;
            C0 --[bend left = 10] I00 ->[bend left = 10] C0;
            l -> C0; C0 -> r;
            };
        \end{tikzpicture},
        \begin{tikzpicture}[remember picture, baseline = -.5ex]
            \graph[grow right = 1.2cm, branch down = 6ex, nodes={mynode}]{
            {C0/{}[draw, circle]}
            };
            \node[left = .5em of C0](l){};
            \node[right = .5em of C0](r){};
            \node[below left = 2ex and .0em of C0, inner sep = 1.5pt, font = \scriptsize, opacity = 1](a00){$a$};
            \node[below right = 2ex and .0em of C0, inner sep = 1.5pt, font = \scriptsize, opacity = 1](I00){$\const{I}$};
            \graph[use existing nodes, edges={color=black, pos = .5, earrow}, edge quotes={fill=white, inner sep=1pt,font= \scriptsize}]{
            C0 --[bend left = 10] a00 ->[bend left = 10] C0;
            C0 --[bend left = 10] I00 ->[bend left = 10] C0;
            l -> C0; C0 -> r;
            };
        \end{tikzpicture},
        \begin{tikzpicture}[remember picture, baseline = -.5ex]
            \graph[grow right = 1.2cm, branch down = 6ex, nodes={mynode}]{
            {D0/{}[draw, circle]}-!-{D1/{}[draw, circle]}
            };
            \node[left = .5em of D0](l){};
            \node[right = .5em of D1](r){};
            \node[below left = 2ex and .0em of D0, inner sep = .5pt, font = \scriptsize, opacity = 1](a00){$\comnf{a}$};
            \node[below left = 2ex and .0em of D1, inner sep = .5pt, font = \scriptsize, opacity = 1](a11){$\comnf{a}$};
            \path (D0) edge [draw = white, opacity = 0, bend left = 18] node[pos= 0.5, inner sep = 1.5pt, font = \scriptsize, opacity = 1](a01){$\comnf{a}$}(D1);
            \path (D1) edge [draw = white, opacity = 0, bend left = 18] node[pos= 0.5, inner sep = .5pt, font = \scriptsize, opacity = 1](a10){$\comnf{a}$}(D0);
            \node[below right = 2ex and .0em of D0, inner sep = 1.5pt, font = \scriptsize, opacity = 1](I00){$\const{I}$};
            \node[below right = 2ex and .0em of D1, inner sep = 1.5pt, font = \scriptsize, opacity = 1](I11){$\const{I}$};
            \path (D0) edge [draw = white, opacity = 0, bend left = 80] node[pos= 0.5, inner sep = 1.pt, font = \scriptsize, opacity = 1](I01){$\comnf{\const{I}}$}(D1);
            \path (D1) edge [draw = white, opacity = 0, bend left = 80] node[pos= 0.5, inner sep = 1.pt, font = \scriptsize, opacity = 1](I10){$\comnf{\const{I}}$}(D0);
            \graph[use existing nodes, edges={color=black, pos = .5, earrow}, edge quotes={fill=white, inner sep=1pt,font= \scriptsize}]{
            D0 --[bend left = 10] a00 ->[bend left = 10] D0;
            D0 --[bend left = 10] a01 ->[bend left = 10] D1;
            D1 --[bend left = 10] a10 ->[bend left = 10] D0;
            D1 --[bend left = 10] a11 ->[bend left = 10] D1;
            D0 --[bend left = 10] I00 ->[bend left = 10] D0;
            D0 --[bend left = 10] I01 ->[bend left = 10] D1;
            D1 --[bend left = 10] I10 ->[bend left = 10] D0;
            D1 --[bend left = 10] I11 ->[bend left = 10] D1;
            l -> D0; D1 -> r;
            };
        \end{tikzpicture}, \dots, \begin{tikzpicture}[remember picture, baseline = -.5ex]
            \graph[grow right = 1.2cm, branch down = 6ex, nodes={mynode}]{
            {E0/{}[draw, circle]}-!-{E1/{}[draw, circle]}
            };
            \node[left = .5em of E0](l){};
            \node[right = .5em of E1](r){};
            \node[below left = 2ex and .0em of E0, inner sep = .5pt, font = \scriptsize, opacity = 1](a00){$\comnf{a}$};
            \node[below left = 2ex and .0em of E1, inner sep = 1.5pt, font = \scriptsize, opacity = 1](a11){$a$};
            \path (E0) edge [draw = white, opacity = 0, bend left = 18] node[pos= 0.5, inner sep = .5pt, font = \scriptsize, opacity = 1](a01){$a$}(E1);
            \path (E1) edge [draw = white, opacity = 0, bend left = 18] node[pos= 0.5, inner sep = 1.5pt, font = \scriptsize, opacity = 1](a10){$\comnf{a}$}(E0);
            \node[below right = 2ex and .0em of E0, inner sep = 1.5pt, font = \scriptsize, opacity = 1](I00){$\const{I}$};
            \node[below right = 2ex and .0em of E1, inner sep = 1.5pt, font = \scriptsize, opacity = 1](I11){$\const{I}$};
            \path (E0) edge [draw = white, opacity = 0, bend left = 80] node[pos= 0.5, inner sep = 1.pt, font = \scriptsize, opacity = 1](I01){$\comnf{\const{I}}$}(E1);
            \path (E1) edge [draw = white, opacity = 0, bend left = 80] node[pos= 0.5, inner sep = 1.pt, font = \scriptsize, opacity = 1](I10){$\comnf{\const{I}}$}(E0);
            \graph[use existing nodes, edges={color=black, pos = .5, earrow}, edge quotes={fill=white, inner sep=1pt,font= \scriptsize}]{
            E0 --[bend left = 10] a00 ->[bend left = 10] E0;
            E0 --[bend left = 10] a01 ->[bend left = 10] E1;
            E1 --[bend left = 10] a10 ->[bend left = 10] E0;
            E1 --[bend left = 10] a11 ->[bend left = 10] E1;
            E0 --[bend left = 10] I00 ->[bend left = 10] E0;
            E0 --[bend left = 10] I01 ->[bend left = 10] E1;
            E1 --[bend left = 10] I10 ->[bend left = 10] E0;
            E1 --[bend left = 10] I11 ->[bend left = 10] E1;
            l -> E0; E1 -> r;
            };
        \end{tikzpicture}, \dots \right\}.\]
    $\QuoFill(\begin{tikzpicture}[baseline = -.5ex]
                \graph[grow right = .7cm, branch down = 6ex, nodes={mynode}]{
                {0/{}[draw, circle]}-!-{1/{}[draw, circle]}
                };
                \node[left = .5em of 0](l){};
                \node[right = .5em of 1](r){};
                \graph[use existing nodes, edges={color=black, pos = .5, earrow}, edge quotes={fill=white, inner sep=1pt,font= \scriptsize}]{
                    l -> 0; 1 -> r;
                };
            \end{tikzpicture})$ has $18$ graphs up to isomorphism (the sum of patterns ($2^{1 \times 1} = 2$) when the two vertices are connected with $\const{I}$ and patterns ($2^{2 \times 2} = 16$) otherwise).
    For every $G$ in $\QuoFill(\begin{tikzpicture}[baseline = -.5ex]
                \graph[grow right = .7cm, branch down = 6ex, nodes={mynode}]{
                {0/{}[draw, circle]}-!-{1/{}[draw, circle]}
                };
                \node[left = .5em of 0](l){};
                \node[right = .5em of 1](r){};
                \graph[use existing nodes, edges={color=black, pos = .5, earrow}, edge quotes={fill=white, inner sep=1pt,font= \scriptsize}]{
                    l -> 0; 1 -> r;
                };
            \end{tikzpicture})$,
    there exists a homomorphism to $G$ from \begin{tikzpicture}[baseline = -.5ex]
        \graph[grow right = 1cm, branch down = 6ex, nodes={mynode}]{
        {0/{}[draw, circle]}-!-{1/{}[draw, circle]}
        };
        \node[left = .5em of 0](l){};
        \node[right = .5em of 1](r){};
        \path (0) edge [draw = white, opacity = 0] node[pos= 0.5,inner sep = 1.5pt, font = \scriptsize, opacity = 1](a1){$a$}(1);
        \graph[use existing nodes, edges={color=black, pos = .5, earrow}, edge quotes={fill=white, inner sep=1pt,font= \scriptsize}]{
            0 -- a1 -> 1;
            l -> 0; 1 -> r;
        };
    \end{tikzpicture} or \begin{tikzpicture}[baseline = -.5ex]
        \graph[grow right = 1cm, branch down = 6ex, nodes={mynode}]{
        {0/{}[draw, circle]}-!-{1/{}[draw, circle]}
        };
        \node[left = .5em of 0](l){};
        \node[right = .5em of 1](r){};
        \path (0) edge [draw = white, opacity = 0] node[pos= 0.5, inner sep = 1.5pt,font = \scriptsize, opacity = 1](a1){$\comnf{a}$}(1);
        \graph[use existing nodes, edges={color=black, pos = .5, earrow}, edge quotes={fill=white, inner sep=1pt,font= \scriptsize}]{
            0 -- a1 -> 1;
            l -> 0; 1 -> r;
        };
    \end{tikzpicture}, as $G$ is edge-saturated:
    \begin{align*}
        \glang(a \cup \comnf{a}) & = \hspace{.2em} \left\
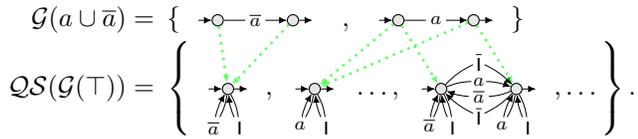
{\begin{tikzpicture}[baseline = -.5ex, remember picture]
                                                              \graph[grow right = 1.cm, branch down = 6ex, nodes={mynode}]{
                                                              {X0/{}[draw, circle]}-!-{X1/{}[draw, circle]}
                                                              };
                                                              \node[left = .5em of X0](l){};
                                                              \node[right = .5em of X1](r){};
                                                              \path (X0) edge [draw = white, opacity = 0] node[pos= 0.5,inner sep = 1.5pt, font = \scriptsize, opacity = 1](a1){$\comnf{a}$}(X1);
                                                              \graph[use existing nodes, edges={color=black, pos = .5, earrow}, edge quotes={fill=white, inner sep=1pt,font= \scriptsize}]{
                                                                  X0 -- a1 -> X1;
                                                                  l -> X0; X1 -> r;
                                                              };
                                                          \end{tikzpicture}\right.
        \hspace{.15em} , \hspace{.15em}
        \left.\begin{tikzpicture}[baseline = -.5ex, remember picture]
                  \graph[grow right = 1.cm, branch down = 6ex, nodes={mynode}]{
                  {Y0/{}[draw, circle]}-!-{Y1/{}[draw, circle]}
                  };
                  \node[left = .5em of Y0](l){};
                  \node[right = .5em of Y1](r){};
                  \path (Y0) edge [draw = white, opacity = 0] node[pos= 0.5, inner sep = 1.5pt,font = \scriptsize, opacity = 1](a1){$a$}(Y1);
                  \graph[use existing nodes, edges={color=black, pos = .5, earrow}, edge quotes={fill=white, inner sep=1pt,font= \scriptsize}]{
                      Y0 -- a1 -> Y1;
                      l -> Y0; Y1 -> r;
                  };
              \end{tikzpicture}\right\}                                                                   \\
        \QuoFill(\glang(\top))   & = \hspace{.2em} \left\{\begin{tikzpicture}[baseline = -.5ex, remember picture]
                                                              \graph[grow right = 1.cm, branch down = 6ex, nodes={mynode}]{
                                                              {C0/{}[draw, circle]}
                                                              };
                                                              \node[left = .5em of C0](l){};
                                                              \node[right = .5em of C0](r){};
                                                              \node[below left = 2ex and .0em of C0, inner sep = 1.5pt, font = \scriptsize, opacity = 1](a00){$\comnf{a}$};
                                                              \node[below right = 2ex and .0em of C0, inner sep = 1.5pt, font = \scriptsize, opacity = 1](I00){$\const{I}$};
                                                              \graph[use existing nodes, edges={color=black, pos = .5, earrow}, edge quotes={fill=white, inner sep=1pt,font= \scriptsize}]{
                                                              C0 --[bend left = 10] a00 ->[bend left = 10] C0;
                                                              C0 --[bend left = 10] I00 ->[bend left = 10] C0;
                                                              l -> C0; C0 -> r;
                                                              };
                                                          \end{tikzpicture},
        \begin{tikzpicture}[baseline = -.5ex, remember picture]
            \graph[grow right = 1.cm, branch down = 6ex, nodes={mynode}]{
            {C20/{}[draw, circle]}
            };
            \node[left = .5em of C20](l){};
            \node[right = .5em of C20](r){};
            \node[below left = 2ex and .0em of C20, inner sep = 1.5pt, font = \scriptsize, opacity = 1](a00){$a$};
            \node[below right = 2ex and .0em of C20, inner sep = 1.5pt, font = \scriptsize, opacity = 1](I00){$\const{I}$};
            \graph[use existing nodes, edges={color=black, pos = .5, earrow}, edge quotes={fill=white, inner sep=1pt,font= \scriptsize}]{
            C20 --[bend left = 10] a00 ->[bend left = 10] C20;
            C20 --[bend left = 10] I00 ->[bend left = 10] C20;
            l -> C20; C20 -> r;
            };
        \end{tikzpicture}
        \dots,
        \begin{tikzpicture}[baseline = -.5ex, remember picture]
            \graph[grow right = 1.cm, branch down = 6ex, nodes={mynode}]{
            {E0/{}[draw, circle]}-!-{E1/{}[draw, circle]}
            };
            \node[left = .5em of E0](l){};
            \node[right = .5em of E1](r){};
            \node[below left = 2ex and .0em of E0, inner sep = .5pt, font = \scriptsize, opacity = 1](a00){$\comnf{a}$};
            \node[below left = 2ex and .0em of E1, inner sep = 1.5pt, font = \scriptsize, opacity = 1](a11){$a$};
            \path (E0) edge [draw = white, opacity = 0, bend left = 18] node[pos= 0.5, inner sep = .5pt, font = \scriptsize, opacity = 1](a01){$a$}(E1);
            \path (E1) edge [draw = white, opacity = 0, bend left = 18] node[pos= 0.5, inner sep = 1.5pt, font = \scriptsize, opacity = 1](a10){$\comnf{a}$}(E0);
            \node[below right = 2ex and .0em of E0, inner sep = 1.5pt, font = \scriptsize, opacity = 1](I00){$\const{I}$};
            \node[below right = 2ex and .0em of E1, inner sep = 1.5pt, font = \scriptsize, opacity = 1](I11){$\const{I}$};
            \path (E0) edge [draw = white, opacity = 0, bend left = 80] node[pos= 0.5, inner sep = 1.pt, font = \scriptsize, opacity = 1](I01){$\comnf{\const{I}}$}(E1);
            \path (E1) edge [draw = white, opacity = 0, bend left = 80] node[pos= 0.5, inner sep = 1.pt, font = \scriptsize, opacity = 1](I10){$\comnf{\const{I}}$}(E0);
            \graph[use existing nodes, edges={color=black, pos = .5, earrow}, edge quotes={fill=white, inner sep=1pt,font= \scriptsize}]{
            E0 --[bend left = 10] a00 ->[bend left = 10] E0;
            E0 --[bend left = 10] a01 ->[bend left = 10] E1;
            E1 --[bend left = 10] a10 ->[bend left = 10] E0;
            E1 --[bend left = 10] a11 ->[bend left = 10] E1;
            E0 --[bend left = 10] I00 ->[bend left = 10] E0;
            E0 --[bend left = 10] I01 ->[bend left = 10] E1;
            E1 --[bend left = 10] I10 ->[bend left = 10] E0;
            E1 --[bend left = 10] I11 ->[bend left = 10] E1;
            l -> E0; E1 -> r;
            };
        \end{tikzpicture}, \dots \right\}. %
        \begin{tikzpicture}[remember picture, overlay]
            \path (X0) edge [homoarrow,->] (C0);
            \path (X1) edge [homoarrow,->] (C0);
            \path (Y0) edge [homoarrow,->] (C20);
            \path (Y1) edge [homoarrow,->] (C20);
            \path (Y0) edge [homoarrow,->] (E0);
            \path (Y1) edge [homoarrow,->] (E1);
        \end{tikzpicture}
    \end{align*}
\end{example}
As above, by using $\QuoFill(\glang(\top))$ instead of $\glang(\top)$, we can avoid the problem in \Cref{example: proposition: graph characterization ECoR counter}.
Using $\QuoFill$, we can strengthen \Cref{proposition: graph characterization} (\hspace{1sp}{\cite[Thm.\ 3.9]{Brunet2017}}, {\cite[Thm.\ 16]{Pous2018}}), for $\ECoRTC$, as \Cref{theorem: graph characterization inequation}.
We first show the following:
\begin{lemma}\label{lemma: graph characterization ECoR QuoFill}
    For any graph $\graph$ over $\SIG_{\const{I}}^{(-)}$,
    $\jump{\graph}_{\model} = \jump{\QuoFill(\graph)}_{\model}$.
\end{lemma}
\begin{proof}
    We prove that for any $\graph$, both $\jump{\graph}_{\model} = \jump{\Fill(\graph)}_{\model}$ and $\jump{\graph}_{\model} = \jump{\graph^{\Quo}}_{\model}$, respectively.

    For $\jump{\graph}_{\model} = \jump{\Fill(\graph)}_{\model}$:
    ($\supseteq$):
    If $\graph[2] \in \Fill(\graph)$ and $h \colon \graph[2] \homo {\model}[x_0, y_0]$,
    then $h \colon \graph \homo {\model}[x_0, y_0]$ because $\graph[2]$ is a saturation of $\graph$.
    ($\subseteq$):
    If $h \colon \graph \homo {\model}[x_0, y_0]$,
    let $\graph[2]$ be the saturation of $\graph[1]$ s.t.\ $a^{H} = \set{\tuple{x,y} \mid \tuple{h(x), h(y)} \in a^{\model}}$ for $a \in \SIG_{\const{I}}^{(-)}$.
    Then $h \colon \graph[2] \homo {\model}[x_0,y_0]$.
    Here, $\graph[2]$ is an edge-extension of $\graph[1]$ because $\tuple{x, y} \in a^{\graph[1]} \Longrightarrow \tuple{h(x), h(y)} \in a^{\model} \Longrightarrow \tuple{x, y} \in a^{\graph[2]}$;
    $\graph[2]$ is edge-saturated because $\model$ is edge-saturated; thus, $\graph[2]$ is indeed a saturation of $\graph[1]$.
    For example, from the homomorphism for $\graph[1] \homo {\model}[x_0, y_0]$, by filling non-existing edges in $\graph[1]$ so that they map to edges of ${\model}$,
    we can construct the saturation $\graph[2]$ of $\graph[1]$ s.t.\ $\graph[2] \homo {\model}[x_0, y_0]$, as follows:
    \begin{align*}
         & \graph[1] = 
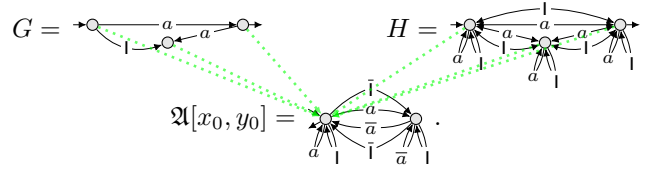
\hspace{-.5em} \begin{tikzpicture}[remember picture, baseline = -1.ex]
                                          \graph[grow right = 1cm, branch down = 1.5ex, nodes={mynode}]{
                                          {1X0/{}[draw, circle]}-!-{/,1X2/{}[draw, circle]}-!-{1X1/{}[draw, circle]}
                                          };
                                          \node[left = .5em of 1X0](1Xl){};
                                          \node[right = .5em of 1X1](1Xr){};
                                          \graph[use existing nodes, edges={color=black, pos = .5, earrow}, edge quotes={fill=white, inner sep=1.pt,font= \scriptsize}]{
                                          1X0 ->["$a$"] 1X1;
                                          1X1 ->["$a$"] 1X2;
                                          1X0 ->["$\const{I}$", bend right = 35] 1X2;
                                          1Xl -> 1X0; 1X1 -> 1Xr;
                                          };
                                      \end{tikzpicture} \hspace{4em} \graph[2] =\hspace{-.5em} \begin{tikzpicture}[remember picture, baseline = -1.ex]
                                                                                                   \graph[grow right = 1cm, branch down = 1.5ex, nodes={mynode}]{
                                                                                                   {2X0/{}[draw, circle]}-!-{/,2X2/{}[draw, circle]}-!-{2X1/{}[draw, circle]}
                                                                                                   };
                                                                                                   \node[left = .5em of 2X0](2Xl){};
                                                                                                   \node[right = .5em of 2X1](2Xr){};
                                                                                                   \node[below right = 2ex and .0em of 2X0, inner sep = 1.5pt, font = \scriptsize, opacity = 1](2XI00){$\const{I}$};
                                                                                                   \node[below left = 2ex and .0em of 2X0, inner sep = .5pt, font = \scriptsize, opacity = 1](2Xa00){$a$};
                                                                                                   \node[below right = 2ex and .0em of 2X1, inner sep = 1.5pt, font = \scriptsize, opacity = 1](2XI11){$\const{I}$};
                                                                                                   \node[below left = 2ex and .0em of 2X1, inner sep = .5pt, font = \scriptsize, opacity = 1](2Xa11){$a$};
                                                                                                   \node[below right = 2ex and .0em of 2X2, inner sep = 1.5pt, font = \scriptsize, opacity = 1](2XI22){$\const{I}$};
                                                                                                   \node[below left = 2ex and .0em of 2X2, inner sep = .5pt, font = \scriptsize, opacity = 1](2Xa22){$a$};
                                                                                                   \graph[use existing nodes, edges={color=black, pos = .5, earrow}, edge quotes={fill=white, inner sep=1.pt,font= \scriptsize}]{
                                                                                                   2X0 <->["$a$"] 2X1;
                                                                                                   2X1 <->["$a$"] 2X2;
                                                                                                   2X0 <->["$a$"] 2X2;
                                                                                                   2X0 <->["$\const{I}$", bend left = 20] 2X1;
                                                                                                   2X2 <->["$\const{I}$", bend right = 35] 2X1;
                                                                                                   2X0 <->["$\const{I}$", bend right = 35] 2X2;
                                                                                                   2X0 --[bend left = 10] 2Xa00 ->[bend left = 10] 2X0;
                                                                                                   2X0 --[bend left = 10] 2XI00 ->[bend left = 10] 2X0;
                                                                                                   2X1 --[bend left = 10] 2Xa11 ->[bend left = 10] 2X1;
                                                                                                   2X1 --[bend left = 10] 2XI11 ->[bend left = 10] 2X1;
                                                                                                   2X2 --[bend left = 10] 2Xa22 ->[bend left = 10] 2X2;
                                                                                                   2X2 --[bend left = 10] 2XI22 ->[bend left = 10] 2X2;
                                                                                                   2Xl -> 2X0; 2X1 -> 2Xr;
                                                                                                   };
                                                                                               \end{tikzpicture}     \\[-2ex]
         & \hspace{6em} {\model}[x_0, y_0] = \hspace{-.5em}\begin{tikzpicture}[remember picture, baseline = -.5ex]
                                                               \graph[grow right = 1.2cm, branch down = 6ex, nodes={mynode, inner sep = 1.5pt}]{
                                                               {4X0/{}[draw, circle]}-!-{4X1/{}[draw, circle]}
                                                               };
                                                               \node[above left = .1ex and .5em of 4X0](4Xl){};
                                                               \node[below left = .1ex and .5em of 4X0](4Xr){};
                                                               \node[below left = 2ex and .0em of 4X0, inner sep = .5pt, font = \scriptsize, opacity = 1](4Xa00){$a$};
                                                               \node[below left = 2ex and .0em of 4X1, inner sep = 1.5pt, font = \scriptsize, opacity = 1](4Xa11){$\comnf{a}$};
                                                               \path (4X0) edge [draw = white, opacity = 0, bend left = 18] node[pos= 0.5, inner sep = .5pt, font = \scriptsize, opacity = 1](4Xa01){$a$}(4X1);
                                                               \path (4X1) edge [draw = white, opacity = 0, bend left = 18] node[pos= 0.5, inner sep = 1.5pt, font = \scriptsize, opacity = 1](4Xa10){$\comnf{a}$}(4X0);
                                                               \node[below right = 2ex and .0em of 4X0, inner sep = 1.5pt, font = \scriptsize, opacity = 1](4XI00){$\const{I}$};
                                                               \node[below right = 2ex and .0em of 4X1, inner sep = 1.5pt, font = \scriptsize, opacity = 1](4XI11){$\const{I}$};
                                                               \path (4X0) edge [draw = white, opacity = 0, bend left = 80] node[pos= 0.5, inner sep = 1.pt, font = \scriptsize, opacity = 1](4XI01){$\comnf{\const{I}}$}(4X1);
                                                               \path (4X1) edge [draw = white, opacity = 0, bend left = 80] node[pos= 0.5, inner sep = 1.pt, font = \scriptsize, opacity = 1](4XI10){$\comnf{\const{I}}$}(4X0);
                                                               \graph[use existing nodes, edges={color=black, pos = .5, earrow}, edge quotes={fill=white, inner sep=1pt,font= \scriptsize}]{
                                                               4X0 --[bend left = 10] 4Xa00 ->[bend left = 10] 4X0;
                                                               4X0 --[bend left = 10] 4Xa01 ->[bend left = 10] 4X1;
                                                               4X1 --[bend left = 10] 4Xa10 ->[bend left = 10] 4X0;
                                                               4X1 --[bend left = 10] 4Xa11 ->[bend left = 10] 4X1;
                                                               4X0 --[bend left = 10] 4XI00 ->[bend left = 10] 4X0;
                                                               4X0 --[bend left = 10] 4XI01 ->[bend left = 10] 4X1;
                                                               4X1 --[bend left = 10] 4XI10 ->[bend left = 10] 4X0;
                                                               4X1 --[bend left = 10] 4XI11 ->[bend left = 10] 4X1;
                                                               4Xl -> 4X0; 4X0 -> 4Xr;
                                                               };
                                                           \end{tikzpicture}.
        \begin{tikzpicture}[remember picture, overlay]
            \path (1X0) edge [homoarrow,->] (4X0);
            \path (1X1) edge [homoarrow,->] (4X0);
            \path (1X2) edge [homoarrow,->] (4X0);
            \path (2X0) edge [homoarrow,->] (4X0);
            \path (2X1) edge [homoarrow,->] (4X0);
            \path (2X2) edge [homoarrow,->] (4X0);
        \end{tikzpicture}
    \end{align*}

    For $\jump{\graph}_{\model} = \jump{\graph^{\Quo}}_{\model}$:
    ($\supseteq$):
    If $\graph[1]^{\Quo} \homo {\model}[x_0, y_0]$,
    then $\graph[1] \homo {\model}[x_0,y_0]$ because $\graph[1] \homo \graph[1]^{\Quo}$ by the quotient map (with transitivity of $\homo$).
    ($\subseteq$):
    If $h \colon \graph \homo {\model}[x_0, y_0]$,
    then $h \circ g \colon \graph[1]^{\Quo} \homo {\model}[x_0, y_0]$, where $g$ is a section of the quotient map from $\graph[1]$ to $\graph[1]^{\Quo}$ (a map $g$ s.t.\ $g(X) \in X$ for $X \in |\graph[1]^{\Quo}|$).
    Hence, $\jump{\graph}_{\model} = \jump{\Fill(\graph)}_{\model} = \jump{\QuoFill(\graph)}_{\model}$.
\end{proof}

\begin{theorem}[cf.\ {\Cref{proposition: graph characterization}}]\label{theorem: graph characterization inequation}
    For every $\ECoRTC$ terms $\term[1], \term[2]$,
    \begin{equation*}
        {} \models_{\REL} \term[1] \le \term[2] \iff \forall G \in \QuoFill(\glang(\term[1])). \exists H \in \glang(\term[2]).\ H \homo G.
    \end{equation*}
\end{theorem}
\begin{proof}
    By the following formula transformation:
    \begin{align*}
                           & {} \models_{\REL} \term[1] \le \term[2]  \iff \forall \model.\  \jump{\term[1]}_{\model} \subseteq \jump{\term[2]}_{\model} \tag{Def.\ of $\models_{\REL}$}                                                 \\
                           & \iff \forall \model.\  \jump{\QuoFill(\glang(\term[1]))}_{\model} \subseteq \jump{\glang(\term[2])}_{\model} \tag{\Cref{proposition: semantics for graphs ECoR,lemma: graph characterization ECoR QuoFill}} \\
                           & \iff \forall \model. \forall G \in \QuoFill(\glang(\term[1])).\ \jump{G}_{\model} \subseteq \jump{\glang(\term[2])}_{\model}  \tag{Def.\ of $\jump{}$}                                                      \\
                           & \iff \forall G \in \QuoFill(\glang(\term[1])).\forall \model. \  \jump{G}_{\model} \subseteq \jump{\glang(\term[2])}_{\model}                                                                               \\
                           & \iff \forall G \in \QuoFill(\glang(\term[1])).\forall \ppmodel.                                                                                                                                             \\
                           & \quad (G \homo \ppmodel) \mbox{ implies } (\exists H \in \glang(\term[2]).\  H \homo \ppmodel) \tag{Def.\ of $\jump{}$}                                                                                     \\
        \label{tag: heart} & \iff \forall G \in \QuoFill(\glang(\term[1])). \exists H \in \glang(\term[2]).\ H \homo G. \tag{$\heartsuit$}
    \end{align*}
    Here, for (\ref{tag: heart}),
    $\Longleftarrow$:
    Let $H \in \glang(\term[2])$ be such that $H \homo G$.
    Then for any $2$-pointed structure $\ppmodel$ s.t.\ $G \homo \ppmodel$,
    we have $H \homo \ppmodel$ by transitivity of $\homo$.
    $\Longrightarrow$:
    By letting $\ppmodel = G$.
    Note that $G$ is a $2$-pointed structure since $G \in \QuoFill(\glang(\term[1]))$ (\Cref{proposition: structure}).
\end{proof}

\begin{remark}
    We should use $\QuoFill$, rather than $\Fill$.
    Consider $\models_{\REL} \top \le \const{I} \cup \comnf{\const{I}}$.
    Then there does not exist any homomorphism
    from any graphs in $\glang(\const{I} \cup \comnf{\const{I}}) =  \set{\begin{tikzpicture}[baseline = -.5ex]
                \graph[grow right = 1cm, branch down = 6ex, nodes={mynode}]{
                {0/{}[draw, circle]}
                };
                \node[left = .5em of 0](l){};
                \node[right = .5em of 0](r){};
                \graph[use existing nodes, edges={color=black, pos = .5, earrow}, edge quotes={fill=white, inner sep=1pt,font= \scriptsize}]{
                    l -> 0; 0 -> r;
                };
            \end{tikzpicture}, \begin{tikzpicture}[baseline = -.5ex]
                \graph[grow right = 1cm, branch down = 6ex, nodes={mynode}]{
                {0/{}[draw, circle]}-!-{1/{}[draw, circle]}
                };
                \node[left = .5em of 0](l){};
                \node[right = .5em of 1](r){};
                \path (0) edge [draw = white, opacity = 0] node[pos= 0.5, inner sep = 1.5pt,font = \scriptsize, opacity = 1](a1){$\comnf{\const{I}}$}(1);
                \graph[use existing nodes, edges={color=black, pos = .5, earrow}, edge quotes={fill=white, inner sep=1pt,font= \scriptsize}]{
                    0 -- a1 -> 1;
                    l -> 0; 1 -> r;
                };
            \end{tikzpicture}}$ to the graph $
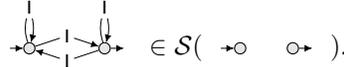
\begin{tikzpicture}[baseline = -.5ex]
            \graph[grow right = 1cm, branch down = 6ex, nodes={mynode}]{
            {0/{}[draw, circle]}-!-{1/{}[draw, circle]}
            };
            \node[left = .5em of 0](l){};
            \node[right = .5em of 1](r){};
            \path (0) edge [draw = white, opacity = 0, bend left] node[pos= 0.5,inner sep = 1.5pt, font = \scriptsize, opacity = 1](a01){$\const{I}$}(1);
            \path (1) edge [draw = white, opacity = 0, bend left] node[pos= 0.5,inner sep = 1.5pt, font = \scriptsize, opacity = 1](a10){$\const{I}$}(0);
            \node[above = 2ex of 0,inner sep = 1.5pt, font = \scriptsize,](a00){$\const{I}$};
            \node[above = 2ex of 1,inner sep = 1.5pt, font = \scriptsize,](a11){$\const{I}$};
            \graph[use existing nodes, edges={color=black, pos = .5, earrow}, edge quotes={fill=white, inner sep=1pt,font= \scriptsize}]{
            0 --[bend left = 15] a00 ->[bend left = 15] 0;
            0 -- a01 -> 1;
            1 -- a10 -> 0;
            1 --[bend left = 15] a11 ->[bend left = 15] 1;
            l -> 0; 1 -> r;
            };
        \end{tikzpicture} \in \Fill(\begin{tikzpicture}[baseline = -.5ex]
                \graph[grow right = .7cm, branch down = 6ex, nodes={mynode}]{
                {0/{}[draw, circle]}-!-{1/{}[draw, circle]}
                };
                \node[left = .5em of 0](l){};
                \node[right = .5em of 1](r){};
                \graph[use existing nodes, edges={color=black, pos = .5, earrow}, edge quotes={fill=white, inner sep=1pt,font= \scriptsize}]{
                    l -> 0; 1 -> r;
                };
            \end{tikzpicture})$.
\end{remark}

\subsection{Bounded model property}
\Cref{theorem: graph characterization inequation} gives an upper bound for the equational theories of existential calculi of relations.
Note that the model checking problem of $\ECoRTC$ is decidable in polynomial time.
\begin{proposition}\label{proposition: model checking}
    The following problem is decidable in $\mathcal{O}(\len{t} \times \card(|\ppstruc|)^{\omega})$ time:
    given a finite $2$-pointed structure $\ppstruc$ and an $\ECoRTC$ term $\term[1]$, does $\ppstruc \models \term[1]$ hold?
    Here, $\omega$ is the matrix multiplication exponent.
\end{proposition}
\begin{proof}[Proof Sketch]
    Let $\ppmodel = {\model}[x, y]$.
    For $\term[2] \in \mathrm{Sub}(\term)$,
    let $f_{\term[2]} \colon |\model| \times |\model| \to \set{\const{true}, \const{false}}$ be such that
    $f_{\term[2]}(x, y) = \const{true}$ iff $\tuple{x, y} \in \jump{\term[2]}_{\model}$.
    Here, $\mathrm{Sub}(\term[1])$ denotes the set of all sub-terms of $\term[1]$.
    Then the function tables of $\set{f_{\term[2]}}_{\term[2] \in \mathrm{Sub}(\term[1])}$ can be calculated by a simple dynamic programming on $\term[2]$,
    where for the case $\term[2] = \term[3]^*$,
    we use the algorithm for the transitive closure of boolean matrix \cite{fischerBooleanMatrixMultiplication1971}.
\end{proof}
\begin{lemma}[bounded model property]\label{lemma: bounded model property for ECoRTC}
    For every $\ECoRTC$ terms $\term[1]$, $\term[2]$,
    we have:
    ${} \not\models_{\REL} \term[1] \le \term[2]$ $\iff$
    there exists a $2$-pointed structure $\ppmodel \in \QuoFill(\glang(\term[1]))$ such that $\ppmodel \not\models \term[1] \le \term[2]$.
\end{lemma}
\begin{proof}
    $\Longleftarrow$: Trivial.
    $\Longrightarrow$:
    Let $\ppmodel = {\model}[x, y] \in \QuoFill(\glang(\term[1]))$ be a $2$-pointed structure such that $\forall H \in \glang(\term[2])$, $H \not\homo \ppmodel$ (\Cref{theorem: graph characterization inequation}).
    Then $\tuple{x, y} \in \jump{\QuoFill(\glang(\term[1]))}_{\model}$ because $\QuoFill(\glang(\term[1])) \ni \ppmodel \homo {\model}[x, y]$ ($\homo$ is reflexive).
    $\tuple{x,y} \not\in \jump{\glang(\term[2])}_{\model}$ because $\forall H \in \glang(\term[2]).\ H \not\homo \ppmodel$ (Def.\ of $\jump{}_{\model}$).
    Thus $\tuple{x, y} \in \jump{\term[1]}_{\model}$ and $\tuple{x, y} \not\in \jump{\term[2]}_{\model}$ (\Cref{proposition: semantics for graphs ECoR,lemma: graph characterization ECoR QuoFill}).
    Hence $\ppstruc \not\models \term[1] \le \term[2]$.
\end{proof}
\begin{lemma}\label{lemma: PCoRTC upper bound}
    The equational theory of $\ECoRTC$ is in $\Pi_{1}^{0}$.
\end{lemma}
\begin{proof}
    Since $\not\models_{\REL} \term[1] = \term[2]$ if and only if $(\not\models_{\REL} \term[1] \le \term[2]) \lor (\not\models_{\REL} \term[2] \le \term[1])$,
    it suffices to show that the following problem is in $\Sigma_{1}^{0}$:
    given $\ECoRTC$ terms $\term[1], \term[2]$, does $\not\models_{\REL} \term[1] \le \term[2]$ hold?
    This follows from \Cref{lemma: bounded model property for ECoRTC} with \Cref{proposition: model checking}.
    Note that for every $G \in \QuoFill(\glang(\term[1]))$, $\#(|G|)$ is always finite,
    and that we can easily enumerate the graphs in $\QuoFill(\glang(\term[1]))$.
\end{proof}
Particularly for $\ECoR$ (not $\ECoRTC$), graphs of each term $\term[1]$ have a linear number of vertices in the size $\len{\term[1]}$.
\begin{proposition}\label{proposition: graph upper bound}
    For every $\ECoR$ term $\term[1]$ and $G \in \QuoFill(\glang(\term[1]))$, we have $\card(|G|) \le 1 + \len{\term[1]}$.
\end{proposition}
\begin{proof}[Proof Sketch]
    By easy induction on $\term[1]$, we have: for every $H \in \glang(\term[1])$, $\card(|H|) \le 1 + \len{\term[1]}$.
    Also $\#(|G|) \le \#(|H|)$ is clear for every $G \in \QuoFill(H)$.
\end{proof}
\begin{theorem}\label{theorem: PCoR complexity}
    The equational theory of $\ECoR$ is coNP-complete.
\end{theorem}
\begin{proof}
    For hardness: By \cref{proposition: PCoR hardness 1}, as $\ECoR$ subsumes $\Termset_{\set{\cdot, \cup}}$.
    For upper bound:
    Similarly for \cref{lemma: PCoRTC upper bound},
    we show that the following problem is in NP: given $\ECoR$ terms $\term[1], \term[2]$, does $\not\models_{\REL} \term[1] \le \term[2]$ hold?
    By \Cref{lemma: bounded model property for ECoRTC}, we can give the following algorithm:
    \begin{enumerate}
        \item Take a graph $H \in \glang(\term[1])$ non-deterministically according to the definition of $\glang$;
              then take a graph $G \in \Fill(H)$, non-deterministically ($G$ is a graph in $\Fill(\glang(\term[1]))$).
        \item Return $\const{true}$ if $G^{\Quo} \not\models \term[1] \le \term[2]$; $\const{false}$ otherwise.
    \end{enumerate}
    Then $\not\models_{\REL} \term[1] \le \term[2]$ if some execution returns $\const{true}$; $\models_{\REL}  \term[1] \le \term[2]$ otherwise.
    Here, $G^{\Quo} \not\models \term[1] \le \term[2]$ can be decided in polynomial time by
    \Cref{proposition: model checking} with $\#(|G^{\Quo}|) \le 1 + \len{\term[1]}$ (\Cref{proposition: graph upper bound}).
\end{proof}
\newcommand{\fl}{\mathrm{cl}}
\NewDocumentCommand\automaton{O{1}}{%
    \ifcase#1
        undefined
    \or \mathcal{A}
    \or \mathcal{B}
    \else undefined

    \fi
}
\newcommand{\empt}{\const{I}}
\section{saturable paths: saturations from a path graph for intersection-free fragments}\label{section: automata}
In this section, we refine the graph characterization of edge-saturations (in the previous section) for $\ECoRTC$ without intersection ($\ExKA$, for short; \Cref{section: ECoR def}) by using \emph{saturable paths}.
Using this characterization, we can show the decidability of the equational theory (\Cref{theorem: Hintikka}) and give an automata construction for two smaller fragments (\Cref{theorem: automata KACC without I-,theorem: automata KACC without a-}).

\newcommand{\AutSIG}{\SIG_{\const{I}}^{(-,\smile)}}
\newcommand{\WordSIG}{\AutSIG \setminus \set{\empt}}
\subsection{NFAs as terms}\label{section: NFA}
A \emph{non-deterministic finite automaton with epsilon translations $\empt$} (NFA, for short) $\automaton = \tuple{|\automaton|, \set{a^{\automaton}}_{a \in A}, \lv^{\automaton}, \rv^{\automaton}}$ over a set $A$ (with $A$ containing a designated element $\empt$) is a graph over $A$.
The \emph{transition relation} $\delta_{w}^{\automaton}$ of a word $w = a_1 \dots a_n \in (A \setminus \set{\empt})^*$ is defined by:
\[\delta_{w}^{\automaton} \defeq (\empt^{\automaton})^* \cdot a_1^{\automaton} \cdot (\empt^{\automaton})^* \cdot \ldots \cdot (\empt^{\automaton})^* \cdot a_n^{\automaton} \cdot (\empt^{\automaton})^*.\]
For notational simplicity, for $x \in |\automaton|$ and $X \subseteq |\automaton|$,
let $\delta_{w}^{\automaton}(x) \defeq \set{y \mid \tuple{x, y} \in \delta_{w}^{\automaton}}$ and $\delta_{w}^{\automaton}(X) \defeq \bigcup_{x \in X} \delta_{w}^{\automaton}(x)$.
The \emph{language} $\lang{\automaton}$ of $\automaton$ is defined by:
\[\lang{\automaton} \defeq \set{w \in (A \setminus \set{\empt})^* \mid \tgt^{\automaton} \in \delta_{w}^{\automaton}(\src^{\automaton})}.\]

Let $\SIG_{\const{I}}^{(-, \smile)} \defeq \set{a, a^{\smile} \mid a \in \SIG^{(-)}} \cup \set{\const{I}, \const{I}^{-}}$.
Similarly for $\comnf{\term}$,
for each $\term \in \SIG_{\const{I}}^{(-, \smile)} \cup \set{\bot, \top}$,
$\widebreve{\term}$ denotes the following term, where $a \in \SIG^{(-)}$:
\begin{align*}
    \widebreve{a} & \defeq a^{\smile} & \widebreve{a^{\smile}} & \defeq a & \widebreve{\const{I}} & \defeq \const{I} & \widebreve{\const{I}^{-}} & \defeq \const{I}^{-} & \widebreve{\top} & \defeq \top & \widebreve{\bot} & \defeq \bot.
\end{align*}
In the following, we always consider NFAs over the set $\AutSIG$ (where $\const{I}$ is used as epsilon transitions).

For an NFA $\automaton$ over $\AutSIG$,
the \emph{binary relation} $\jump{\automaton}_{\struc} \subseteq |\automaton|^2$ is defined by:
\[\jump{\automaton}_{\struc} \defeq \bigcup_{w \in \lang{\automaton}} \jump{w}_{\struc}.\]

Naturally, we can give a construction from $\ExKA$ terms to NFAs using Thompson's construction \cite{Thompson68}, as follows:
\begin{definition}\label{definition: NFA def}
    The \emph{NFA} $\automaton_t$ of an $\ExKA$ term $\term$ is defined by:
    \begin{align*}
        \automaton_{a}
         & \defeq \begin{tikzpicture}[baseline = -.5ex]
                      \graph[grow right = 1.cm, branch down = 2.ex, nodes={mynode, font = \scriptsize}]{
                      {s1/{}[draw, circle]}
                      -!- {sm/{}[draw, circle]}
                      };
                      \node[left = 4pt of s1](s1l){} edge[earrow, ->] (s1);
                      \node[right = 4pt of sm](sml){}; \path (sm) edge[earrow, ->] (sml);
                      \path (s1) edge [draw = white, opacity = 0] node[pos= 0.5, inner sep = 1.pt, font = \scriptsize, opacity = 1](tau1){$a$}(sm);
                      \graph[use existing nodes, edges={color=black, pos = .5, earrow}, edge quotes={fill=white, inner sep=1pt,font= \scriptsize}]{
                          s1 -- tau1 -> sm;
                      };
                  \end{tikzpicture}        \mbox{($a \in \WordSIG$)}                       &
        \automaton_{\const{I}}
         & \defeq \begin{tikzpicture}[baseline = -.5ex]
                      \graph[grow right = 1.cm, branch down = 2.ex, nodes={mynode, font = \scriptsize}]{
                      {s1/{}[draw, circle]}
                      -!- {sm/{}[draw, circle]}
                      };
                      \node[left = 4pt of s1](s1l){} edge[earrow, ->] (s1);
                      \node[right = 4pt of sm](sml){}; \path (sm) edge[earrow, ->] (sml);
                      \path (s1) edge [draw = white, opacity = 0] node[pos= 0.5, inner sep = 1.pt, font = \scriptsize, opacity = 1](tau1){$\empt$}(sm);
                      \graph[use existing nodes, edges={color=black, pos = .5, earrow}, edge quotes={fill=white, inner sep=1pt,font= \scriptsize}]{
                          s1 -- tau1 -> sm;
                      };
                  \end{tikzpicture}                   \\
        \automaton_{\term[1] \cup \term[2]}
         & \defeq \begin{tikzpicture}[baseline = -.5ex]
                      \graph[grow right = 1.cm, branch down = 2.5ex, nodes={mynode, font = \scriptsize}]{
                      {s1/{}[draw, circle]}
                      -!- {s2/{}[draw, circle, above = .5ex], s3/{}[draw, circle, above = .5ex]}
                      -!- {t2/{}[draw, circle, above = .5ex], t3/{}[draw, circle, above = .5ex]}
                      -!- {t1/{}[draw, circle]}
                      };
                      \node[left = 4pt of s1](s1l){} edge[earrow, ->] (s1);
                      \node[right = 4pt of t1](t1l){}; \path (t1) edge[earrow, ->] (t1l);
                      \path (s1) edge [draw = white, opacity = 0, bend left = 10] node[pos= 0.5, inner sep = 1.pt, font = \scriptsize, opacity = 1](tau12){$\empt$}(s2);
                      \path (s1) edge [draw = white, opacity = 0, bend right = 10] node[pos= 0.5, inner sep = 1.pt, font = \scriptsize, opacity = 1](tau13){$\empt$}(s3);
                      \path (t2) edge [draw = white, opacity = 0, bend left = 10] node[pos= 0.5, inner sep = 1.pt, font = \scriptsize, opacity = 1](tau21){$\empt$}(t1);
                      \path (t3) edge [draw = white, opacity = 0, bend right = 10] node[pos= 0.5, inner sep = 1.pt, font = \scriptsize, opacity = 1](tau31){$\empt$}(t1);
                      \graph[use existing nodes, edges={color=black, pos = .5, earrow}, edge quotes={fill=white, inner sep=1pt,font= \scriptsize}]{
                      s1 -- {tau12, tau13} -> {s2, s3}; {t2, t3} -- {tau21, tau31} -> t1;
                      };
                      \graph[use existing nodes, edges={color=black, pos = .5, earrow}, edge quotes={fill=white, inner sep=1pt,font= \scriptsize}]{
                      s2 ->["$\automaton_{t}$"] t2;
                      s3 ->["$\automaton_{s}$"] t3;
                      };
                  \end{tikzpicture} &
        \automaton_{\bot}
         & \defeq \begin{tikzpicture}[baseline = -.5ex]
                      \graph[grow right = 1.cm, branch down = 2.ex, nodes={mynode, font = \scriptsize}]{
                      {s1/{}[draw, circle]}
                      -!- {sm/{}[draw, circle]}
                      };
                      \node[left = 4pt of s1](s1l){} edge[earrow, ->] (s1);
                      \node[right = 4pt of sm](sml){}; \path (sm) edge[earrow, ->] (sml);
                      \graph[use existing nodes, edges={color=black, pos = .5, earrow}, edge quotes={fill=white, inner sep=1pt,font= \scriptsize}]{
                      };
                  \end{tikzpicture}                                                       \\
        \automaton_{\term[1] \cdot \term[2]}
         & \defeq \begin{tikzpicture}[baseline = -.5ex]
                      \graph[grow right = 1.cm, branch down = 2.5ex, nodes={mynode, font = \scriptsize}]{
                      {s1/{}[draw, circle]}
                      -!- {c1/{}[draw, circle]}
                      -!- {c2/{}[draw, circle]}
                      -!- {t1/{}[draw, circle]}
                      };
                      \node[left = 4pt of s1](s1l){} edge[earrow, ->] (s1);
                      \node[right = 4pt of t1](t1l){}; \path (t1) edge[earrow, ->] (t1l);
                      \path (c1) edge [draw = white, opacity = 0] node[pos= 0.5, inner sep = 1.pt, font = \scriptsize, opacity = 1](tau1){$\empt$}(c2);
                      \graph[use existing nodes, edges={color=black, pos = .5, earrow}, edge quotes={fill=white, inner sep=1pt,font= \scriptsize}]{
                      c1 -- {tau1} -> c2;
                      };
                      \graph[use existing nodes, edges={color=black, pos = .5, earrow}, edge quotes={fill=white, inner sep=1pt,font= \scriptsize}]{
                      s1 ->["$\automaton_{t}$"] c1;
                      c2 ->["$\automaton_{s}$"] t1;
                      };
                  \end{tikzpicture}                   & \automaton_{\top}
         & \defeq \automaton_{a \cup \comnf{a}}                                                                                                                      \\
        \automaton_{\term[1]^{*}}
         & \defeq \begin{tikzpicture}[baseline = -.5ex]
                      \graph[grow right = 1.cm, branch down = 2.5ex, nodes={mynode, font = \scriptsize}]{
                      {s1/{}[draw, circle]}
                      -!- {c1/{}[draw, circle]}
                      -!- {c2/{}[draw, circle]}
                      -!- {t1/{}[draw, circle]}
                      };
                      \node[left = 4pt of s1](s1l){} edge[earrow, ->] (s1);
                      \node[right = 4pt of t1](t1l){}; \path (t1) edge[earrow, ->] (t1l);
                      \path (s1) edge [draw = white, opacity = 0] node[pos= 0.5, inner sep = 1.pt, font = \scriptsize, opacity = 1](tau1){$\empt$}(c1);
                      \path (c2) edge [draw = white, opacity = 0] node[pos= 0.5, inner sep = 1.pt, font = \scriptsize, opacity = 1](tau2){$\empt$}(t1);
                      \path (c2) edge [draw = white, opacity = 0, bend left = 60] node[pos= 0.5, inner sep = 1.pt, font = \scriptsize, opacity = 1](tau3){$\empt$}(c1);
                      \path (s1) edge [draw = white, opacity = 0, bend left = 20] node[pos= 0.5, inner sep = 1.pt, font = \scriptsize, opacity = 1](tau4){$\empt$}(t1);
                      \graph[use existing nodes, edges={color=black, pos = .5, earrow}, edge quotes={fill=white, inner sep=1pt,font= \scriptsize}]{
                      s1 -- {tau1} -> c1;
                      c2 -- {tau2} -> t1;
                      c2 --[bend left = 15] {tau3} ->[bend left = 15] c1;
                      s1 --[bend left = 5] {tau4} ->[bend left = 5] t1;
                      };
                      \graph[use existing nodes, edges={color=black, pos = .5, earrow}, edge quotes={fill=white, inner sep=1pt,font= \scriptsize}]{
                      c1 ->["$\automaton_{t}$"] c2;
                      };
                  \end{tikzpicture}.
    \end{align*}
\end{definition}
\begin{proposition}\label{proposition: Thompson}
    For every structure $\struc$ and $\ExKA$ term $\term$, we have $\jump{\term}_{\struc} = \jump{\automaton_{\term}}_{\struc}$.
\end{proposition}
\begin{proof}[Proof Sketch]
    Let $\term[2]$ be the term $\term$ in which each $\top$ has been replaced with $a \cup \comnf{a}$.
    Then, $\jump{\term[2]}_{\model} = \jump{\term}_{\model}$ holds by \Cref{equation: I- top 0}
    and $\automaton_{\term[2]}$ coincides with $\automaton_{\term}$ (by \Cref{definition: NFA def}).
    By viewing $\term[2]$ as the regular expression over $\WordSIG$, we have $\lang{\term[2]} = \lang{\automaton_{\term[2]}}$ \cite{Thompson68}.
    By straightforward induction on $\term[2]$ using the distributivity, $\jump{\term[2]}_{\struc} = \bigcup_{w \in \lang{\term[2]}} \jump{w}_{\struc}$ \ifiscameraready \else (\Cref{section: proposition: Thompson})\fi.
    Thus $\jump{\term}_{\struc} = \jump{\term[2]}_{\struc} = \bigcup_{w \in \lang{\term[2]}} \jump{w}_{\struc} = \bigcup_{w \in \lang{\automaton_{\term[2]}}} \jump{w}_{\struc} = \jump{\automaton_{\term[2]}}_{\struc} = \jump{\automaton_{\term}}_{\struc}$.
\end{proof}
Thanks to the above proposition, we work directly with NFAs rather than terms in the sequel.

\newcommand{\clang}[1]{[#1]_{\mathrm{c}}}
\newcommand{\Con}{\mathop{\mathrm{Con}}\nolimits}
\subsection{Saturable paths}\label{section: saturable path def}
For a word $w \in (\WordSIG)^*$,
we use $\graph(w)$ to denote the unique graph in $\glang(w)$ (\Cref{definition: graph language ECoR}), up to isomorphism.
Here, each vertex in $\graph(w)$ is indexed by a number in $[0, n]$, from the left to the right,
where $w = a_1 \dots a_n$.
For example,
\[\graph(\comnf{\const{I}} a^{\smile} \comnf{a}) =  \begin{tikzpicture}[baseline = -.5ex]
        \graph[grow right = 1.cm, branch down = 2.5ex, nodes={mysmallnode, font = \scriptsize}]{
        {s1/{$0$}[draw, circle]}
        -!- {c1/{$1$}[draw, circle]}
        -!- {c2/{$2$}[draw, circle]}
        -!- {t1/{$3$}[draw, circle]}
        };
        \node[left = 4pt of s1](s1l){} edge[earrow, ->] (s1);
        \node[right = 4pt of t1](t1l){}; \path (t1) edge[earrow, ->] (t1l);
        \graph[use existing nodes, edges={color=black, pos = .5, earrow}, edge quotes={fill=white, inner sep=1pt,font= \scriptsize}]{
        s1 ->["$\comnf{\const{I}}$"] c1;
        c2 ->["$a$"] c1;
        c2 ->["$\comnf{a}$"] t1;
        };
    \end{tikzpicture}.\]
For each $w$, $G(w)$ is a path graph (by forgetting labels and directions of edges).

Recall $\jump{\automaton}_{\model}$ in \Cref{section: NFA}.
Based on $\jump{\automaton}_{\model}$, we use $\models$ (\Cref{notation: models}) also for automata.
Note that:
\begin{align*}
    \ppmodel \models \automaton & \iff \exists w \in \lang{\automaton}.\ \ppmodel \models w \tag{Def.\ of $\jump{\automaton}_{\model}$}           \\
                                & \iff \exists w \in \lang{\automaton}.\ G(w) \homo \ppmodel. \tag{\Cref{proposition: semantics for graphs ECoR}}
\end{align*}

In this subsection, we consider the following saturability problem:
given an NFA $\automaton[1]$ and a word $w$, is there a saturation $\graph[2]$ of $\graph(w)$ such that $\graph[2]^{\Quo} \not\models \automaton[1]$? (Recall $\graph[2]^{\Quo}$ (\Cref{proposition: structure}).)
This problem can apply to the equational theory (\Cref{theorem: Hintikka}).
For this problem, we introduce \emph{saturable paths}.
\begin{example}\label{example: saturable path}
    Let $\SIG = \set{a}$.
    Let $\automaton[1]$ be the NFA obtained from the term $\comnf{a} \cdot \comnf{a} \cdot \comnf{a}^*$ and $w = \comnf{a}$:
    \begin{align*}
        \automaton[1] & = \begin{tikzpicture}[baseline = -.5ex]
                              \graph[grow right = 1.25cm, branch down = 2.5ex, nodes={mysmallnode, font = \scriptsize}]{
                              {1/{$\mathtt{A}$}[draw, circle]}
                              -!- {2/{$\mathtt{B}$}[draw, circle]}
                              -!- {3/{$\mathtt{C}$}[draw, circle]}
                              };
                              \node[left = 4pt of 1](l){} edge[earrow, ->] (1);
                              \node[right = 4pt of 3](r){}; \path (3) edge[earrow, ->] (r);
                              \path (1) edge [draw = white, opacity = 0] node[pos= 0.5, inner sep = .5pt, font = \scriptsize, opacity = 1](tau1){$\comnf{a}$}(2);
                              \path (2) edge [draw = white, opacity = 0] node[pos= 0.5, inner sep = 1.pt, font = \scriptsize, opacity = 1](tau2){$\comnf{a}$}(3);
                              \node[above = 1.5ex of 3, inner sep = 1.pt,  font= \scriptsize, opacity = 1](tau3){$\comnf{a}$};
                              \graph[use existing nodes, edges={color=black, pos = .5, earrow}, edge quotes={fill=white, inner sep=1pt,font= \scriptsize}]{
                              1 -- {tau1} -> 2;
                              2 -- {tau2} -> 3;
                              3 --[bend left = 20] {tau3} ->[bend left = 20] 3;
                              };
                          \end{tikzpicture} &
        G(w)          & = \begin{tikzpicture}[baseline = -.5ex]
                              \graph[grow right = 1.25cm, branch down = 2.5ex, nodes={mysmallnode, font = \scriptsize}]{
                              {1/{$0$}[draw, circle]}
                              -!- {2/{$1$}[draw, circle]}
                              };
                              \node[left = 4pt of 1](l){} edge[earrow, ->] (1);
                              \node[right = 4pt of 2](r){}; \path (2) edge[earrow, ->] (r);
                              \path (1) edge [draw = white, opacity = 0] node[pos= 0.5, inner sep = 1.pt, font = \scriptsize, opacity = 1](tau1){$\comnf{a}$}(2);
                              \graph[use existing nodes, edges={color=black, pos = .5, earrow}, edge quotes={fill=white, inner sep=1pt,font= \scriptsize}]{
                              1 -- {tau1} -> 2;
                              };
                          \end{tikzpicture}.
    \end{align*}
    Let us consider constructing a saturation $\graph[2]$ of $\graph(w)$ such that $\graph[2]^{\Quo} \not\models \automaton[1]$.
    In this case, the following $\textcolor{blue}{\graph[2]_1}$ is the unique solution ($\textcolor{red}{\graph[2]_2}, \textcolor{red}{\graph[2]_3}, \textcolor{red}{\graph[2]_4}$ are not the solution because $\mathtt{C} \in U_1$):
    \begin{align*}
        \textcolor{blue}{\graph[2]_1} & = \begin{tikzpicture}[baseline = -.5ex]
                                              \graph[grow right = 1.5cm, branch down = 6ex, nodes={mysmallnode, font = \scriptsize}]{
                                              {0/{$0$}[draw, circle]}-!-{1/{$1$}[draw, circle]}
                                              };
                                              \node[left = .5em of 0](l){};
                                              \node[right = .5em of 1](r){};
                                              \node[above left = 2ex and .0em of 0, inner sep = 1.5pt, font = \scriptsize, opacity = 1](a00){$a$};
                                              \node[above left = 2ex and .0em of 1, inner sep = .5pt, font = \scriptsize, opacity = 1](a11){$a$};
                                              \path (0) edge [draw = white, opacity = 0] node[pos= 0.5, inner sep = 1.5pt, font = \scriptsize, opacity = 1](a01){$\comnf{a}$}(1);
                                              \path (1) edge [draw = white, opacity = 0, bend left = 25] node[pos= 0.5, inner sep = .5pt, font = \scriptsize, opacity = 1](a10){$a$}(0);
                                              \node[above right = 2ex and .0em of 0, inner sep = 1.5pt, font = \scriptsize, opacity = 1](I00){$\const{I}$};
                                              \node[above right = 2ex and .0em of 1, inner sep = 1.5pt, font = \scriptsize, opacity = 1](I11){$\const{I}$};
                                              \path (0) edge [draw = white, opacity = 0, bend left = 65] node[pos= 0.5, inner sep = .5pt, font = \scriptsize, opacity = 1](I01){$\comnf{\const{I}}$}(1);
                                              \graph[use existing nodes, edges={color=black, pos = .5, earrow}, edge quotes={fill=white, inner sep=1pt,font= \scriptsize}]{
                                              0 --[bend left = 10] a00 ->[bend left = 10] 0;
                                              1 --[bend left = 10] a11 ->[bend left = 10] 1;
                                              0 -- a01 -> 1;
                                              1 --[bend left = 15] a10 ->[bend left = 15] 0;
                                              0 --[bend left = 10] I00 ->[bend left = 10] 0;
                                              1 --[bend left = 10] I11 ->[bend left = 10] 1;
                                              0 <-[bend left = 10] I01 ->[bend left = 10] 1;
                                              l -> 0; 1 -> r;
                                              };
                                              \node[below = .5ex of 0, font = \scriptsize]{($\mathtt{A}$)};
                                              \node[below = .5ex of 1, font = \scriptsize]{($\mathtt{B}$)};
                                          \end{tikzpicture} &
        \textcolor{red}{\graph[2]_2}  & = \begin{tikzpicture}[baseline = -.5ex]
                                              \graph[grow right = 1.5cm, branch down = 6ex, nodes={mysmallnode, font = \scriptsize}]{
                                              {0/{$0$}[draw, circle]}-!-{1/{$1$}[draw, circle]}
                                              };
                                              \node[left = .5em of 0](l){};
                                              \node[right = .5em of 1](r){};
                                              \node[above left = 2ex and .0em of 0, inner sep = 1.5pt, font = \scriptsize, opacity = 1](a00){$a$};
                                              \node[above left = 2ex and .0em of 1, inner sep = .5pt, font = \scriptsize, opacity = 1](a11){$a$};
                                              \path (0) edge [draw = white, opacity = 0] node[pos= 0.5, inner sep = 1.5pt, font = \scriptsize, opacity = 1](a01){$\comnf{a}$}(1);
                                              \path (1) edge [draw = white, opacity = 0, bend left = 25] node[pos= 0.5, inner sep = .5pt, font = \scriptsize, opacity = 1](a10){$\comnf{a}$}(0);
                                              \node[above right = 2ex and .0em of 0, inner sep = 1.5pt, font = \scriptsize, opacity = 1](I00){$\const{I}$};
                                              \node[above right = 2ex and .0em of 1, inner sep = 1.5pt, font = \scriptsize, opacity = 1](I11){$\const{I}$};
                                              \path (0) edge [draw = white, opacity = 0, bend left = 65] node[pos= 0.5, inner sep = .5pt, font = \scriptsize, opacity = 1](I01){$\comnf{\const{I}}$}(1);
                                              \graph[use existing nodes, edges={color=black, pos = .5, earrow}, edge quotes={fill=white, inner sep=1pt,font= \scriptsize}]{
                                              0 --[bend left = 10] a00 ->[bend left = 10] 0;
                                              1 --[bend left = 10] a11 ->[bend left = 10] 1;
                                              0 -- a01 -> 1;
                                              1 --[bend left = 15] a10 ->[bend left = 15] 0;
                                              0 --[bend left = 10] I00 ->[bend left = 10] 0;
                                              1 --[bend left = 10] I11 ->[bend left = 10] 1;
                                              0 <-[bend left = 10] I01 ->[bend left = 10] 1;
                                              l -> 0; 1 -> r;
                                              };
                                              \node[below = .5ex of 0, font = \scriptsize]{($\mathtt{A}, \mathtt{C}$)};
                                              \node[below = .5ex of 1, font = \scriptsize]{($\mathtt{B}, \textcolor{red}{\mathtt{C}}$)};
                                          \end{tikzpicture} \\
        \textcolor{red}{\graph[2]_3}  & = \begin{tikzpicture}[baseline = -.5ex]
                                              \graph[grow right = 1.5cm, branch down = 6ex, nodes={mysmallnode, font = \scriptsize}]{
                                              {0/{$0$}[draw, circle]}-!-{1/{$1$}[draw, circle]}
                                              };
                                              \node[left = .5em of 0](l){};
                                              \node[right = .5em of 1](r){};
                                              \node[above left = 2ex and .0em of 0, inner sep = 1.5pt, font = \scriptsize, opacity = 1](a00){$\comnf{a}$};
                                              \node[above left = 2ex and .0em of 1, inner sep = .5pt, font = \scriptsize, opacity = 1](a11){$a$};
                                              \path (0) edge [draw = white, opacity = 0] node[pos= 0.5, inner sep = 1.5pt, font = \scriptsize, opacity = 1](a01){$\comnf{a}$}(1);
                                              \path (1) edge [draw = white, opacity = 0, bend left = 25] node[pos= 0.5, inner sep = .5pt, font = \scriptsize, opacity = 1](a10){$a$}(0);
                                              \node[above right = 2ex and .0em of 0, inner sep = 1.5pt, font = \scriptsize, opacity = 1](I00){$\const{I}$};
                                              \node[above right = 2ex and .0em of 1, inner sep = 1.5pt, font = \scriptsize, opacity = 1](I11){$\const{I}$};
                                              \path (0) edge [draw = white, opacity = 0, bend left = 65] node[pos= 0.5, inner sep = .5pt, font = \scriptsize, opacity = 1](I01){$\comnf{\const{I}}$}(1);
                                              \graph[use existing nodes, edges={color=black, pos = .5, earrow}, edge quotes={fill=white, inner sep=1pt,font= \scriptsize}]{
                                              0 --[bend left = 10] a00 ->[bend left = 10] 0;
                                              1 --[bend left = 10] a11 ->[bend left = 10] 1;
                                              0 -- a01 -> 1;
                                              1 --[bend left = 15] a10 ->[bend left = 15] 0;
                                              0 --[bend left = 10] I00 ->[bend left = 10] 0;
                                              1 --[bend left = 10] I11 ->[bend left = 10] 1;
                                              0 <-[bend left = 10] I01 ->[bend left = 10] 1;
                                              l -> 0; 1 -> r;
                                              };
                                              \node[below = .5ex of 0, font = \scriptsize]{($\mathtt{A}, \mathtt{B}, \mathtt{C}$)};
                                              \node[below = .5ex of 1, font = \scriptsize]{($\mathtt{B}, \textcolor{red}{\mathtt{C}}$)};
                                          \end{tikzpicture} &
        \textcolor{red}{\graph[2]_4}  & = \begin{tikzpicture}[baseline = -.5ex]
                                              \graph[grow right = 1.5cm, branch down = 6ex, nodes={mysmallnode, font = \scriptsize}]{
                                              {0/{$0$}[draw, circle]}-!-{1/{$1$}[draw, circle]}
                                              };
                                              \node[left = .5em of 0](l){};
                                              \node[right = .5em of 1](r){};
                                              \node[above left = 2ex and .0em of 0, inner sep = 1.5pt, font = \scriptsize, opacity = 1](a00){$a$};
                                              \node[above left = 2ex and .0em of 1, inner sep = .5pt, font = \scriptsize, opacity = 1](a11){$\comnf{a}$};
                                              \path (0) edge [draw = white, opacity = 0] node[pos= 0.5, inner sep = 1.5pt, font = \scriptsize, opacity = 1](a01){$\comnf{a}$}(1);
                                              \path (1) edge [draw = white, opacity = 0, bend left = 25] node[pos= 0.5, inner sep = .5pt, font = \scriptsize, opacity = 1](a10){$a$}(0);
                                              \node[above right = 2ex and .0em of 0, inner sep = 1.5pt, font = \scriptsize, opacity = 1](I00){$\const{I}$};
                                              \node[above right = 2ex and .0em of 1, inner sep = 1.5pt, font = \scriptsize, opacity = 1](I11){$\const{I}$};
                                              \path (0) edge [draw = white, opacity = 0, bend left = 65] node[pos= 0.5, inner sep = .5pt, font = \scriptsize, opacity = 1](I01){$\comnf{\const{I}}$}(1);
                                              \graph[use existing nodes, edges={color=black, pos = .5, earrow}, edge quotes={fill=white, inner sep=1pt,font= \scriptsize}]{
                                              0 --[bend left = 10] a00 ->[bend left = 10] 0;
                                              1 --[bend left = 10] a11 ->[bend left = 10] 1;
                                              0 -- a01 -> 1;
                                              1 --[bend left = 15] a10 ->[bend left = 15] 0;
                                              0 --[bend left = 10] I00 ->[bend left = 10] 0;
                                              1 --[bend left = 10] I11 ->[bend left = 10] 1;
                                              0 <-[bend left = 10] I01 ->[bend left = 10] 1;
                                              l -> 0; 1 -> r;
                                              };
                                              \node[below = .5ex of 0, font = \scriptsize]{($\mathtt{A}$)};
                                              \node[below = .5ex of 1, font = \scriptsize]{($\mathtt{B}, \textcolor{red}{\mathtt{C}}$)};
                                          \end{tikzpicture}.
    \end{align*}
    Here the states under each vertex denote that they are reachable from the state $\mathtt{A}$ on the vertex $0$;
    more precisely, each set $U_{i} \subseteq |\automaton|$ (where $i \in |G(w)|$) is defined by:
    \begin{align*}
        z \in U_{i} & \defiff \exists v \in \lang{{\automaton}[\bl, z]}.\  G(v) \homo \graph[2]^{\Quo}[\bl, \gquo{i}{\graph[2]}].
    \end{align*}
    Here, ${\automaton}[\bl, z]$ denotes the graph $\automaton$ in which $\rv^{\automaton}$ has been replaced with $z$ (similarly for $\graph[2]^{\Quo}[\bl, \gquo{i}{\graph[2]}]$).
    By definition, $\graph[2]^{\Quo} \not\models \automaton$ (iff $\lnot (\exists v \in \lang{{\automaton}}.\  G(v) \homo \graph[2]^{\Quo})$) iff $\rv^{\automaton[1]} \not\in U_{\rv^{\graph[2]}}$.
    For example, in $\textcolor{red}{\graph[2]_2}$, $\mathtt{C} \in U_{1}$ because $G(\comnf{a} \; \comnf{a} \;  \comnf{a}) \homo \graph[2]_2^{\Quo}$ holds by:
    \begin{align*}
         & 
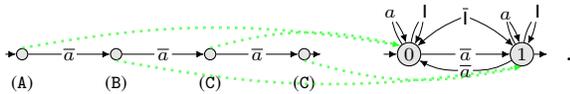
\begin{tikzpicture}[baseline = -.5ex, remember picture]
               \graph[grow right = 1.25cm, branch down = 2.5ex, nodes={mynode, font = \scriptsize}]{
               {L1/{}[draw, circle]}
               -!- {L2/{}[draw, circle]}
               -!- {L3/{}[draw, circle]}
               -!- {L4/{}[draw, circle]}
               };
               \node[left = 4pt of L1](Ll){} edge[earrow, ->] (L1);
               \node[right = 4pt of L4](Lr){}; \path (L4) edge[earrow, ->] (Lr);
               \path (L1) edge [draw = white, opacity = 0] node[pos= 0.5, inner sep = .5pt, font = \scriptsize, opacity = 1](Ltau1){$\comnf{a}$}(L2);
               \path (L2) edge [draw = white, opacity = 0] node[pos= 0.5, inner sep = 1.pt, font = \scriptsize, opacity = 1](Ltau2){$\comnf{a}$}(L3);
               \path (L3) edge [draw = white, opacity = 0] node[pos= 0.5, inner sep = 1.pt, font = \scriptsize, opacity = 1](Ltau3){$\comnf{a}$}(L4);
               \graph[use existing nodes, edges={color=black, pos = .5, earrow}, edge quotes={fill=white, inner sep=1pt,font= \scriptsize}]{
               L1 -- {Ltau1} -> L2;
               L2 -- {Ltau2} -> L3;
               L3 -- {Ltau3} -> L4;
               };
               \node[below = .5ex of L1, font = \scriptsize]{($\mathtt{A}$)};
               \node[below = .5ex of L2, font = \scriptsize]{($\mathtt{B}$)};
               \node[below = .5ex of L3, font = \scriptsize]{($\mathtt{C}$)};
               \node[below = .5ex of L4, font = \scriptsize]{($\mathtt{C}$)};
           \end{tikzpicture}
        \quad
        \begin{tikzpicture}[baseline = -.5ex, remember picture]
            \graph[grow right = 1.5cm, branch down = 6ex, nodes={mysmallnode, font = \scriptsize}]{
            {R0/{$0$}[draw, circle]}-!-{R1/{$1$}[draw, circle]}
            };
            \node[left = .5em of R0](Rl){};
            \node[right = .5em of R1](Rr){};
            \node[above left = 2ex and .0em of R0, inner sep = 1.5pt, font = \scriptsize, opacity = 1](Ra00){$a$};
            \node[above left = 2ex and .0em of R1, inner sep = .5pt, font = \scriptsize, opacity = 1](Ra11){$a$};
            \path (R0) edge [draw = white, opacity = 0] node[pos= 0.5, inner sep = 1.5pt, font = \scriptsize, opacity = 1](Ra01){$\comnf{a}$}(R1);
            \path (R1) edge [draw = white, opacity = 0, bend left = 25] node[pos= 0.5, inner sep = .5pt, font = \scriptsize, opacity = 1](Ra10){$\comnf{a}$}(R0);
            \node[above right = 2ex and .0em of R0, inner sep = 1.5pt, font = \scriptsize, opacity = 1](RI00){$\const{I}$};
            \node[above right = 2ex and .0em of R1, inner sep = 1.5pt, font = \scriptsize, opacity = 1](RI11){$\const{I}$};
            \path (R0) edge [draw = white, opacity = 0, bend left = 65] node[pos= 0.5, inner sep = .5pt, font = \scriptsize, opacity = 1](RI01){$\comnf{\const{I}}$}(R1);
            \graph[use existing nodes, edges={color=black, pos = .5, earrow}, edge quotes={fill=white, inner sep=1pt,font= \scriptsize}]{
            R0 --[bend left = 10] Ra00 ->[bend left = 10] R0;
            R1 --[bend left = 10] Ra11 ->[bend left = 10] R1;
            R0 -- Ra01 -> R1;
            R1 --[bend left = 15] Ra10 ->[bend left = 15] R0;
            R0 --[bend left = 10] RI00 ->[bend left = 10] R0;
            R1 --[bend left = 10] RI11 ->[bend left = 10] R1;
            R0 <-[bend left = 10] RI01 ->[bend left = 10] R1;
            Rl -> R0; R1 -> Rr;
            };
            \node[below = 3ex of R0]{};
        \end{tikzpicture}.
        \begin{tikzpicture}[remember picture, overlay]
            \path (L1.north) edge [homoarrow,->, bend left = 10] (R0.north west);
            \path (L2.south) edge [homoarrow,->, bend right = 10] (R1.south);
            \path (L3.north) edge [homoarrow,->, bend left = 15] (R0.north west);
            \path (L4.south) edge [homoarrow,->, bend right = 15] (R1.south);
        \end{tikzpicture}
    \end{align*}
    In contrast, in \textcolor{blue}{$\graph[2]_1$}, because $\mathtt{C} \not\in U_1$ (i.e., $\forall v \in [\automaton].\  \graph(v) \not\homo \graph[2]_{1}^{\Quo}$) holds, we have $\graph[2]_1^{\Quo} \not\models \automaton$.

    Consider the following $P$, which is the graph $G(w)$ with an equivalence relation $\const{I}$, its complement $\comnf{\const{I}}$, and $\set{U_{i}}_{i \in |G(w)|}$:
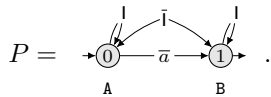
    \begin{align*}
        P & = \begin{tikzpicture}[baseline = -.5ex]
                  \graph[grow right = 1.5cm, branch down = 6ex, nodes={mysmallnode, font = \scriptsize}]{
                  {0/{$0$}[draw, circle]}-!-{1/{$1$}[draw, circle]}
                  };
                  \node[left = .5em of 0](l){};
                  \node[right = .5em of 1](r){};
                  \path (0) edge [draw = white, opacity = 0] node[pos= 0.5, inner sep = 1.5pt, font = \scriptsize, opacity = 1](a01){$\comnf{a}$}(1);
                  \node[above right = 2ex and .0em of 0, inner sep = 1.5pt, font = \scriptsize, opacity = 1](I00){$\const{I}$};
                  \node[above right = 2ex and .0em of 1, inner sep = 1.5pt, font = \scriptsize, opacity = 1](I11){$\const{I}$};
                  \path (0) edge [draw = white, opacity = 0, bend left = 65] node[pos= 0.5, inner sep = .5pt, font = \scriptsize, opacity = 1](I01){$\comnf{\const{I}}$}(1);
                  \graph[use existing nodes, edges={color=black, pos = .5, earrow}, edge quotes={fill=white, inner sep=1pt,font= \scriptsize}]{
                  0 -- a01 -> 1;
                  0 --[bend left = 10] I00 ->[bend left = 10] 0;
                  1 --[bend left = 10] I11 ->[bend left = 10] 1;
                  0 <-[bend left = 10] I01 ->[bend left = 10] 1;
                  l -> 0; 1 -> r;
                  };
                  \node[below = .5ex of 0, font = \scriptsize]{$\mathtt{A}$};
                  \node[below = .5ex of 1, font = \scriptsize]{$\mathtt{B}$};
              \end{tikzpicture}.
    \end{align*}
    Using only the data of $P$, we can show the existence of a saturation of $G(w)$.
    For each pair $\tuple{i, j}$ of vertices,
    if ($\delta_{a}^{\automaton[1]}(U_{i}) \subseteq U_{j}$ and $\delta_{\widebreve{a}}^{\automaton[1]}(U_{j}) \subseteq U_{i}$) or
    ($\delta_{\comnf{a}}^{\automaton[1]}(U_{i}) \subseteq U_{j}$ and $\delta_{\widebreve{\comnf{a}}}^{\automaton[1]}(U_{j}) \subseteq U_{i}$) holds (cf.\ \ref{definition: Hintikka: saturate} in \Cref{definition: Hintikka}),
    then we add either the $a$- or $\comnf{a}$-labeled edge, according to this condition;
    for example, for the pair $\tuple{1, 0}$,
    since $\delta_{a}^{\automaton[1]}(U_1) = \emptyset \subseteq U_0$ and $\delta_{\widebreve{a}}^{\automaton[1]}(U_0) = \emptyset \subseteq U_1$ hold,
    we add the edge for $\tuple{1, 0} \in a^{\graph[2]}$ (we cannot add the edge for $\tuple{1, 0} \in \comnf{a}^{\graph[2]}$ because $\delta_{\comnf{a}}^{\automaton[1]}(U_1) = \set{\mathtt{C}} \not\subseteq U_0$).
    Note that $\set{U_{i}}_{x \in |G(w)|}$ is invariant when we add edges in this strategy.
    By adding such edges as much as possible, we can give a saturation of $G(w)$ from $P$ preserving $\set{U_{i}}_{x \in |G(w)|}$ (cf.\ \Cref{lemma: Hintikka saturate}); finally, \textcolor{blue}{$H_1$} is obtained, as follows (from the left to the right):\\
    \scalebox{.85}
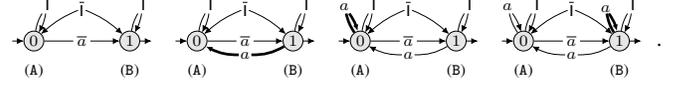
{\begin{tikzpicture}[baseline = -.5ex]
            \graph[grow right = 1.5cm, branch down = 6ex, nodes={mysmallnode, font = \scriptsize}]{
            {0/{$0$}[draw, circle]}-!-{1/{$1$}[draw, circle]}
            };
            \node[left = .5em of 0](l){};
            \node[right = .5em of 1](r){};
            \path (0) edge [draw = white, opacity = 0] node[pos= 0.5, inner sep = 1.5pt, font = \scriptsize, opacity = 1](a01){$\comnf{a}$}(1);
            \node[above right = 2ex and .0em of 0, inner sep = 1.5pt, font = \scriptsize, opacity = 1](I00){$\const{I}$};
            \node[above right = 2ex and .0em of 1, inner sep = 1.5pt, font = \scriptsize, opacity = 1](I11){$\const{I}$};
            \path (0) edge [draw = white, opacity = 0, bend left = 65] node[pos= 0.5, inner sep = .5pt, font = \scriptsize, opacity = 1](I01){$\comnf{\const{I}}$}(1);
            \graph[use existing nodes, edges={color=black, pos = .5, earrow}, edge quotes={fill=white, inner sep=1pt,font= \scriptsize}]{
            0 -- a01 -> 1;
            0 --[bend left = 10] I00 ->[bend left = 10] 0;
            1 --[bend left = 10] I11 ->[bend left = 10] 1;
            0 <-[bend left = 10] I01 ->[bend left = 10] 1;
            l -> 0; 1 -> r;
            };
            \node[below = .5ex of 0, font = \scriptsize]{($\mathtt{A}$)};
            \node[below = .5ex of 1, font = \scriptsize]{($\mathtt{B}$)};
        \end{tikzpicture} \hspace{-1em}
        \begin{tikzpicture}[baseline = -.5ex]
            \graph[grow right = 1.5cm, branch down = 6ex, nodes={mysmallnode, font = \scriptsize}]{
            {0/{$0$}[draw, circle]}-!-{1/{$1$}[draw, circle]}
            };
            \node[left = .5em of 0](l){};
            \node[right = .5em of 1](r){};
            \path (0) edge [draw = white, opacity = 0] node[pos= 0.5, inner sep = 1.5pt, font = \scriptsize, opacity = 1](a01){$\comnf{a}$}(1);
            \path (1) edge [draw = white, opacity = 0, bend left = 25] node[pos= 0.5, inner sep = .5pt, font = \scriptsize, opacity = 1](a10){$a$}(0);
            \node[above right = 2ex and .0em of 0, inner sep = 1.5pt, font = \scriptsize, opacity = 1](I00){$\const{I}$};
            \node[above right = 2ex and .0em of 1, inner sep = 1.5pt, font = \scriptsize, opacity = 1](I11){$\const{I}$};
            \path (0) edge [draw = white, opacity = 0, bend left = 65] node[pos= 0.5, inner sep = .5pt, font = \scriptsize, opacity = 1](I01){$\comnf{\const{I}}$}(1);
            \graph[use existing nodes, edges={color=black, pos = .5, earrow}, edge quotes={fill=white, inner sep=1pt,font= \scriptsize}]{
            0 -- a01 -> 1;
            1 --[bend left = 15, line width = 1.2pt] a10 ->[bend left = 15, line width = 1.2pt] 0;
            0 --[bend left = 10] I00 ->[bend left = 10] 0;
            1 --[bend left = 10] I11 ->[bend left = 10] 1;
            0 <-[bend left = 10] I01 ->[bend left = 10] 1;
            l -> 0; 1 -> r;
            };
            \node[below = .5ex of 0, font = \scriptsize]{($\mathtt{A}$)};
            \node[below = .5ex of 1, font = \scriptsize]{($\mathtt{B}$)};
        \end{tikzpicture} \hspace{-1em}
        \begin{tikzpicture}[baseline = -.5ex]
            \graph[grow right = 1.5cm, branch down = 6ex, nodes={mysmallnode, font = \scriptsize}]{
            {0/{$0$}[draw, circle]}-!-{1/{$1$}[draw, circle]}
            };
            \node[left = .5em of 0](l){};
            \node[right = .5em of 1](r){};
            \node[above left = 2ex and .0em of 0, inner sep = 1.5pt, font = \scriptsize, opacity = 1](a00){$a$};
            \path (0) edge [draw = white, opacity = 0] node[pos= 0.5, inner sep = 1.5pt, font = \scriptsize, opacity = 1](a01){$\comnf{a}$}(1);
            \path (1) edge [draw = white, opacity = 0, bend left = 25] node[pos= 0.5, inner sep = .5pt, font = \scriptsize, opacity = 1](a10){$a$}(0);
            \node[above right = 2ex and .0em of 0, inner sep = 1.5pt, font = \scriptsize, opacity = 1](I00){$\const{I}$};
            \node[above right = 2ex and .0em of 1, inner sep = 1.5pt, font = \scriptsize, opacity = 1](I11){$\const{I}$};
            \path (0) edge [draw = white, opacity = 0, bend left = 65] node[pos= 0.5, inner sep = .5pt, font = \scriptsize, opacity = 1](I01){$\comnf{\const{I}}$}(1);
            \graph[use existing nodes, edges={color=black, pos = .5, earrow}, edge quotes={fill=white, inner sep=1pt,font= \scriptsize}]{
            0 --[bend left = 10, line width = 1.2pt] a00 ->[bend left = 10, line width = 1.2pt] 0;
            0 -- a01 -> 1;
            1 --[bend left = 15] a10 ->[bend left = 15] 0;
            0 --[bend left = 10] I00 ->[bend left = 10] 0;
            1 --[bend left = 10] I11 ->[bend left = 10] 1;
            0 <-[bend left = 10] I01 ->[bend left = 10] 1;
            l -> 0; 1 -> r;
            };
            \node[below = .5ex of 0, font = \scriptsize]{($\mathtt{A}$)};
            \node[below = .5ex of 1, font = \scriptsize]{($\mathtt{B}$)};
        \end{tikzpicture} \hspace{-1em}
        \begin{tikzpicture}[baseline = -.5ex]
            \graph[grow right = 1.5cm, branch down = 6ex, nodes={mysmallnode, font = \scriptsize}]{
            {0/{$0$}[draw, circle]}-!-{1/{$1$}[draw, circle]}
            };
            \node[left = .5em of 0](l){};
            \node[right = .5em of 1](r){};
            \node[above left = 2ex and .0em of 0, inner sep = 1.5pt, font = \scriptsize, opacity = 1](a00){$a$};
            \node[above left = 2ex and .0em of 1, inner sep = .5pt, font = \scriptsize, opacity = 1](a11){$a$};
            \path (0) edge [draw = white, opacity = 0] node[pos= 0.5, inner sep = 1.5pt, font = \scriptsize, opacity = 1](a01){$\comnf{a}$}(1);
            \path (1) edge [draw = white, opacity = 0, bend left = 25] node[pos= 0.5, inner sep = .5pt, font = \scriptsize, opacity = 1](a10){$a$}(0);
            \node[above right = 2ex and .0em of 0, inner sep = 1.5pt, font = \scriptsize, opacity = 1](I00){$\const{I}$};
            \node[above right = 2ex and .0em of 1, inner sep = 1.5pt, font = \scriptsize, opacity = 1](I11){$\const{I}$};
            \path (0) edge [draw = white, opacity = 0, bend left = 65] node[pos= 0.5, inner sep = .5pt, font = \scriptsize, opacity = 1](I01){$\comnf{\const{I}}$}(1);
            \graph[use existing nodes, edges={color=black, pos = .5, earrow}, edge quotes={fill=white, inner sep=1pt,font= \scriptsize}]{
            0 --[bend left = 10] a00 ->[bend left = 10] 0;
            1 --[bend left = 10, line width = 1.2pt] a11 ->[bend left = 10, line width = 1.2pt] 1;
            0 -- a01 -> 1;
            1 --[bend left = 15] a10 ->[bend left = 15] 0;
            0 --[bend left = 10] I00 ->[bend left = 10] 0;
            1 --[bend left = 10] I11 ->[bend left = 10] 1;
            0 <-[bend left = 10] I01 ->[bend left = 10] 1;
            l -> 0; 1 -> r;
            };
            \node[below = .5ex of 0, font = \scriptsize]{($\mathtt{A}$)};
            \node[below = .5ex of 1, font = \scriptsize]{($\mathtt{B}$)};
        \end{tikzpicture}.
    }
\end{example}
\begin{example}[another example with $\bl^{\smile}$ and $\comnf{\const{I}}$]\label{example: saturable path 2}
    Let $\SIG = \set{a}$.
    Let $\automaton[1]$ be the NFA obtained from $\comnf{\const{I}} \cdot (\widebreve{a} \cdot \widebreve{\comnf{a}})^*$ and $w = a \comnf{a}$:
    \begin{align*}
        \automaton[1] & = \scalebox{.85}{\begin{tikzpicture}[baseline = -.5ex]
                                                 \graph[grow right = 1.cm, branch down = 2.5ex, nodes={mysmallnode, font = \scriptsize}]{
                                                 {1/{$\mathtt{A}$}[draw, circle]}
                                                 -!- {2/{$\mathtt{B}$}[draw, circle]}
                                                 -!- {3/{$\mathtt{C}$}[draw, circle]}
                                                 -!- {4/{$\mathtt{D}$}[draw, circle]}
                                                 };
                                                 \node[left = 4pt of 1](l){} edge[earrow, ->] (1);
                                                 \node[right = 4pt of 4](r){}; \path (4) edge[earrow, ->] (r);
                                                 \path (1) edge [draw = white, opacity = 0] node[pos= 0.5, inner sep = .5pt, font = \scriptsize, opacity = 1](tau1){$\comnf{\const{I}}$}(2);
                                                 \path (2) edge [draw = white, opacity = 0] node[pos= 0.5, inner sep = 1.pt, font = \scriptsize, opacity = 1](tau2){$\widebreve{a}$}(3);
                                                 \path (3) edge [draw = white, opacity = 0] node[pos= 0.5, inner sep = 1.pt, font = \scriptsize, opacity = 1](tau3){$\widebreve{\comnf{a}}$}(4);
                                                 \path (2) edge [draw = white, opacity = 0, bend left = 25] node[pos= 0.5, inner sep = 1.pt, font = \scriptsize, opacity = 1](g){$\empt$}(4);
                                                 \path (4) edge [draw = white, opacity = 0, bend left = 25] node[pos= 0.5, inner sep = 1.pt, font = \scriptsize, opacity = 1](b){$\empt$}(2);
                                                 \graph[use existing nodes, edges={color=black, pos = .5, earrow}, edge quotes={fill=white, inner sep=1pt,font= \scriptsize}]{
                                                 1 -- {tau1} -> 2;
                                                 2 -- {tau2} -> 3;
                                                 3 -- {tau3} -> 4;
                                                 2 --[bend left = 15] {g} ->[bend left = 15] 4;
                                                 4 --[bend left = 15] {b} ->[bend left = 15] 2;
                                                 };
                                             \end{tikzpicture}} &
        G(w)          & = \scalebox{.85}{\begin{tikzpicture}[baseline = -.5ex]
                                                 \graph[grow right = 1.cm, branch down = 2.5ex, nodes={mysmallnode, font = \scriptsize}]{
                                                 {1/{$0$}[draw, circle]}
                                                 -!- {2/{$1$}[draw, circle]}
                                                 -!- {3/{$2$}[draw, circle]}
                                                 };
                                                 \node[left = 4pt of 1](l){} edge[earrow, ->] (1);
                                                 \node[right = 4pt of 3](r){}; \path (3) edge[earrow, ->] (r);
                                                 \path (1) edge [draw = white, opacity = 0] node[pos= 0.5, inner sep = 1.pt, font = \scriptsize, opacity = 1](tau1){$a$}(2);
                                                 \path (2) edge [draw = white, opacity = 0] node[pos= 0.5, inner sep = 1.pt, font = \scriptsize, opacity = 1](tau2){$\comnf{a}$}(3);
                                                 \graph[use existing nodes, edges={color=black, pos = .5, earrow}, edge quotes={fill=white, inner sep=1pt,font= \scriptsize}]{
                                                 1 -- {tau1} -> 2;
                                                 2 -- {tau2} -> 3;
                                                 };
                                             \end{tikzpicture}}.
    \end{align*}
    Let us consider constructing a saturation $\graph[2]$ of $\graph(w)$ such that $\graph[2]^{\Quo} \not\models \automaton[1]$.
    Then, the following \textcolor{blue}{$\graph[2]$} is a solution (note that $\mathtt{D} \not\in U_2$; thus, $\graph[2]^{\Quo} \not\models \automaton$)
    and the following $P$ is the corresponding saturable path:
    \begin{align*}
        \textcolor{blue}{\graph[2]} & = 
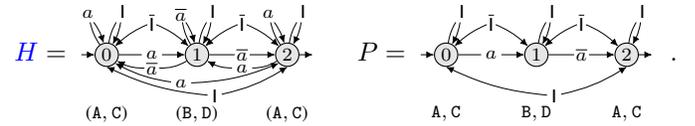
\hspace{-.4em} \begin{tikzpicture}[baseline = -.5ex]
                                                           \graph[grow right = 1.2cm, branch down = 6ex, nodes={mysmallnode, font = \scriptsize}]{
                                                           {0/{$0$}[draw, circle]}-!-{1/{$1$}[draw, circle]}-!-{2/{$2$}[draw, circle]}
                                                           };
                                                           \node[left = .5em of 0](l){};
                                                           \node[right = .5em of 2](r){};
                                                           \node[above left = 2ex and .0em of 0, inner sep = 1.5pt, font = \scriptsize, opacity = 1](a00){$a$};
                                                           \node[above left = 2ex and .0em of 1, inner sep = .5pt, font = \scriptsize, opacity = 1](a11){$\comnf{a}$};
                                                           \node[above left = 2ex and .0em of 2, inner sep = 1.5pt, font = \scriptsize, opacity = 1](a22){$a$};
                                                           \path (0) edge [draw = white, opacity = 0] node[pos= 0.5, inner sep = 1.5pt, font = \scriptsize, opacity = 1](a01){$a$}(1);
                                                           \path (1) edge [draw = white, opacity = 0, bend left = 25] node[pos= 0.5, inner sep = .5pt, font = \scriptsize, opacity = 1](a10){$\comnf{a}$}(0);
                                                           \path (1) edge [draw = white, opacity = 0] node[pos= 0.5, inner sep = .5pt, font = \scriptsize, opacity = 1](a12){$\comnf{a}$}(2);
                                                           \path (2) edge [draw = white, opacity = 0, bend left = 25] node[pos= 0.5, inner sep = 1.5pt, font = \scriptsize, opacity = 1](a21){$a$}(1);
                                                           \path (0) edge [draw = white, opacity = 0, bend right = 30] node[pos= 0.4, inner sep = 1.5pt, font = \scriptsize, opacity = 1](a02){$a$}(2);
                                                           \node[above right = 2ex and .0em of 0, inner sep = 1.5pt, font = \scriptsize, opacity = 1](I00){$\const{I}$};
                                                           \node[above right = 2ex and .0em of 1, inner sep = 1.5pt, font = \scriptsize, opacity = 1](I11){$\const{I}$};
                                                           \node[above right = 2ex and .0em of 2, inner sep = 1.5pt, font = \scriptsize, opacity = 1](I22){$\const{I}$};
                                                           \path (0) edge [draw = white, opacity = 0, bend left = 65] node[pos= 0.5, inner sep = .5pt, font = \scriptsize, opacity = 1](I01){$\comnf{\const{I}}$}(1);
                                                           \path (1) edge [draw = white, opacity = 0, bend left = 65] node[pos= 0.5, inner sep = .5pt, font = \scriptsize, opacity = 1](I12){$\comnf{\const{I}}$}(2);
                                                           \path (0) edge [draw = white, opacity = 0, bend right = 45] node[pos= 0.6, inner sep = 1.5pt, font = \scriptsize, opacity = 1](I02){$\const{I}$}(2);
                                                           \graph[use existing nodes, edges={color=black, pos = .5, earrow}, edge quotes={fill=white, inner sep=1pt,font= \scriptsize}]{
                                                           0 --[bend left = 10] a00 ->[bend left = 10] 0;
                                                           1 --[bend left = 10] a11 ->[bend left = 10] 1;
                                                           2 --[bend left = 10] a22 ->[bend left = 10] 2;
                                                           0 --[] a01 ->[] 1;
                                                           1 --[bend left = 15] a10 ->[bend left = 15] 0;
                                                           1 --[] a12 ->[] 2;
                                                           2 --[bend left = 15] a21 ->[bend left = 15] 1;
                                                           0.south east <-[bend right = 10] a02 ->[bend right = 10] 2.south west;
                                                           0 --[bend left = 10] I00 ->[bend left = 10] 0;
                                                           1 --[bend left = 10] I11 ->[bend left = 10] 1;
                                                           2 --[bend left = 10] I22 ->[bend left = 10] 2;
                                                           0 <-[bend left = 10] I01 ->[bend left = 10] 1;
                                                           1 <-[bend left = 10] I12 ->[bend left = 10] 2;
                                                           0.south <-[bend right = 10] I02 ->[bend right = 10] 2.south;
                                                           l -> 0; 2 -> r;
                                                           };
                                                           \node[below = 2.5ex of 0,  font = \scriptsize]{($\mathtt{A,C}$)};
                                                           \node[below = 2.5ex of 1,  font = \scriptsize]{($\mathtt{B,D}$)};
                                                           \node[below = 2.5ex of 2,  font = \scriptsize]{($\mathtt{A,C}$)};
                                                       \end{tikzpicture}  &
        P                           & = \hspace{-.4em}  \begin{tikzpicture}[baseline = -.5ex]
                                                            \graph[grow right = 1.2cm, branch down = 6ex, nodes={mysmallnode, font = \scriptsize}]{
                                                            {0/{$0$}[draw, circle]}-!-{1/{$1$}[draw, circle]}-!-{2/{$2$}[draw, circle]}
                                                            };
                                                            \node[left = .5em of 0](l){};
                                                            \node[right = .5em of 2](r){};
                                                            \path (0) edge [draw = white, opacity = 0] node[pos= 0.5,  inner sep = 1.5pt, font = \scriptsize, opacity = 1](a01){$a$}(1);
                                                            \path (1) edge [draw = white, opacity = 0] node[pos= 0.5,  inner sep = .5pt, font = \scriptsize, opacity = 1](a12){$\comnf{a}$}(2);
                                                            \node[above right = 2ex and .0em of 0, inner sep = 1.5pt, font = \scriptsize, opacity = 1](I00){$\const{I}$};
                                                            \node[above right = 2ex and .0em of 1, inner sep = 1.5pt, font = \scriptsize, opacity = 1](I11){$\const{I}$};
                                                            \node[above right = 2ex and .0em of 2, inner sep = 1.5pt, font = \scriptsize, opacity = 1](I22){$\const{I}$};
                                                            \path (0) edge [draw = white, opacity = 0, bend left = 65] node[pos= 0.5,  inner sep = .5pt, font = \scriptsize, opacity = 1](I01){$\comnf{\const{I}}$}(1);
                                                            \path (1) edge [draw = white, opacity = 0, bend left = 65] node[pos= 0.5,  inner sep = .5pt, font = \scriptsize, opacity = 1](I12){$\comnf{\const{I}}$}(2);
                                                            \path (0) edge [draw = white, opacity = 0, bend right = 45] node[pos= 0.6,  inner sep = 1.5pt, font = \scriptsize, opacity = 1](I02){$\const{I}$}(2);
                                                            \graph[use existing nodes, edges={color=black, pos = .5, earrow}, edge quotes={fill=white, inner sep=1pt,font= \scriptsize}]{
                                                            0 --[] a01 ->[] 1;
                                                            1 --[] a12 ->[] 2;
                                                            0 --[bend left = 10] I00 ->[bend left = 10] 0;
                                                            1 --[bend left = 10] I11 ->[bend left = 10] 1;
                                                            2 --[bend left = 10] I22 ->[bend left = 10] 2;
                                                            0 <-[bend left = 10] I01 ->[bend left = 10] 1;
                                                            1 <-[bend left = 10] I12 ->[bend left = 10] 2;
                                                            0.south <-[bend right = 10] I02 ->[bend right = 10] 2.south;
                                                            l -> 0; 2 -> r;
                                                            };
                                                            \node[below = 2.5ex of 0,  font = \scriptsize]{$\mathtt{A,C}$};
                                                            \node[below = 2.5ex of 1,  font = \scriptsize]{$\mathtt{B,D}$};
                                                            \node[below = 2.5ex of 2,  font = \scriptsize]{$\mathtt{A,C}$};
                                                        \end{tikzpicture}.
    \end{align*}
\end{example}
(See \Cref{remark: full KACC} for an example when there is no saturation.)

Inspired by $\set{U_{i}}_{i \in |G(w)|}$ above, for characterizing the saturability problem, we define \emph{saturable paths};
they are path graphs with additional data (an equivalence relation $\const{I}$, its complement $\comnf{\const{I}}$, and $\set{U_{i}}_{i \in |G(w)|}$) for taking saturations appropriately.
For graphs $\graph[1], \graph[2]$ over $\SIG_{\const{I}}^{(-)}$,
we say that $\graph[2]$ is an \emph{$\const{I}$-saturation} of $\graph[1]$ if
$\graph[2]$ is an edge-extension of $\graph[1]$ such that
\begin{itemize}
    \item for every $a \in \SIG^{(-)}$, $a^{\graph[2]} = a^{\graph[1]}$;
    \item $\const{I}^{\graph[2]}$ is an equivalence relation and $\comnf{\const{I}}^{\graph[2]} = |\graph[2]|^2 \setminus \const{I}^{\graph[2]}$;
    \item $\graph[2]$ is consistent.
\end{itemize}
(By definition, $\graph[2]$ is uniquely determined from $\const{I}^{\graph[2]}$, if exists.)

\begin{definition}[saturable path]\label{definition: Hintikka}
    For an NFA $\automaton$ over $\AutSIG$ and a word $w$ over $\WordSIG$,
    consider a pair $P = \tuple{G, \set{U_{i}}_{i \in |G|}}$ of
    \begin{itemize}
        \item $G$ an $\const{I}$-saturation of $G(w)$;
        \item $U_i \subseteq |\automaton|$ for each $i \in |\graph|$.
    \end{itemize}
    For $a \in \SIG_{\const{I}}^{(-, \smile)}$, let
    \[\Con_{a}^{\automaton}(U, U')  \defeq (\delta_{a}^{\automaton}(U) \subseteq U' \land \delta_{\widebreve{a}}^{\automaton}(U') \subseteq U).\]
    Then we say that $P$ is a \emph{saturable path} for $\not\models_{\REL} w \le \automaton$ if the following three hold:
    \begin{description}
        \item[(P-s-t)\label{definition: Hintikka: u src tgt}] $\lv^{\automaton} \in U_{\lv^{G}}$ and $\rv^{\automaton} \not\in U_{\rv^{G}}$;
        \item[(P-Con)\label{definition: Hintikka: consistent}] for all $a \in \SIG_{\const{I}}^{(-)}$ and $\tuple{i, j} \in a^{G}$, $\Con_{a}^{\automaton}(U_{i}, U_{j})$;
        \item[(P-Sat)\label{definition: Hintikka: saturate}]  for all $a \in \SIG_{\const{I}}$ and $\tuple{i, j} \in |G|^2$,\\
            $\Con_{a}^{\automaton}(U_{i}, U_{j}) \lor \Con_{\comnf{a}}^{\automaton}(U_{i}, U_{j})$.
    \end{description}
\end{definition}
Saturable paths can characterize the saturability problem, as \Cref{lemma: Hintikka}.
We first show the following:
\begin{lemma}\label{lemma: Hintikka saturate}
    For every saturable path $P = \tuple{G, \set{U_{i}}_{i \in |G|}}$ for $\not\models_{\REL} w \le \automaton$,
    there is a saturation $H$ of $G$ such that
    \begin{description}
        \item[(P-Con')\label{definition: Hintikka: consistent 2}] for all $a \in \SIG_{\const{I}}^{(-)}$ and $\tuple{i, j} \in a^{\graph[2]}$, $\Con_{a}^{\automaton}(U_{i}, U_{j})$.
    \end{description}
\end{lemma}
\begin{proof}
    Starting from $\graph[2] = G$, we add edges labeled with $a \in \SIG^{(-)}$ while preserving \ref{definition: Hintikka: consistent 2}, by repeating the following:
    \begin{itemize}
        \item If $\tuple{i, j} \in (\const{I}^{\graph[2]} \cdot a^{\graph[2]} \cdot \const{I}^{\graph[2]}) \setminus a^{\graph[2]}$ for some $a \in \SIG^{(-)}$ and $i, j \in |\graph[2]|$,
              we add the edge to $\graph[2]$.
              Then, $\Con_{a}^{\automaton[1]}(U_{i}, U_{j})$ holds as follows.
              Let $i', j'$ be s.t.\  $\tuple{i, i'} \in \const{I}^{\graph[2]}$, $\tuple{i', j'} \in a^{\graph[2]}$, and $\tuple{j', j} \in \const{I}^{\graph[2]}$.
              For every $z \in |\automaton|$,
              we have
              \begin{align*}
                  z \in U_{i'} & \Longrightarrow z \in \delta_{\const{I}}^{\automaton}(U_{i}) \tag{$U_{i'} \subseteq \delta_{\const{I}}^{\automaton}(U_{i})$ \ref{definition: Hintikka: consistent 2}} \\
                               & \Longrightarrow z \in U_i \tag{$\delta_{\const{I}}^{\automaton} = (\const{I}^{\automaton})^*$ is reflexive}                                                           \\
                               & \dots \Longrightarrow z \in U_{i'}. \tag*{($\const{I}^{\graph[2]}$ is symmetric)}
              \end{align*}
              Thus $U_{i} = U_{i'}$.
              In the same way, $U_{j} = U_{j'}$.
              Hence $\Con_{a}^{\automaton[1]}(U_{i}, U_{j})$ is derived from $\Con_{a}^{\automaton[1]}(U_{i'}, U_{j'})$.
        \item Otherwise, since $\graph[2]$ is not edge-saturated, $\tuple{i, j} \not\in a^{\graph[2]} \cup \comnf{a}^{\graph[2]}$ for some $a \in \SIG$ and $i, j \in |\graph[2]|$.
              If $\Con_{a}^{\automaton[1]}(U_{i}, U_{j})$ holds, then we add the edge for $\tuple{i, j} \in a^{\graph[2]}$ to $\graph[2]$.
              Otherwise, since $\Con_{\comnf{a}}^{\automaton[1]}(U_{i}, U_{j})$ holds by \ref{definition: Hintikka: saturate},
              we add the edge for $\tuple{i, j} \in \comnf{a}^{\graph[2]}$ to $\graph[2]$.
    \end{itemize}
    Then $\graph[2]$ is a saturation of $G$, as follows.
    For \ref{definition: consistent: I}:
    Because $G$ is an $\const{I}$-saturation.
    For \ref{definition: consistent: saturated}: Clear.
    For \ref{definition: consistent: consistent}:
    By that $G$ is consistent and $(\const{I}^{\graph[2]} \cdot a^{\graph[2]} \cdot \const{I}^{\graph[2]}) \cap (\const{I}^{\graph[2]} \cdot \comnf{a}^{\graph[2]} \cdot \const{I}^{\graph[2]}) = \emptyset$ is preserved (by the construction of $\graph[2]$).
\end{proof}
\begin{example}[of \Cref{lemma: Hintikka saturate}]\label{example: Hintikka saturate}
    Recall $\automaton$, $w$, and the saturable path $P = \tuple{G, \set{U_{i}}_{i \in |G|}}$ in \Cref{example: saturable path 2}.
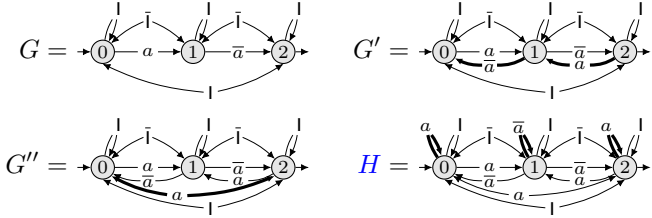
    \begin{align*}
        G                                                                                                                                                                  & = \hspace{-.6em}
        \begin{tikzpicture}[baseline = -.5ex]
            \graph[grow right = 1.2cm, branch down = 6ex, nodes={mysmallnode, font = \scriptsize}]{
            {0/{$0$}[draw, circle]}-!-{1/{$1$}[draw, circle]}-!-{2/{$2$}[draw, circle]}
            };
            \node[left = .5em of 0](l){};
            \node[right = .5em of 2](r){};
            \path (0) edge [draw = white, opacity = 0] node[pos= 0.5, inner sep = 1.5pt, font = \scriptsize, opacity = 1](a01){$a$}(1);
            \path (1) edge [draw = white, opacity = 0] node[pos= 0.5, inner sep = .5pt, font = \scriptsize, opacity = 1](a12){$\comnf{a}$}(2);
            \node[above right = 2ex and .0em of 0, inner sep = 1.5pt, font = \scriptsize, opacity = 1](I00){$\const{I}$};
            \node[above right = 2ex and .0em of 1, inner sep = 1.5pt, font = \scriptsize, opacity = 1](I11){$\const{I}$};
            \node[above right = 2ex and .0em of 2, inner sep = 1.5pt, font = \scriptsize, opacity = 1](I22){$\const{I}$};
            \path (0) edge [draw = white, opacity = 0, bend left = 65] node[pos= 0.5, inner sep = .5pt, font = \scriptsize, opacity = 1](I01){$\comnf{\const{I}}$}(1);
            \path (1) edge [draw = white, opacity = 0, bend left = 65] node[pos= 0.5, inner sep = .5pt, font = \scriptsize, opacity = 1](I12){$\comnf{\const{I}}$}(2);
            \path (0) edge [draw = white, opacity = 0, bend right = 45] node[pos= 0.6, inner sep = 1.5pt, font = \scriptsize, opacity = 1](I02){$\const{I}$}(2);
            \graph[use existing nodes, edges={color=black, pos = .5, earrow}, edge quotes={fill=white, inner sep=1pt,font= \scriptsize}]{
            0 --[] a01 ->[] 1;
            1 --[] a12 ->[] 2;
            0 --[bend left = 10] I00 ->[bend left = 10] 0;
            1 --[bend left = 10] I11 ->[bend left = 10] 1;
            2 --[bend left = 10] I22 ->[bend left = 10] 2;
            0 <-[bend left = 10] I01 ->[bend left = 10] 1;
            1 <-[bend left = 10] I12 ->[bend left = 10] 2;
            0.south <-[bend right = 10] I02 ->[bend right = 10] 2.south;
            l -> 0; 2 -> r;
            };
        \end{tikzpicture} &
        G'                                                                                                                                                                 & = \hspace{-.6em} \begin{tikzpicture}[baseline = -.5ex]
                                                                                                                                                                                                  \graph[grow right = 1.2cm, branch down = 6ex, nodes={mysmallnode, font = \scriptsize}]{
                                                                                                                                                                                                  {0/{$0$}[draw, circle]}-!-{1/{$1$}[draw, circle]}-!-{2/{$2$}[draw, circle]}
                                                                                                                                                                                                  };
                                                                                                                                                                                                  \node[left = .5em of 0](l){};
                                                                                                                                                                                                  \node[right = .5em of 2](r){};
                                                                                                                                                                                                  \path (0) edge [draw = white, opacity = 0] node[pos= 0.5, inner sep = 1.5pt, font = \scriptsize, opacity = 1](a01){$a$}(1);
                                                                                                                                                                                                  \path (1) edge [draw = white, opacity = 0, bend left = 25] node[pos= 0.5, inner sep = .5pt, font = \scriptsize, opacity = 1](a10){$\comnf{a}$}(0);
                                                                                                                                                                                                  \path (1) edge [draw = white, opacity = 0] node[pos= 0.5, inner sep = .5pt, font = \scriptsize, opacity = 1](a12){$\comnf{a}$}(2);
                                                                                                                                                                                                  \path (2) edge [draw = white, opacity = 0, bend left = 25] node[pos= 0.5, inner sep = 1.5pt, font = \scriptsize, opacity = 1](a21){$a$}(1);
                                                                                                                                                                                                  \node[above right = 2ex and .0em of 0, inner sep = 1.5pt, font = \scriptsize, opacity = 1](I00){$\const{I}$};
                                                                                                                                                                                                  \node[above right = 2ex and .0em of 1, inner sep = 1.5pt, font = \scriptsize, opacity = 1](I11){$\const{I}$};
                                                                                                                                                                                                  \node[above right = 2ex and .0em of 2, inner sep = 1.5pt, font = \scriptsize, opacity = 1](I22){$\const{I}$};
                                                                                                                                                                                                  \path (0) edge [draw = white, opacity = 0, bend left = 65] node[pos= 0.5, inner sep = .5pt, font = \scriptsize, opacity = 1](I01){$\comnf{\const{I}}$}(1);
                                                                                                                                                                                                  \path (1) edge [draw = white, opacity = 0, bend left = 65] node[pos= 0.5, inner sep = .5pt, font = \scriptsize, opacity = 1](I12){$\comnf{\const{I}}$}(2);
                                                                                                                                                                                                  \path (0) edge [draw = white, opacity = 0, bend right = 45] node[pos= 0.6, inner sep = 1.5pt, font = \scriptsize, opacity = 1](I02){$\const{I}$}(2);
                                                                                                                                                                                                  \graph[use existing nodes, edges={color=black, pos = .5, earrow}, edge quotes={fill=white, inner sep=1pt,font= \scriptsize}]{
                                                                                                                                                                                                  0 --[] a01 ->[] 1;
                                                                                                                                                                                                  1 --[bend left = 15, line width = 1.2pt] a10 ->[bend left = 15, line width = 1.2pt] 0;
                                                                                                                                                                                                  1 --[] a12 ->[] 2;
                                                                                                                                                                                                  2 --[bend left = 15, line width = 1.2pt] a21 ->[bend left = 15, line width = 1.2pt] 1;
                                                                                                                                                                                                  0 --[bend left = 10] I00 ->[bend left = 10] 0;
                                                                                                                                                                                                  1 --[bend left = 10] I11 ->[bend left = 10] 1;
                                                                                                                                                                                                  2 --[bend left = 10] I22 ->[bend left = 10] 2;
                                                                                                                                                                                                  0 <-[bend left = 10] I01 ->[bend left = 10] 1;
                                                                                                                                                                                                  1 <-[bend left = 10] I12 ->[bend left = 10] 2;
                                                                                                                                                                                                  0.south <-[bend right = 10] I02 ->[bend right = 10] 2.south;
                                                                                                                                                                                                  l -> 0; 2 -> r;
                                                                                                                                                                                                  };
                                                                                                                                                                                              \end{tikzpicture}
        \\
        G''                                                                                                                                                                & = \hspace{-.6em} \begin{tikzpicture}[baseline = -.5ex]
                                                                                                                                                                                                  \graph[grow right = 1.2cm, branch down = 6ex, nodes={mysmallnode, font = \scriptsize}]{
                                                                                                                                                                                                  {0/{$0$}[draw, circle]}-!-{1/{$1$}[draw, circle]}-!-{2/{$2$}[draw, circle]}
                                                                                                                                                                                                  };
                                                                                                                                                                                                  \node[left = .5em of 0](l){};
                                                                                                                                                                                                  \node[right = .5em of 2](r){};
                                                                                                                                                                                                  \path (0) edge [draw = white, opacity = 0] node[pos= 0.5, inner sep = 1.5pt, font = \scriptsize, opacity = 1](a01){$a$}(1);
                                                                                                                                                                                                  \path (1) edge [draw = white, opacity = 0, bend left = 25] node[pos= 0.5, inner sep = .5pt, font = \scriptsize, opacity = 1](a10){$\comnf{a}$}(0);
                                                                                                                                                                                                  \path (1) edge [draw = white, opacity = 0] node[pos= 0.5, inner sep = .5pt, font = \scriptsize, opacity = 1](a12){$\comnf{a}$}(2);
                                                                                                                                                                                                  \path (2) edge [draw = white, opacity = 0, bend left = 25] node[pos= 0.5, inner sep = 1.5pt, font = \scriptsize, opacity = 1](a21){$a$}(1);
                                                                                                                                                                                                  \path (0) edge [draw = white, opacity = 0, bend right = 30] node[pos= 0.4, inner sep = 1.5pt, font = \scriptsize, opacity = 1](a02){$a$}(2);
                                                                                                                                                                                                  \node[above right = 2ex and .0em of 0, inner sep = 1.5pt, font = \scriptsize, opacity = 1](I00){$\const{I}$};
                                                                                                                                                                                                  \node[above right = 2ex and .0em of 1, inner sep = 1.5pt, font = \scriptsize, opacity = 1](I11){$\const{I}$};
                                                                                                                                                                                                  \node[above right = 2ex and .0em of 2, inner sep = 1.5pt, font = \scriptsize, opacity = 1](I22){$\const{I}$};
                                                                                                                                                                                                  \path (0) edge [draw = white, opacity = 0, bend left = 65] node[pos= 0.5, inner sep = .5pt, font = \scriptsize, opacity = 1](I01){$\comnf{\const{I}}$}(1);
                                                                                                                                                                                                  \path (1) edge [draw = white, opacity = 0, bend left = 65] node[pos= 0.5, inner sep = .5pt, font = \scriptsize, opacity = 1](I12){$\comnf{\const{I}}$}(2);
                                                                                                                                                                                                  \path (0) edge [draw = white, opacity = 0, bend right = 45] node[pos= 0.6, inner sep = 1.5pt, font = \scriptsize, opacity = 1](I02){$\const{I}$}(2);
                                                                                                                                                                                                  \graph[use existing nodes, edges={color=black, pos = .5, earrow}, edge quotes={fill=white, inner sep=1pt,font= \scriptsize}]{
                                                                                                                                                                                                  0 --[] a01 ->[] 1;
                                                                                                                                                                                                  1 --[bend left = 15] a10 ->[bend left = 15] 0;
                                                                                                                                                                                                  1 --[] a12 ->[] 2;
                                                                                                                                                                                                  2 --[bend left = 15] a21 ->[bend left = 15] 1;
                                                                                                                                                                                                  0.south east <-[bend right = 10, line width = 1.2pt] a02 ->[bend right = 10, line width = 1.2pt] 2.south west;
                                                                                                                                                                                                  0 --[bend left = 10] I00 ->[bend left = 10] 0;
                                                                                                                                                                                                  1 --[bend left = 10] I11 ->[bend left = 10] 1;
                                                                                                                                                                                                  2 --[bend left = 10] I22 ->[bend left = 10] 2;
                                                                                                                                                                                                  0 <-[bend left = 10] I01 ->[bend left = 10] 1;
                                                                                                                                                                                                  1 <-[bend left = 10] I12 ->[bend left = 10] 2;
                                                                                                                                                                                                  0.south <-[bend right = 10] I02 ->[bend right = 10] 2.south;
                                                                                                                                                                                                  l -> 0; 2 -> r;
                                                                                                                                                                                                  };
                                                                                                                                                                                              \end{tikzpicture}
                                                                                                                                                                           &
        \textcolor{blue}{H}                                                                                                                                                & = \hspace{-.6em} \begin{tikzpicture}[baseline = -.5ex]
                                                                                                                                                                                                  \graph[grow right = 1.2cm, branch down = 6ex, nodes={mysmallnode, font = \scriptsize}]{
                                                                                                                                                                                                  {0/{$0$}[draw, circle]}-!-{1/{$1$}[draw, circle]}-!-{2/{$2$}[draw, circle]}
                                                                                                                                                                                                  };
                                                                                                                                                                                                  \node[left = .5em of 0](l){};
                                                                                                                                                                                                  \node[right = .5em of 2](r){};
                                                                                                                                                                                                  \node[above left = 2ex and .0em of 0, inner sep = 1.5pt, font = \scriptsize, opacity = 1](a00){$a$};
                                                                                                                                                                                                  \node[above left = 2ex and .0em of 1, inner sep = .5pt, font = \scriptsize, opacity = 1](a11){$\comnf{a}$};
                                                                                                                                                                                                  \node[above left = 2ex and .0em of 2, inner sep = 1.5pt, font = \scriptsize, opacity = 1](a22){$a$};
                                                                                                                                                                                                  \path (0) edge [draw = white, opacity = 0] node[pos= 0.5, inner sep = 1.5pt, font = \scriptsize, opacity = 1](a01){$a$}(1);
                                                                                                                                                                                                  \path (1) edge [draw = white, opacity = 0, bend left = 25] node[pos= 0.5, inner sep = .5pt, font = \scriptsize, opacity = 1](a10){$\comnf{a}$}(0);
                                                                                                                                                                                                  \path (1) edge [draw = white, opacity = 0] node[pos= 0.5, inner sep = .5pt, font = \scriptsize, opacity = 1](a12){$\comnf{a}$}(2);
                                                                                                                                                                                                  \path (2) edge [draw = white, opacity = 0, bend left = 25] node[pos= 0.5, inner sep = 1.5pt, font = \scriptsize, opacity = 1](a21){$a$}(1);
                                                                                                                                                                                                  \path (0) edge [draw = white, opacity = 0, bend right = 30] node[pos= 0.4, inner sep = 1.5pt, font = \scriptsize, opacity = 1](a02){$a$}(2);
                                                                                                                                                                                                  \node[above right = 2ex and .0em of 0, inner sep = 1.5pt, font = \scriptsize, opacity = 1](I00){$\const{I}$};
                                                                                                                                                                                                  \node[above right = 2ex and .0em of 1, inner sep = 1.5pt, font = \scriptsize, opacity = 1](I11){$\const{I}$};
                                                                                                                                                                                                  \node[above right = 2ex and .0em of 2, inner sep = 1.5pt, font = \scriptsize, opacity = 1](I22){$\const{I}$};
                                                                                                                                                                                                  \path (0) edge [draw = white, opacity = 0, bend left = 65] node[pos= 0.5, inner sep = .5pt, font = \scriptsize, opacity = 1](I01){$\comnf{\const{I}}$}(1);
                                                                                                                                                                                                  \path (1) edge [draw = white, opacity = 0, bend left = 65] node[pos= 0.5, inner sep = .5pt, font = \scriptsize, opacity = 1](I12){$\comnf{\const{I}}$}(2);
                                                                                                                                                                                                  \path (0) edge [draw = white, opacity = 0, bend right = 45] node[pos= 0.6, inner sep = 1.5pt, font = \scriptsize, opacity = 1](I02){$\const{I}$}(2);
                                                                                                                                                                                                  \graph[use existing nodes, edges={color=black, pos = .5, earrow}, edge quotes={fill=white, inner sep=1pt,font= \scriptsize}]{
                                                                                                                                                                                                  0 --[bend left = 10, line width = 1.2pt] a00 ->[bend left = 10, line width = 1.2pt] 0;
                                                                                                                                                                                                  1 --[bend left = 10, line width = 1.2pt] a11 ->[bend left = 10, line width = 1.2pt] 1;
                                                                                                                                                                                                  2 --[bend left = 10, line width = 1.2pt] a22 ->[bend left = 10, line width = 1.2pt] 2;
                                                                                                                                                                                                  0 --[] a01 ->[] 1;
                                                                                                                                                                                                  1 --[bend left = 15] a10 ->[bend left = 15] 0;
                                                                                                                                                                                                  1 --[] a12 ->[] 2;
                                                                                                                                                                                                  2 --[bend left = 15] a21 ->[bend left = 15] 1;
                                                                                                                                                                                                  0.south east <-[bend right = 10] a02 ->[bend right = 10] 2.south west;
                                                                                                                                                                                                  0 --[bend left = 10] I00 ->[bend left = 10] 0;
                                                                                                                                                                                                  1 --[bend left = 10] I11 ->[bend left = 10] 1;
                                                                                                                                                                                                  2 --[bend left = 10] I22 ->[bend left = 10] 2;
                                                                                                                                                                                                  0 <-[bend left = 10] I01 ->[bend left = 10] 1;
                                                                                                                                                                                                  1 <-[bend left = 10] I12 ->[bend left = 10] 2;
                                                                                                                                                                                                  0.south <-[bend right = 10] I02 ->[bend right = 10] 2.south;
                                                                                                                                                                                                  l -> 0; 2 -> r;
                                                                                                                                                                                                  };
                                                                                                                                                                                              \end{tikzpicture}
    \end{align*}
    First, we add the $a$-labeled edge for $\tuple{2, 1}$ because $\tuple{2, 1} \in (\const{I}^{G} \cdot a^{G} \cdot \const{I}^{G})$; similarly, we add the $\comnf{a}$-labeled edge for $\tuple{1, 0}$ (let $G'$ be the graph).
    Second, we consider adding an $a$- or $\comnf{a}$-labeled edge for $\tuple{0, 2}$.
    Then, because $\Con_{a}^{\automaton}(\set{\mathtt{A,C}}, \set{\mathtt{A,C}})$, we add an $a$-labeled edge for $\tuple{0, 2}$ (note that $\lnot \Con_{\comnf{a}}^{\automaton}(\set{\mathtt{A,C}}, \set{\mathtt{A,C}})$ because $\delta_{\widebreve{\comnf{a}}}^{\automaton}(\set{\mathtt{A, C}}) = \set{\mathtt{D}} \not\subseteq \set{\mathtt{A, C}}$); we also add an $a$-labeled edge for $\tuple{2, 0}$ because $\tuple{2, 0}, \tuple{0, 2} \in \const{I}^{G'}$ (let $G''$ be the graph).
    By adding the other edges similarly, the saturation $\textcolor{blue}{\graph[2]}$ of $G(w)$ can be obtained.
\end{example}
\begin{lemma}\label{lemma: Hintikka}
    For an NFA $\automaton$ over $\AutSIG$ and a word $w$ over $\WordSIG$,
    the following are equivalent:
    \begin{itemize}
        \item There is a saturation $\graph[2]$ of $G(w)$ such that $\graph[2]^{\Quo} \not\models w \le \automaton$.
        \item There is a saturable path $P$ for $\not\models_{\REL} w \le \automaton$.
    \end{itemize}
\end{lemma}
\begin{proof}
    $\Longleftarrow$:
    Let $P = \tuple{G, \set{U_{i}}_{i \in |G|}}$.
    Let $\graph[2]$ be the saturation of $G$ (\Cref{lemma: Hintikka saturate}).
    Then $\graph[2]^{\Quo} \not\models w \le \automaton[1]$ holds as follows.
    For $\graph[2]^{\Quo} \models w$:
    Because $G(w) \homo \graph[2] \homo \graph[2]^{\Quo}$ (cf.\ \Cref{lemma: graph characterization ECoR QuoFill}).
    For $\graph[2]^{\Quo} \not\models \automaton[1]$:
    Assume that $\graph[2]^{\Quo} \models \automaton[1]$.
    Then there is a word $\aterm[1]_1 \dots \aterm[1]_m \in \lang{\automaton[1]}$ such that $
        \graph(\aterm[1]_1 \dots \aterm[1]_m) \homo \graph[2]^{\Quo}$.
    By $\aterm[1]_1 \dots \aterm[1]_m \in \lang{\automaton[1]}$, there are $z_0, \dots, z_m \in |\automaton[1]|$ such that
    $\tuple{z_0, z_m} = \tuple{\lv^{\automaton[1]}, \rv^{\automaton[1]}}$ and for $k \in [m]$, $z_k \in \delta_{\aterm[1]_{k}}^{\automaton[1]}(z_{k-1})$.
    By $\graph(\aterm[1]_1 \dots \aterm[1]_m) \homo \graph[2]^{\Quo}$, there are $l_0, \dots, l_m \in |\graph[2]|$ such that
    $\tuple{l_0, l_{m}} = \tuple{\lv^{\graph[2]}, \rv^{\graph[2]}}$ and for $k \in [m]$, $\tuple{l_{k-1}, l_{k}} \in \aterm[1]_k^{\graph[2]}$.
    Then,
    \begin{sublemma*}
        For every $k \in [0, m]$, $z_k \in U_{l_{k}}$.
    \end{sublemma*}
    \begin{proof}
        By induction on $k$.
        Case $k = 0$:
        By $\src^{\automaton[1]} \in \tset[3]_{\lv^{\graph[2]}}$ \ref{definition: Hintikka: u src tgt}.
        Case $k \ge 1$:
        We have
        \begin{align*}
            z_k \in \delta_{\aterm[1]_k}^{\automaton[1]}(z_{k-1})
             & \subseteq \delta_{\aterm[1]_k}^{\automaton[1]}(\tset[3]_{l_{k-1}}) \tag{By $z_{k-1} \in \tset[3]_{l_{k-1}}$ (IH)}                              \\
             & \subseteq \tset[3]_{l_{k}}. \tag*{(By $\tuple{l_{k-1}, l_{k}} \in \aterm[1]_k^{\graph[2]}$ \ref{definition: Hintikka: consistent 2}) \qedhere}
        \end{align*}
    \end{proof}
    \noindent Since $z_m \in \tset[3]_{l_{m}}$ contradicts $\tgt^{\automaton[1]} \not\in \tset[3]_{\rv^{\graph[2]}}$ \ref{definition: Hintikka: u src tgt},
    $\graph[2]^{\Quo} \not\models \automaton[1]$.

    $\Longrightarrow$:
    We define $P = \tuple{G, \set{U_{i}}_{i \in |G|}}$ as follows:
    \begin{itemize}
        \item $G$ is the edge extension of $G(w)$ such that
              $a^{G} = a^{\graph[2]}$ for $a \in \set{\const{I}, \comnf{\const{I}}}$ and $a^{G} = a^{G(w)}$ for $a \in \SIG^{(-)}$;
        \item
              each set $U_{i} \subseteq |\automaton|$ is defined by:
              \begin{align*}
                  z \in U_{i} & \defiff \exists v \in \lang{{\automaton}[\bl, z]}.\  G(v) \homo {\graph[2]^{\Quo}}[\bl, \gquo{i}{\graph[2]}].
              \end{align*}
              Here, ${\automaton}[\bl, z]$ denotes the graph $\automaton$ in which $\rv^{\automaton}$ has been replaced with $z$ (similarly for $\graph[2]^{\Quo}[\bl, \gquo{i}{\graph[2]}]$).
    \end{itemize}
    Then $P$ is a saturable path for $\not\models_{\REL} w \le \automaton$, as follows.
    \begin{sublemma*}
        Let $a \in \SIG_{\const{I}}^{(-, \smile)}$.
        If $G(a) \homo {\graph[2]^{\Quo}}[[i]_{\graph[2]}, [j]_{\graph[2]}]$, then
        $\Con_{a}^{\automaton[1]}(U_{i}, U_{j})$.
        (Here, ${\graph[2]^{\Quo}}[[i]_{\graph[2]}, [j]_{\graph[2]}]$ denotes the graph ${\graph[2]^{\Quo}}$ in which $\lv^{{\graph[2]^{\Quo}}}$ and $\rv^{{\graph[2]^{\Quo}}}$ have been replaced with $[i]_{\graph[2]}$ and $[j]_{\graph[2]}$, respectively.)
    \end{sublemma*}
    \begin{proof}
        Let $z \in U_{i}$.
        By definition,
        there is $v \in \lang{\automaton[1][\bl, z]}$ such that $G(v) \homo {\graph[2]^{\Quo}}[\bl, \gquo{i}{\graph[2]}]$.
        Combining with $G(a) \homo {\graph[2]^{\Quo}}[\gquo{i}{\graph[2]}, \gquo{j}{\graph[2]}]$ yields $G(va) \homo {\graph[2]^{\Quo}}[\bl, \gquo{j}{\graph[2]}]$.
        Thus, for every $z' \in \delta_{a}^{\automaton[1]}(z)$, $z' \in U_{j}$.
        Hence $\delta_{a}^{\automaton[1]}(U_{i}) \subseteq U_{j}$.
        Similarly by $G(\widebreve{a}) \homo {\graph[2]^{\Quo}}[\gquo{j}{\graph[2]}, \gquo{i}{\graph[2]}]$, $\delta_{\widebreve{a}}^{\automaton[1]}(U_{j}) \subseteq U_{i}$.
    \end{proof}
    \noindent For \ref{definition: Hintikka: u src tgt}:
    $\lv^{\automaton[1]} \in U_{\lv^{G}}$ is shown by considering $v = \const{I}$.
    $\rv^{\automaton[1]} \not\in U_{\rv^{G}}$ is because $\graph[2]^{\Quo} \not\models \automaton[1]$.
    For \ref{definition: Hintikka: consistent}:
    Let $\tuple{i, j} \in a^{G}$.
    Because $G(a) \homo {\graph[2]^{\Quo}}[\gquo{i}{\graph[2]},\gquo{j}{\graph[2]}]$ (since $\graph[2]$ is a saturation of $G$),
    we have $\Con_{a}^{\automaton[1]}(U_{i}, U_{j})$ (the sub-lemma).
    For \ref{definition: Hintikka: saturate}:
    For every $a \in \SIG$ and $i, j \in |G|$,
    because either $G(a) \homo {\graph[2]^{\Quo}}[[i]_{\graph[2]}, [j]_{\graph[2]}]$ or $G(\comnf{a}) \homo {\graph[2]^{\Quo}}[[i]_{\graph[2]}, [j]_{\graph[2]}]$ always holds,
    we have $\Con_{a}^{\automaton[1]}(U_{i}, U_{j}) \lor \Con_{\comnf{a}}^{\automaton[1]}(U_{i}, U_{j})$ (the sub-lemma).
\end{proof}

\begin{theorem}\label{theorem: Hintikka}
    For two NFAs $\automaton[1], \automaton[2]$ over $\AutSIG$, TFAE:
    \begin{itemize}
        \item $\not\models_{\REL} \automaton[1] \le \automaton[2]$.
        \item $\exists w \in \lang{\automaton[1]}.\ $ there is a saturable path for $\not\models_{\REL} w \le \automaton[2]$.
    \end{itemize}
\end{theorem}
\begin{proof}
    We have
    \begin{align*}
         & \not\models_{\REL} \automaton[1] \le \automaton[2]                                                                                                                                                                                                   \\
         & \iff \exists w \in \lang{\automaton[1]}. \exists \graph[2] \in \Fill(G(w)).\ \graph[2]^{\Quo} \not\models \automaton[1] \le \automaton[2] \tag{\Cref{lemma: bounded model property for ECoRTC}}                                                      \\
         & \iff \exists w \in \lang{\automaton[1]}. \exists \graph[2] \in \Fill(G(w)).\ \graph[2]^{\Quo} \not\models w \le \automaton[2] \tag{Because $\graph[2]^{\Quo} \models w$ and $\graph[2]^{\Quo} \models \automaton[1]$ by $\graph[2] \in \Fill(G(w))$} \\
         & \iff \exists w \in \lang{\automaton[1]}.\ \mbox{there is a saturable path for $\not\models_{\REL} w \le \automaton[2]$}. \tag*{(\Cref{lemma: Hintikka}) \qedhere}
    \end{align*}
\end{proof}

\subsection{Exponential-size model property}
The characterization by saturable paths (\Cref{theorem: Hintikka}) gives another bounded model property for $\ExKA$ terms (cf.\ \Cref{lemma: bounded model property for ECoRTC}) as follows.
The following proof is an analogy of the well-known pumping lemma from automata theory.
\begin{lemma}[Exponential-size model property]\label{lemma: Hintikka bound}
    For every NFAs $\automaton[1], \automaton[2]$ over $\AutSIG$,
    if ${} \not\models_{\REL} \automaton[1] \le \automaton[2]$, then
    there is a $2$-pointed structure $\ppstruc$ of size $\#(|\ppstruc|) \le \#(|\automaton[1]|) \times 2^{\#(|\automaton[2]|)}$
    such that $\ppstruc \not\models_{\REL} \automaton[1] \le \automaton[2]$.
\end{lemma}
\begin{proof}
    By \Cref{theorem: Hintikka}, there are a word $w = a_1 \dots a_n \in \lang{\automaton[1]}$ and a saturable path $P = \tuple{G, \set{U_{i}}_{i \in [0, n]}}$ for $\not\models_{\REL} w \le \automaton[2]$.
    Without loss of generality, we can assume that $n$ is the minimum among such words.
    Since $w \in \lang{\automaton[1]}$, let $\set{s_{i}}_{i \in [0, n]}$ be such that
    $\tuple{s_0, s_{n}} = \tuple{\src^{\automaton[1]}, \tgt^{\automaton[1]}}$ and $s_x \in \delta_{a_i}^{\automaton[1]}(s_{i-1})$ for $i \in [n]$.
    Assume that $n+1 > \#(|\automaton[1]|) \times 2^{\#(|\automaton[2]|)}$.
    By the pigeonhole principle, there are $0 \le x < y \le n$ s.t.\ $\tuple{\term[2]_x, U_x} = \tuple{\term[2]_y, U_y}$.
    Let $w' \defeq a_1 \dots a_x a_{y + 1} \dots a_n$.
    Let $P'$ be the $P$ in which
    the source of the edge for $\tuple{y, y+1} \in a_{y+1}^{G}$ has been replaced with $x$
    and the vertices between $x+1$ and $y$ are removed:
    \begin{align*}
        P  & = \begin{tikzpicture}[baseline = -.5ex]
                   \node[mynode, draw, circle](L1){};
                   \node[mynode, draw, circle, right = 1.8cm of L1, line width = 1.pt](L2){};
                   \node[mynode, draw = gray!80, circle, right = 1.8cm of L2, text = gray!80](L3){};
                   \node[mynode, draw, circle, right = 1cm of L3, line width = 1.pt](L4){};
                   \node[mynode, draw, circle, right = 1.8cm of L4](L5){};
                   \node[left = 4pt of L1](Ll){} edge[earrow, ->] (L1);
                   \node[right = 4pt of L5](Lr){}; \path (L5) edge[earrow, ->] (Lr);
                   \path (L1) edge [draw = white, opacity = 0] node[pos= 0.5, inner sep = .5pt, font = \scriptsize, opacity = 1](Ltau1){$a_1 \dots a_x$}(L2);
                   \path (L2) edge [draw = white, opacity = 0] node[pos= 0.5, inner sep = 1.pt, font = \scriptsize, opacity = 1, text = gray!80](Ltau2){$a_{x+1} \dots a_{y}$}(L3);
                   \path (L3) edge [draw = white, opacity = 0] node[pos= 0.5, inner sep = 1.pt, font = \scriptsize, opacity = 1](Ltau3){$a_{y+1}$}(L4);
                   \path (L4) edge [draw = white, opacity = 0] node[pos= 0.5, inner sep = 1.pt, font = \scriptsize, opacity = 1](Ltau4){$a_{y+2} \dots a_{n}$}(L5);
                   \graph[use existing nodes, edges={color=black, pos = .5, earrow}, edge quotes={fill=white, inner sep=1pt,font= \scriptsize}]{
                   L1 --[] {Ltau1} ->[] L2;
                   L2 --[ color = gray!80] {Ltau2} ->[, color = gray!80] L3;
                   L3 --[] {Ltau3} ->[] L4;
                   L4 --[] {Ltau4} ->[] L5;
                   };
                   \node[below = .5ex of L1, font = \scriptsize]{$U_{0}$};
                   \node[below = .5ex of Ltau1, font = \scriptsize]{$\dots$};
                   \node[below = .5ex of L2, font = \scriptsize]{$U_{x}$};
                   \node[below = .5ex of L3, font = \scriptsize, text = gray!80]{$U_{y}$};
                   \node[below = .5ex of L4, font = \scriptsize]{$U_{y+1}$};
                   \node[below = .5ex of Ltau4, font = \scriptsize]{$\dots$};
                   \node[below = .5ex of L5, font = \scriptsize]{$U_{n}$};
               \end{tikzpicture} \\
        P' & =\begin{tikzpicture}[baseline = -.5ex]
                  \node[mynode, draw, circle](L1){};
                  \node[mynode, draw, circle, right = 1.8cm of L1, line width = 1.pt](L2){};
                  \node[mynode, draw, circle, right = 1cm of L2, line width = 1.pt](L4){};
                  \node[mynode, draw, circle, right = 1.8cm of L4](L5){};
                  \node[left = 4pt of L1](Ll){} edge[earrow, ->] (L1);
                  \node[right = 4pt of L5](Lr){}; \path (L5) edge[earrow, ->] (Lr);
                  \path (L1) edge [draw = white, opacity = 0] node[pos= 0.5, inner sep = .5pt, font = \scriptsize, opacity = 1](Ltau1){$a_1 \dots a_x$}(L2);
                  \path (L2) edge [draw = white, opacity = 0] node[pos= 0.5, inner sep = 1.pt, font = \scriptsize, opacity = 1](Ltau3){$a_{y+1}$}(L4);
                  \path (L4) edge [draw = white, opacity = 0] node[pos= 0.5, inner sep = 1.pt, font = \scriptsize, opacity = 1](Ltau4){$a_{y+2} \dots a_{n}$}(L5);
                  \graph[use existing nodes, edges={color=black, pos = .5, earrow}, edge quotes={fill=white, inner sep=1pt,font= \scriptsize}]{
                  L1 --[] {Ltau1} ->[] L2;
                  L2 --[] {Ltau3} ->[] L4;
                  L4 --[] {Ltau4} ->[] L5;
                  };
                  \node[below = .5ex of L1, font = \scriptsize]{$U_{0}$};
                  \node[below = .5ex of Ltau1, font = \scriptsize]{$\dots$};
                  \node[below = .5ex of L2, font = \scriptsize]{$U_{x}$};
                  \node[below = .5ex of L4, font = \scriptsize]{$U_{y+1}$};
                  \node[below = .5ex of Ltau4, font = \scriptsize]{$\dots$};
                  \node[below = .5ex of L5, font = \scriptsize]{$U_{n}$};
              \end{tikzpicture}
    \end{align*}
    ($\const{I}$- or $\comnf{\const{I}}$-labeled edges and some intermediate vertices are omitted, for simplicity.)
    Here, when $\tuple{x, y+1} \in \const{I}^{G}$ and $a_{y+1} = \comnf{\const{I}}$, the graph of $P'$ is not consistent; so, we replace the label $\const{I}$ with the label $\comnf{\const{I}}$ for every pair $\tuple{i, j}$ s.t. ${i \neq j} \land \Con_{\comnf{\const{I}}}^{\automaton[2]}(U_{i}, U_{j})$.
    ($i \neq j$ is for the reflexivity of the relation of $\const{I}$ and $\Con_{\comnf{\const{I}}}^{\automaton[2]}(U_{i}, U_{j})$ is for preserving \ref{definition: Hintikka: consistent}).
    Then, $w' \in [\automaton]$ holds by $s_x = s_y$ and $P'$ is an saturable path for $\not\models_{\REL} w' \le \automaton[2]$ because
    each condition for $P'$ is shown by that for $P$ (with $U_{x} = U_{y}$) and that $P'$ is (almost) a ``subgraph'' of $P$ (see \ifiscameraready the full version \cite{nakamuraExistentialCalculusRelations2023}\else \Cref{section: lemma: Hintikka bound}\fi, for more details).
    However, this contradicts that $n$ is the minimum.
    Thus, $\#(|G|) = n + 1 \le \#(|\automaton[1]|) \times 2^{\#(|\automaton[2]|)}$.
    Finally, $\ppstruc = \graph[2]^{\Quo}$ is the desired $2$-pointed structure, where $\graph[2]$ is the saturation of $G$ obtained from \Cref{lemma: Hintikka saturate}.
\end{proof}

\begin{theorem}\label{theorem: KACC complexity}
    The equational theory of $\ExKA$ terms (w.r.t.\ binary relations) is decidable in coNEXP.
\end{theorem}
\begin{proof}
    Similarly for \cref{lemma: PCoRTC upper bound},
    it suffices to show that the following problem is in NEXP: given $\ExKA$ terms $\term[1], \term[2]$, does $\not\models_{\REL} \term[1] \le \term[2]$ hold?
    By \Cref{lemma: Hintikka bound} (with \Cref{proposition: Thompson}), we can give the following algorithm:
    \begin{enumerate}
        \item Take a $2$-pointed structure $\ppstruc$ of size $\#(|\ppstruc|) \le \#(|\automaton_{\term[1]}|) \times 2^{\#(|\automaton_{\term[2]}|)}$, non-deterministically.
              Here, $\automaton_{\term[1]}$ and $\automaton_{\term[2]}$ are the NFAs obtained from $\term[1]$ and $\term[2]$, respectively (\Cref{definition: NFA def}).
        \item Return $\const{true}$, if $\ppstruc \not\models \term[1] \le \term[2]$; $\const{false}$, otherwise.
    \end{enumerate}
    Then $\not\models_{\REL} \term[1] \le \term[2]$, if some execution returns $\const{true}$; $\models_{\REL} \term[1] \le \term[2]$, otherwise.
    Here, $\ppstruc \not\models \term[1] \le \term[2]$ can be decided in exponential time (\Cref{proposition: model checking}).
\end{proof}

\newcommand{\AH}[1]{#1^{\mathcal{S}}}
\subsection{From saturable paths to word automata}\label{section: word automata}
For some cases, for an NFA $\automaton$, we can construct an NFA $\AH{\automaton}$ such that:
for every word $w$ over $\WordSIG$, TFAE:
\begin{itemize}
    \item $w \in \lang{\AH{\automaton}}$;
    \item there is a saturable path for $\not\models_{\REL} w \le \automaton$.
\end{itemize}
To this end, first, let
\[\fml(\mathcal{U}, U) \defeq \bigwedge \begin{cases}
        U \times (|\automaton| \setminus U) \times U \times (|\automaton| \setminus U) \subseteq \mathcal{U} \\
        \forall \tuple{t_1, t_2, t_3, t_4} \in \mathcal{U}. \forall a \in \SIG_{\const{I}}.                  \\
        \quad  \bigvee \begin{cases}
                           \delta_a^{\automaton}(t_1) \subseteq U \land  t_2 \not\in \delta_{\widebreve{a}}^{\automaton}(U) \\
                           \delta_{\comnf{a}}^{\automaton}(t_3) \subseteq U \land t_4 \not\in \delta_{\widebreve{\comnf{a}}}^{\automaton}(U)
                       \end{cases}
    \end{cases}\]
and we show the following lemma:
\begin{lemma}\label{lemma: key formula transformation}
    Let $\automaton$ be an NFA over $\AutSIG$ and $w = a_1 \dots a_n$ be a word over $\WordSIG$.
    Recall the formula of \ref{definition: Hintikka: saturate}:
    \begin{align*}
         & \forall \tuple{i, j} \in [0, n]^2.\forall a \in \SIG_{\const{I}}.  \bigvee \left\{\begin{aligned}
                                                                                                 \delta_{a}^{\automaton}(U_{i}) \subseteq U_{j} \land \delta_{\widebreve{a}}^{\automaton}(U_{j}) \subseteq U_{i} \\
                                                                                                 \delta_{\comnf{a}}^{\automaton}(U_{i}) \subseteq U_{j} \land \delta_{\widebreve{\comnf{a}}}^{\automaton}(U_{j}) \subseteq U_{i}
                                                                                             \end{aligned}\right.
    \end{align*}
    This formula is equivalent to the following formula:\footnote{This transformation is also used for the automata construction in \cite{Lutz2005_neg} (roughly speaking, the $\mathcal{U}$ corresponds to the ``$P$'' in \cite{Lutz2005_neg}), but is a bit more complicated due to converse.}
    \begin{align*}
         & \exists \mathcal{U} \subseteq |\automaton|^4. \forall i \in [0, n].\ \fml(\mathcal{U}, U_{i}).
    \end{align*}
\end{lemma}
\begin{proof}
    We have
    \begin{align*}
         & \delta_{a}^{\automaton}(U_{i}) \subseteq U_{j} \ \land \  \delta_{\widebreve{a}}^{\automaton}(U_{j}) \subseteq U_{i}                                                                             \\
         & \Leftrightarrow \delta_{a}^{\automaton}(U_{i}) \subseteq U_{j} \ \land \  |\automaton| \setminus U_{i} \subseteq |\automaton|  \setminus \delta_{\widebreve{a}}^{\automaton}(U_{j})              \\
         & \Leftrightarrow (\forall t_1 \in U_{i}. \delta_a^{\automaton}(t_1) \subseteq U_{j}) \land (\forall t_2 \in |\automaton| \setminus U_{i}. t_2 \not\in \delta_{\widebreve{a}}^{\automaton}(U_{j})) \\
         & \Leftrightarrow \forall t_1 \in U_{i}. \forall t_2 \in |\automaton| \setminus U_{i}.\  \delta_a^{\automaton}(t_1) \subseteq U_{j} \land t_2 \not\in \delta_{\widebreve{a}}^{\automaton}(U_{j}).
    \end{align*}
    Thus by letting
    \begin{align*}
        \xi(U) & \defeq \bigvee \left\{\begin{aligned}
                                           \delta_a^{\automaton}(t_1) \subseteq U \land  t_2 \not\in \delta_{\widebreve{a}}^{\automaton}(U) \\
                                           \delta_{\comnf{a}}^{\automaton}(t_3) \subseteq U \land t_4 \not\in \delta_{\widebreve{\comnf{a}}}^{\automaton}(U)
                                       \end{aligned}\right.; \\
        \nu(U) & \defeq U \times (|\automaton| \setminus U) \times U \times (|\automaton| \setminus U),
    \end{align*}
    we have
    \begin{align*}
                             & \forall \tuple{i, j} \in [0, n]^2.\forall a \in \SIG_{\const{I}}.  \bigvee \left\{\begin{aligned}
                                                                                                                     \delta_{a}^{\automaton}(U_{i}) \subseteq U_{j} \land \delta_{\widebreve{a}}^{\automaton}(U_{j}) \subseteq U_{i} \\
                                                                                                                     \delta_{\comnf{a}}^{\automaton}(U_{i}) \subseteq U_{j} \land \delta_{\widebreve{\comnf{a}}}^{\automaton}(U_{j}) \subseteq U_{i}
                                                                                                                 \end{aligned}\right. \\
                             & \Leftrightarrow \forall \tuple{i, j} \in [0, n]^2. \forall a \in \SIG_{\const{I}}. \forall \tuple{t_1, t_2, t_3, t_4} \in \nu(U_{i}).\  \xi(U_{j})                                                                               \\
                             & \Leftrightarrow \forall \tuple{t_1, t_2, t_3, t_4} \in \bigcup_{i = 0}^{n} \nu(U_{i}).\   \forall a \in \SIG_{\const{I}}. \forall i \in [0, n].\  \xi(U_{i})                                                                     \\
                             & \Leftrightarrow \exists \mathcal{U}  \subseteq |\automaton|^4.\  \bigcup_{i = 0}^{n} \nu(U_{i}) \subseteq \mathcal{U} \quad \land {}                                                                                             \\
        \label{tag: diamond} & \qquad \quad \forall \tuple{t_1, t_2, t_3, t_4} \in \mathcal{U}. \forall a \in \SIG_{\const{I}}.  \forall i \in [0, n].\  \xi(U_{i}) \tag{$\diamondsuit$}
        \\
                             & \Leftrightarrow  \exists \mathcal{U}  \subseteq |\automaton|^4. \forall i \in [0, n].\  \nu(U_{i}) \subseteq \mathcal{U} \quad \land {}                                                                                          \\
                             & \qquad \quad \forall \tuple{t_1, t_2, t_3, t_4} \in \mathcal{U}. \forall a \in \SIG_{\const{I}}.\  \xi(U_{i})                                                                                                                    \\
                             & \Leftrightarrow  \exists \mathcal{U} \subseteq |\automaton|^4. \forall i \in [0, n].\ \fml(\mathcal{U}, U_{i}).
    \end{align*}
    Here, for (\ref{tag: diamond}),
    $\Longrightarrow$:
    By letting $\mathcal{U} = \bigcup_{i = 0}^{n} \nu(U_{i})$.
    $\Longleftarrow$:
    Because the formula $\forall \tuple{t_1, t_2, t_3, t_4} \in \mathcal{U}'. \forall a \in \SIG_{\const{I}}.  \forall i \in [0, n].\  \xi(U_{i})$ holds for any $\mathcal{U}' \subseteq \mathcal{U}$
    and $\bigcup_{i = 0}^{n} \nu(U_{i}) \subseteq \mathcal{U}$.
\end{proof}

By using the formula of \Cref{lemma: key formula transformation} for \ref{definition: Hintikka: saturate},
we can check the condition \ref{definition: Hintikka: saturate} \emph{pointwisely} (without considering pairs of $[0,n]^2$).
Using this, we give the following NFAs construction:
\begin{definition}\label{definition: AH}
    For an NFA $\automaton$ over $\AutSIG$, let $\AH{\automaton}$ be the NFA over $\AutSIG$ defined by:
    \begin{itemize}
        \item $|\AH{\automaton}| \defeq \set{\blacktriangleright, \blacktriangleleft} \cup \set{\tuple{\mathcal{U}, U} \in \wp(|\automaton|^4) \times \wp(|\automaton|) \mid \fml(\mathcal{U}, \tset[3]) \land \delta_{\const{I}}^{\automaton}(\tset[3]) \subseteq \tset[3]}$;
        \item
              $\const{I}^{\AH{\automaton}}$ is the minimal set such that
              \begin{itemize}
                  \item for all $\tuple{\mathcal{U}, U} \in |\AH{\automaton}|$ s.t.\  $\src^{\automaton} \in U$, $\tuple{\blacktriangleright, \tuple{\mathcal{U}, U}} \in \const{I}^{\AH{\automaton}}$;
                  \item  for all $\tuple{\mathcal{U}, U} \in |\AH{\automaton}|$ s.t.\  $\tgt^{\automaton} \not\in U$, $\tuple{\tuple{\mathcal{U}, \tset[3]}, \blacktriangleleft} \in  \const{I}^{\AH{\automaton}}$;
              \end{itemize}
        \item for each $a \in \AutSIG$, $a^{\AH{\automaton}}$ is the minimal set such that for every $\tuple{\mathcal{U}, U}, \tuple{\mathcal{U}, U'} \in |\AH{\automaton}|$ s.t.\  $\Con_{a}^{\automaton}(\tset[3], \tset[3]')$,
              $\tuple{\tuple{\mathcal{U}, \tset[3]}, \tuple{\mathcal{U}, \tset[3]'}} \in a^{\AH{\automaton}}$ holds;
        \item $\src^{\AH{\automaton}} = {\blacktriangleright}$ and $\tgt^{\AH{\automaton}} = {\blacktriangleleft}$.
    \end{itemize}
    (Here, $\blacktriangleright$ and $\blacktriangleleft$ are two fresh symbols.
    $\mathcal{U}$ is introduced for \ref{definition: Hintikka: saturate}, cf.\ \Cref{lemma: key formula transformation}.
    Note that $\mathcal{U}$ is invariant in transitions.)
\end{definition}
For example, when $|\automaton[1]| = \set{\mathtt{A}, \mathtt{B}}$ and
$\tuple{\lv^{\automaton[1]}, \rv^{\automaton[1]}} = \tuple{\mathtt{A},\mathtt{B}}$,
the NFA $\AH{\automaton[1]}$ is of the following form,
where the existence of each dashed state $\tuple{\mathcal{U}_1, U}$ depends on whether $\fml(\mathcal{U}_1, U) \land \delta_{\const{I}}^{\automaton[1]}(U) \subseteq U$ holds
and the existence of the $a$-labeled edge on each dashed edge from $\tuple{\mathcal{U}_1, U}$ to $\tuple{\mathcal{U}_1, U'}$ depend on whether $\Con_{a}^{\automaton[1]}(U, U')$ holds:
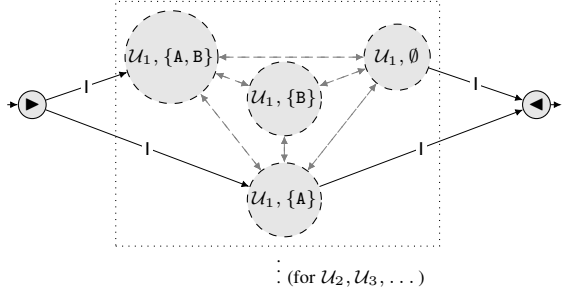
\begin{align*}
    \begin{tikzpicture}[baseline = -.5ex]
        \graph[grow right = 1.5cm, branch down = 4.5ex, nodes={ font = \scriptsize}]{
        {,/,1/{$\blacktriangleright$}[mysmallnode, draw, circle, yshift = -1ex, xshift = -1em]}
        -!- {/,u12/{$\mathcal{U}_1,\set{\mathtt{A},\mathtt{B}}$}[mysmallnode, draw, circle, yshift=3ex, minimum width=2.5em ,dashed],/, }
        -!- {/, u2/{$\mathcal{U}_1,\set{\mathtt{B}}$}[mysmallnode, draw, circle, yshift = -0.5ex, minimum width=2.5em ,dashed],/,  u1/{$\mathcal{U}_1,\set{\mathtt{A}}$}[mysmallnode, draw, circle, yshift=0ex, minimum width=2.5em ,dashed],/{$\vdots$ (for $\mathcal{U}_2, \mathcal{U}_3, \dots$)}[yshift = -1ex, xshift = 2.5em]}
        -!- {/,u0/{$\mathcal{U}_1,\emptyset$}[mysmallnode, draw, circle, yshift=3ex, minimum width=2.5em ,dashed],/}
        -!- {,/,4/{$\blacktriangleleft$}[mysmallnode, draw, circle, yshift = -1ex, xshift = 1em]}
        };
        \node[left = 4pt of 1](l){} edge[earrow, ->] (1);
        \node[right = 4pt of 4](r){}; \path (4) edge[earrow, ->] (r);
        \graph[use existing nodes, edges={color=black, pos = .5, earrow}, edge quotes={fill=white, inner sep=1pt,font= \scriptsize}]{
        1 ->["$\const{I}$"] {u1,u12};
        {u1,u0} ->["$\const{I}$"] {4};
        {u0,u1,u2,u12} <->[color = gray, dashed, complete bipartite] {u0,u1,u2,u12};
        };
        \node[fit=(u0)(u1)(u12)(u2), draw, dotted]{};
    \end{tikzpicture}.
\end{align*}

Using this transformation, $\AH{\automaton}$ satisfies the following:
\begin{lemma}[Completeness (of $\AH{\automaton}$)]\label{lemma: automaton construction completeness}
    For every NFA $\automaton$ over $\AutSIG$ and word $w$ over $\WordSIG$,
    we have:

    $w \in \lang{\AH{\automaton}}$ $\Longleftarrow$ there is a saturable path for $\not\models_{\REL} w \le \automaton$.
\end{lemma}
\begin{proof}
    Let $w = a_1 \dots a_n$ and $P = \tuple{G, \set{\tset[3]_{i}}_{i \in [0, n]}}$ be a saturable path for $\not\models_{\REL} w \le \automaton$.
    Let $\mathcal{U} \defeq \bigcup_{i \in [0,n]} \tset[3]_i \times (|\automaton| \setminus \tset[3]_i) \times \tset[3]_i \times (|\automaton| \setminus \tset[3]_i)$.
    Then we have
    \begin{itemize}
        \item $\tuple{{\blacktriangleright}, \tuple{\mathcal{U}, \tset[3]_{0}}} \in \delta_{\const{I}}^{\AH{\automaton}}$ (by \ref{definition: Hintikka: u src tgt});
        \item for all $i \in [n]$, $\tuple{\tuple{\mathcal{U}, \tset[3]_{i-1}}, \tuple{\mathcal{U}, \tset[3]_{i}}} \in \delta_{a_i}^{\AH{\automaton}}$ (by \ref{definition: Hintikka: consistent});
        \item $\tuple{\tuple{\mathcal{U}, \tset[3]_{n}}, {\blacktriangleleft}} \in \delta_{\const{I}}^{\AH{\automaton}}$ (by \ref{definition: Hintikka: u src tgt}).
    \end{itemize}
    Here, for $i \in [0,n]$, $\tuple{\mathcal{U}, \tset[3]_i} \in |\AH{\automaton}|$ holds, because
    $\fml(\mathcal{U}, U_i)$ holds by \ref{definition: Hintikka: saturate} with \Cref{lemma: key formula transformation}
    and $\delta_{\const{I}}^{\automaton}(U_i) \subseteq U_i$ holds by \ref{definition: Hintikka: consistent}.
    Hence, $w \in \lang{\AH{\automaton}}$.
\end{proof}
\begin{lemma}[Soundness for the $\const{I}^{-}$-free fragment]\label{lemma: automaton construction soundness: I-}
    For every NFA $\automaton$ over $\AutSIG$ and word $w$ over $\WordSIG$,
    if $\automaton$ does not contain $\const{I}^{-}$, then we have:

    $w \in \lang{\AH{\automaton}}$ $\Longrightarrow$ there is a saturable path for $\models_{\REL} w \le \automaton$.
\end{lemma}
\begin{proof}
    Let $w = a_1 \dots a_n$.
    By the form of $\AH{\automaton}$, there are $\mathcal{U}$, $\tset[3]_0$, $\dots$, and $\tset[3]_n$ such that
    \begin{itemize}
        \item $\tuple{{\blacktriangleright}, \tuple{\mathcal{U}, \tset[3]_{0}}} \in \delta_{\const{I}}^{\AH{\automaton}}$;
        \item for every $i \in [n]$, $\tuple{\tuple{\mathcal{U}, \tset[3]_{i-1}}, \tuple{\mathcal{U}, \tset[3]_{i}}} \in \delta_{a_i}^{\AH{\automaton}}$;
        \item $\tuple{\tuple{\mathcal{U}, \tset[3]_{n}}, {\blacktriangleleft}} \in \delta_{\const{I}}^{\AH{\automaton}}$.
    \end{itemize}
    Let $\graph[2]$ be the $\const{I}$-saturation of $G(w)$ such that
    \[\mbox{$\const{I}^{\graph[2]}$ is the identity relation.}\]
    $\graph[2]$ is consistent, because $G(w)$ is consistent and $\const{I}^{\graph[2]}$ is the identity relation.
    $\graph[2]$ is an edge-extension of $G(w)$, because
    $\const{I}^{\graph[2]} \supseteq \const{I}^{G(w)} = \emptyset$ and $\comnf{\const{I}}^{\graph[2]} = |\graph[2]|^2 \setminus \const{I}^{\graph[2]} \supseteq \comnf{\const{I}}^{G(w)}$.
    Hence $\graph[2]$ is indeed an $\const{I}$-saturation of $G(w)$.

    Then $P = \tuple{\graph[2], \set{U_{i}}_{i \in [0, n]}}$ is a saturable path for ${} \not\models_{\REL} w \le \automaton$ as follows.
    For \ref{definition: Hintikka: u src tgt}: By the definition of $\const{I}^{\AH{\automaton}}$, $\src^{\automaton} \in U_0$ and $\tgt^{\automaton} \not\in U_n$.
    For \ref{definition: Hintikka: consistent} for $a \in \SIG^{(-)}$: By the definition of $a_{i}^{\AH{\automaton}}$,
    $\Con_{a_i}^{\automaton}(U_{i-1}, U_i)$.
    (Note that the other edges do not exist.)
    For \ref{definition: Hintikka: consistent} for $a = \const{I}$:
    Because $\const{I}^{\graph[2]}$ is the identity relation
    and $U_{i} \in \delta_{\const{I}}^{\automaton}(U_i)$ for every $i \in [0, n]$ (by the definition of $|\AH{\automaton}|$).
    For \ref{definition: Hintikka: consistent} for $a = \comnf{\const{I}}$:
    Because $\comnf{\const{I}}$ does not occur in $\automaton$, $\delta_{\comnf{\const{I}}}^{\automaton}(U_i) = \emptyset \subseteq U_j$ always holds for every $i, j \in [0, n]$.
    Hence $\Con_{\comnf{\const{I}}}^{\automaton[1]}(U_{i}, U_{j})$.
    For \ref{definition: Hintikka: saturate}: Because $\mathcal{U}$ satisfies $\fml(\mathcal{U}, \tset[3]_i)$ for every $i \in [0, n]$ (by the definition of $|\AH{\automaton}|$), with \Cref{lemma: key formula transformation}.
\end{proof}
\begin{theorem}\label{theorem: automata KACC without I-}
    For every NFAs $\automaton[1], \automaton[2]$, over $\AutSIG$,
    if $\automaton[2]$ does not contain $\const{I}^{-}$, then we have
    \[\not\models_{\REL} \automaton[1] \le \automaton[2] \quad \iff \quad \lang{\automaton[1]} \cap \lang{\AH{\automaton[2]}} \neq \emptyset.\]
\end{theorem}
\begin{proof}
    We have
    \begin{align*}
         & \not\models_{\REL} \automaton[1] \le \automaton[2]                                                                                                                            \\
         & \iff \exists w \in \lang{\automaton[1]}.\ \mbox{there is a saturable path for $\not\models_{\REL} w \le \automaton[2]$}. \tag{\Cref{theorem: Hintikka}}                       \\
         & \iff \exists w \in \lang{\automaton[1]}.\ w \in \lang{\AH{\automaton[2]}} \tag{\Cref{lemma: automaton construction completeness,lemma: automaton construction soundness: I-}} \\
         & \iff \lang{\automaton[1]} \cap \lang{\AH{\automaton[2]}} \neq \emptyset. \tag*{\qedhere}
    \end{align*}
\end{proof}
\begin{theorem}\label{theorem: KACC without I-}
    The equational theory of $\ExKA$ terms without $\const{I}^{-}$ (w.r.t.\ binary relations) is PSPACE-complete.
\end{theorem}
\begin{proof}
    For hardness:
    By \Cref{proposition: PCoR hardness 2}, as the term class subsumes $\Termset_{\set{\cdot, \cup, \bl^{*}}}$.
    For upper bound:
    Similar to \cref{lemma: PCoRTC upper bound}, with (co-)NPSPACE = PSPACE (Savitch's theorem \cite{Savitch1970}), it suffices to show that the following problem is in NPSPACE:
    given $\ExKA$ terms $\term[1], \term[2]$ without $\const{I}^{-}$, does $\not\models_{\REL} \term[1] \le \term[2]$ hold?
    By \Cref{theorem: automata KACC without I-} (with \Cref{proposition: Thompson}),
    we can reduce this problem into the emptiness problem of NFAs (of size exponential in the input).
    Therefore, by using a standard on-the-fly algorithm for the non-emptiness problem of NFAs (which is essentially the graph reachability problem),
    we can give a non-deterministic polynomial space algorithm.
\end{proof}

\subsection{Remark on the case of full $\ExKA$ terms}
Additionally, we remark that the soundness also holds for $\ExKA$ terms without complements of term variables.
\begin{lemma}[Soundness for the $a^-$-free fragment, cf.\ \Cref{lemma: automaton construction soundness: I-}]\label{lemma: automaton construction soundness: a-}
    For every NFA $\automaton$ over $\AutSIG$ and word $w$ over $\WordSIG$,
    if $w$ does not contain $a^{-}$ for any $a \in \SIG$, then we have:

    $w \in \lang{\AH{\automaton}}$ $\Longrightarrow$ there is a saturable path for $\not\models_{\REL} w \le \automaton$.
\end{lemma}
\begin{proof}
    Let $w = a_1 \dots a_n$.
    By the form of $\AH{\automaton}$, there are $\mathcal{U}$, $\tset[3]_0$, $\dots$, and $\tset[3]_n$ such that
    \begin{itemize}
        \item $\tuple{{\blacktriangleright}, \tuple{\mathcal{U}, \tset[3]_{0}}} \in \delta_{\const{I}}^{\AH{\automaton}}$;
        \item for every $i \in [n]$, $\tuple{\tuple{\mathcal{U}, \tset[3]_{i-1}}, \tuple{\mathcal{U}, \tset[3]_{i}}} \in \delta_{a_i}^{\AH{\automaton}}$;
        \item $\tuple{\tuple{\mathcal{U}, \tset[3]_{n}}, {\blacktriangleleft}} \in \delta_{\const{I}}^{\AH{\automaton}}$.
    \end{itemize}
    Let $\graph[2]$ be the $\const{I}$-saturation of $G(w)$ such that
    \[\const{I}^{\graph[2]} = \set{\tuple{i,j} \in [0, n]^2 \mid i = j \lor \lnot \Con_{\comnf{\const{I}}}^{\automaton}(U_{i}, U_{j})}.\]
    We have $\Con_{\comnf{\const{I}}}^{\automaton}(U_{i}, U_{j}) \lor \Con_{\const{I}}^{\automaton}(U_{i}, U_{j})$ because $\mathcal{U}$ satisfies $\fml(\mathcal{U}, \tset[3]_i)$ for every $i \in [0, n]$ (\Cref{lemma: key formula transformation}).
    If $\lnot \Con_{\comnf{\const{I}}}^{\automaton}(U_{i}, U_{j})$, then $\Con_{\const{I}}^{\automaton}(U_{i}, U_{j})$; thus $U_{i} = U_{j}$.
    Therefore, the binary relation $\set{\tuple{i, j} \in [0, n]^2 \mid \lnot \Con_{\comnf{\const{I}}}^{\automaton}(U_{i}, U_{j})}$ is symmetric and transitive;
    thus $\const{I}^{\graph[2]}$ is an equivalence relation.
    Additionally, $\graph[2]$ is consistent, because $a^{-}$ does not occur in $\graph[2]$ for any $a \in \SIG$.
    $\graph[2]$ is an edge-extension of $G(w)$, because
    $\const{I}^{\graph[2]} \supseteq \const{I}^{G(w)} = \emptyset$ and $\comnf{\const{I}}^{\graph[2]} \supseteq \comnf{\const{I}}^{G(w)}$ (by the definition of $\const{I}^{\graph[2]}$).
    Hence $\graph[2]$ is indeed an $\const{I}$-saturation of $G(w)$.

    Then $P = \tuple{\graph[2], \set{U_{i}}_{i \in [0, n]}}$ is a saturable path for ${} \not\models_{\REL} w \le \automaton$ as follows.
    For \ref{definition: Hintikka: u src tgt}, \ref{definition: Hintikka: saturate}, and \ref{definition: Hintikka: consistent} for $a \in \SIG^{(-)}$: Similarly for \Cref{lemma: automaton construction soundness: I-}.
    For \ref{definition: Hintikka: consistent} for $a = \const{I}$:
    for every $\tuple{i, j} \in \const{I}^{\graph[2]}$,
    if $i = j$, then $\delta_{\const{I}}^{\automaton}(U_{i}) \subseteq U_{j}$ by the definition of $|\AH{\automaton}|$; thus $\Con_{\const{I}}^{\automaton}(U_{i}, U_{j})$.
    If $\lnot \Con_{\comnf{\const{I}}}^{\automaton}(U_{i}, U_{j})$, then $\Con_{\const{I}}^{\automaton}(U_{i}, U_{j})$ by \ref{definition: Hintikka: saturate}.
    For \ref{definition: Hintikka: consistent} for $a = \comnf{\const{I}}$:
    By the definition of $\const{I}^{\graph[2]}$,
    we have $\Con_{\comnf{\const{I}}}^{\automaton}(U_{i}, U_{j})$ for every $\tuple{i, j} \in [0,n]^2 \setminus \const{I}^{\graph[2]}$.
\end{proof}
\begin{theorem}\label{theorem: automata KACC without a-}
    For every NFAs $\automaton[1], \automaton[2]$, over $\AutSIG$,
    if $\automaton[1]$ does not contain $a^{-}$ for any $a \in \SIG$, then we have
    \[\not\models_{\REL} \automaton[1] \le \automaton[2] \quad \iff \quad \lang{\automaton[1]} \cap \lang{\AH{\automaton[2]}} \neq \emptyset.\]
\end{theorem}
\begin{proof}
    Cf.\ \Cref{theorem: automata KACC without I-} (use \Cref{lemma: automaton construction soundness: a-} instead of \Cref{lemma: automaton construction soundness: I-}).
\end{proof}
\begin{theorem}\label{theorem: KACC without a-}
    The equational theory of $\ExKA$ terms without $a^{-}$ for any $a \in \SIG$ (w.r.t.\ binary relations) is PSPACE-complete.
\end{theorem}
\begin{proof}
    Cf.\ \Cref{theorem: KACC without I-} (use \Cref{theorem: automata KACC without a-} instead of \Cref{theorem: automata KACC without I-}).
\end{proof}

However, we leave open the precise complexity of the equational theory of $\ExKA$ terms, while it is decidable in coNEXP (\Cref{theorem: KACC complexity}) and at least PSPACE-hard (\Cref{proposition: PCoR hardness 2}).
The problematic case is when both $\const{I}^{-}$ and $a^{-}$ occur.
Our automata construction cannot apply to the full case, as follows:
\begin{remark}[Failure of the automata construction (\Cref{definition: AH}) for (full) $\ExKA$ terms]\label{remark: full KACC}
    Consider the soundness (\Cref{lemma: automaton construction soundness: I-,lemma: automaton construction soundness: a-}) for $\ExKA$ terms:
    for a given word $w \in \lang{\AH{\automaton}}$, construct a saturable path for $\not\models_{\REL} w \le \automaton$.
    The essence of the proofs in \Cref{lemma: automaton construction soundness: I-,lemma: automaton construction soundness: a-}
    is that an $\const{I}$-saturation of $G(w)$ always exists.
    However, for (full) $\ExKA$ terms, the situation is changed.
    For example, let $\automaton$ be the NFA obtained from the term $a \comnf{\const{I}} \cup \comnf{\const{I}} \comnf{a}$ (cf.\ \Cref{equation: I- top}) and $w = a b \comnf{a}$:
    \begin{align*}
        \automaton & = \hspace{-.5em} \scalebox{.9}{\begin{tikzpicture}[baseline = -2.ex]
                                                            \graph[grow right = 1.cm, branch down = 1.5ex, nodes={mysmallnode, font = \scriptsize}]{
                                                            {/, 1/{$\mathtt{A}$}[draw, circle]}
                                                            -!- {2/{$\mathtt{B}$}[draw, circle], /, 3/{$\mathtt{C}$}[draw, circle]}
                                                            -!- {/, 4/{$\mathtt{D}$}[draw, circle]}
                                                            };
                                                            \node[left = 4pt of 1](l){} edge[earrow, ->] (1);
                                                            \node[right = 4pt of 4](r){}; \path (4) edge[earrow, ->] (r);
                                                            \graph[use existing nodes, edges={color=black, pos = .5, earrow}, edge quotes={fill=white, inner sep=1pt,font= \scriptsize}]{
                                                            1 ->["$a$"] 2;
                                                            2 ->["$\comnf{\const{I}}$"] 4;
                                                            1 ->["$\comnf{\const{I}}$"] 3;
                                                            3 ->["$\comnf{a}$"]4;
                                                            };
                                                        \end{tikzpicture}}                    &
        G(w)       & = \hspace{-.5em} \scalebox{.9}{\begin{tikzpicture}[baseline = -.5ex]
                                                            \graph[grow right = 1.cm, branch down = 6ex, nodes={mysmallnode, font = \scriptsize}]{
                                                            {0/{$0$}[draw, circle]}-!-{1/{$1$}[draw, circle]}-!-{2/{$2$}[draw, circle]}-!-{3/{$3$}[draw, circle]}
                                                            };
                                                            \node[left = .5em of 0](l){};
                                                            \node[right = .5em of 3](r){};
                                                            \path (0) edge [draw = white, opacity = 0] node[pos= 0.5, inner sep = 1.5pt, font = \scriptsize, opacity = 1](a01){$a$}(1);
                                                            \path (1) edge [draw = white, opacity = 0] node[pos= 0.5, inner sep = 1.5pt, font = \scriptsize, opacity = 1](a12){$b$}(2);
                                                            \path (2) edge [draw = white, opacity = 0] node[pos= 0.5, inner sep = .5pt, font = \scriptsize, opacity = 1](a23){$\comnf{a}$}(3);
                                                            \graph[use existing nodes, edges={color=black, pos = .5, earrow}, edge quotes={fill=white, inner sep=1pt,font= \scriptsize}]{
                                                                0 -- a01 -> 1;
                                                                1 -- a12 -> 2;
                                                                2 -- a23 -> 3;
                                                                l -> 0; 3 -> r;
                                                            };
                                                        \end{tikzpicture}}.
    \end{align*}
    Then $w \in \lang{\AH{\automaton}}$ holds because
    \begin{itemize}
        \item $\tuple{{\blacktriangleright}, \tuple{\mathcal{U}, \tset[3]_{0}}} \in \delta_{\const{I}}^{\AH{\automaton}}$;
        \item $\tuple{\tuple{\mathcal{U}, \tset[3]_{0}}, \tuple{\mathcal{U}, \tset[3]_{1}}} \in \delta_{a}^{\AH{\automaton}}$;
        \item $\tuple{\tuple{\mathcal{U}, \tset[3]_{1}}, \tuple{\mathcal{U}, \tset[3]_{2}}} \in \delta_{b}^{\AH{\automaton}}$;
        \item $\tuple{\tuple{\mathcal{U}, \tset[3]_{2}}, \tuple{\mathcal{U}, \tset[3]_{3}}} \in \delta_{\comnf{a}}^{\AH{\automaton}}$;
        \item $\tuple{\tuple{\mathcal{U}, \tset[3]_{3}}, {\blacktriangleleft}} \in \delta_{\const{I}}^{\AH{\automaton}}$,
    \end{itemize}
    by letting $U_0 = U_2 = \set{\mathtt{A}, \mathtt{D}}$, $U_1 = U_3 = \set{\mathtt{B}, \mathtt{C}}$,
    and $\mathcal{U} = \bigcup_{i \in [0, 3]} \tset[3]_{i} \times (|\automaton| \setminus \tset[3]_{i}) \times \tset[3]_{i} \times (|\automaton| \setminus \tset[3]_{i})$.
    However, there does not exist any saturable path for $\not\models_{\REL} w \le \automaton$, because ${} \models_{\REL} w \le \automaton$ (\Cref{equation: I- top}).
    Additionally, if exists, an $\const{I}$-saturation $\graph[2]$ of $G(w)$ should satisfy $\tuple{0, 2} \in \const{I}^{\graph[2]}$ (by $\comnf{\const{I}} \comnf{a} \in \lang{\automaton}$) and $\tuple{1, 3} \in \const{I}^{\graph[2]}$ (by $a \comnf{\const{I}} \in \lang{\automaton}$) for $\graph[2]^{\Quo} \not\models \automaton$:
    \[\begin{tikzpicture}[baseline = -.5ex]
            \graph[grow right = 1.2cm, branch down = 6ex, nodes={mysmallnode, font = \scriptsize}]{
            {0/{$0$}[draw, circle]}-!-{1/{$1$}[draw, circle]}-!-{2/{$2$}[draw, circle]}-!-{3/{$3$}[draw, circle]}
            };
            \node[left = .5em of 0](l){};
            \node[right = .5em of 3](r){};
            \path (0) edge [draw = white, opacity = 0] node[pos= 0.5, inner sep = 1.5pt, font = \scriptsize, opacity = 1](a01){$a$}(1);
            \path (1) edge [draw = white, opacity = 0] node[pos= 0.5, inner sep = 1.5pt, font = \scriptsize, opacity = 1](a12){$b$}(2);
            \path (2) edge [draw = white, opacity = 0] node[pos= 0.5, inner sep = .5pt, font = \scriptsize, opacity = 1](a23){$\comnf{a}$}(3);
            \path (0) edge [draw = white, opacity = 0, bend right = 30] node[pos= 0.5, inner sep = 1.5pt, font = \scriptsize, opacity = 1](I02){$\const{I}$}(2);
            \path (1) edge [draw = white, opacity = 0, bend left = 30] node[pos= 0.5, inner sep = 1.5pt, font = \scriptsize, opacity = 1](I13){$\const{I}$}(3);
            \graph[use existing nodes, edges={color=black, pos = .5, earrow}, edge quotes={fill=white, inner sep=1pt,font= \scriptsize}]{
                0 -- a01 -> 1;
                1 -- a12 -> 2;
                2 -- a23 -> 3;
                0.south <-[bend right = 5] I02 ->[bend right = 5] 2.south;
                1.north <-[bend left = 5] I13 ->[bend left = 5] 3.north;
                l -> 0; 3 -> r;
            };
        \end{tikzpicture}.\]
    However, $\graph[2]$ does not satisfy \ref{definition: consistent: consistent}; thus we cannot construct consistent $\graph[2]$ even if $w \in \lang{\AH{\automaton}}$.
\end{remark}

\begin{remark}
    One may think that $\AH{\automaton}$ behaves as a complement of an NFA $\automaton$ (cf.\ \Cref{theorem: automata KACC without a-,theorem: automata KACC without I-}).
    Note that $\AH{\automaton}$ is not the language complement of an NFA $\automaton$, i.e., the following does not hold:
    for every word $w$, $w \in \lang{\AH{\automaton}} \iff w \not\in \lang{\automaton}$.
    E.g., let $\automaton$ be the NFA obtained from the term $a \cup \comnf{a}$ and $w = \const{I}$:
    \begin{align*}
        \automaton[1] & = \begin{tikzpicture}[baseline = -.5ex]
                              \graph[grow right = 1.25cm, branch down = 2.5ex, nodes={mysmallnode, font = \scriptsize}]{
                              {1/{$\mathtt{A}$}[draw, circle]}
                              -!- {2/{$\mathtt{B}$}[draw, circle]}
                              };
                              \node[left = 4pt of 1](l){} edge[earrow, ->] (1);
                              \node[right = 4pt of 2](r){}; \path (2) edge[earrow, ->] (r);
                              \path (1) edge [draw = white, opacity = 0, bend left = 25] node[pos= 0.5, inner sep = .5pt, font = \scriptsize, opacity = 1](tau1){$a$}(2);
                              \path (1) edge [draw = white, opacity = 0, bend right = 25] node[pos= 0.5, inner sep = 1.pt, font = \scriptsize, opacity = 1](tau2){$\comnf{a}$}(2);
                              \graph[use existing nodes, edges={color=black, pos = .5, earrow}, edge quotes={fill=white, inner sep=1pt,font= \scriptsize}]{
                              1 -- {tau1} -> 2;
                              1 -- {tau2} -> 2;
                              };
                          \end{tikzpicture} & G(w) & = \begin{tikzpicture}[baseline = -.5ex]
                                                           \graph[grow right = 1.25cm, branch down = 2.5ex, nodes={mysmallnode, font = \scriptsize}]{
                                                           {1/{$0$}[draw, circle]}
                                                           };
                                                           \node[left = 4pt of 1](l){} edge[earrow, ->] (1);
                                                           \node[right = 4pt of 1](r){}; \path (1) edge[earrow, ->] (r);
                                                       \end{tikzpicture}.
    \end{align*}
    By the form of $\automaton$, $w \not\in \lang{\automaton}$.
    However, $w \not\in \lang{\AH{\automaton}}$, as follows.
    Assume that $w \in \lang{\AH{\automaton}}$.
    By the form of $\AH{\automaton}$, there are $\mathcal{U}$ and $\tset[3]_0$ such that $\tuple{{\blacktriangleright}, \tuple{\mathcal{U}, \tset[3]_{0}}} \in \delta_{\const{I}}^{\AH{\automaton}}$, $\tuple{\tuple{\mathcal{U}, \tset[3]_{0}}, {\blacktriangleleft}} \in \delta_{\const{I}}^{\AH{\automaton}}$,
    $\mathtt{A} \in U_0$, and $\mathtt{B} \not\in U_0$.
    By \Cref{lemma: key formula transformation}, $\Con_{\automaton}^{a}(\tset[3]_0, \tset[3]_0) \lor \Con_{\automaton}^{\comnf{a}}(\tset[3]_0, \tset[3]_0)$.
    In either case, $\mathtt{B} \in U_0$, which contradicts $\mathtt{B} \not\in U_0$.
    Thus $w \not\in \lang{\AH{\automaton}}$.
\end{remark}

\newcommand{\CSIG}{A}
\section{Undecidability for {\ECoRTC}}\label{section: undecidability}
A \emph{context-free grammar} (CFG) $\mathcal{C}$ over a finite set $\CSIG$ is a tuple $\tuple{X, \mathcal{R}, \const{s}}$, where
\begin{itemize}
    \item $X$ is a finite set of \emph{non-terminal labels} s.t.\ $\CSIG \cap X = \emptyset$;
    \item $\mathcal{R}$ is a finite set of \emph{rewriting rules} $x \leftarrow w$ of $x \in X$ and $w \in (\CSIG \cup X)^*$;
    \item $\const{s} \in X$ is the \emph{start label}.
\end{itemize}
The \emph{relation} $x \vdash_{\mathcal{C}} w$, where $x \in X$ and $w \in \CSIG^*$, is defined as the minimal relation closed under the following rule:
if $x \leftarrow w_0 x_1 w_1 \dots x_n w_n \in \mathcal{R}$ (where $x, x_1, \dots, x_n \in X$ and $w_0, \dots, w_n \in \CSIG^*$),
then \begin{prooftree}
    \hypo{x_1 \vdash_{\mathcal{C}} v_1}
    \hypo{\dots}
    \hypo{x_n \vdash_{\mathcal{C}} v_n}
    \infer3{x \vdash_{\mathcal{C}} w_0 v_1 w_1 \dots v_n w_n}
\end{prooftree}.
The \emph{language} $\lang{\mathcal{C}} \subseteq \CSIG^*$ is the set $\set{w \mid \const{s} \vdash_{\mathcal{C}} w}$.
It is well-known that the \emph{universality problem} for CFGs---given a CFG $\mathcal{C}$, does $\lang{\mathcal{C}} = \CSIG^*$ hold?---is $\Pi_{1}^{0}$-complete.

Let $\Gamma$ be a set of equations.
We write:
\begin{align*}
    \struc \models \Gamma                     & \defiff \struc \models \term[1] = \term[2] \mbox{ for every ${\term[1] = \term[2]} \in \Gamma$};     \\
    \Gamma \models_{\REL} \term[1] = \term[2] & \defiff \struc \models \term[1] = \term[2] \mbox{ for every $\struc$ s.t.\ $\struc \models \Gamma$}.
\end{align*}

For a CFG $\mathcal{C} = \tuple{X, \mathcal{R}, \const{s}}$ and a word $w = a_1 \dots a_m$,
let $\Gamma_{\mathcal{C}} \defeq \set{w \le x \mid x \leftarrow w \in \mathcal{R}}$.
Here, $\ppmodel_{\mathcal{C}, w} \defeq \model_{\mathcal{C}, w}[0, m]$, where $\model_{\mathcal{C}, w} \defeq \tuple{|\model|, \set{a^{\model}}_{a \in A \cup X}}$ is the structure over $A \cup X$, defined as follows:
\begin{itemize}
    \item $|\model| = [0, m]$;
    \item for $a \in \CSIG$, $a^{\model} = \set{\tuple{i-1, i} \mid i \in [1, m] \mbox{ and } a_i = a}$;
    \item the elements of $\set{x^{\model}}_{x \in X}$ are the minimal relations closed under the following rule:
          if $(w_0 x_1 w_1 \dots x_n w_n \le x) \in \Gamma_{\mathcal{C}}$ (where $x_1, \dots, x_n \in X$ and $w_0, \dots, w_n \in A^*$),
          then
          \[\begin{prooftree}
                  \hypo{\stackanchor{$\tuple{i_0, j_0} \in \jump{w_0}_{\model}$ \quad $\tuple{j_0, i_1} \in x_1^{\model}$ \qquad \dots}{
                  $\tuple{j_{n-1}, i_n} \in x_n^{\model}$ \quad $\tuple{i_n, j_n} \in \jump{w_n}_{\model}$}}
                  \infer1{\tuple{i_0,j_n} \in x^{\model}}
              \end{prooftree}.\]
\end{itemize}
For example, if $A = \set{(,)}$,
$\mathcal{C} = \tuple{\set{\const{s}}, \set{\const{s} \ \leftarrow \ (\const{s}) \const{s},\  \const{s} \ \leftarrow \ \const{I}}, \const{s}}$ (i.e., $\lang{\mathcal{C}}$ is the Dyck-1 language), and $w = (()())$, then $\model_{\mathcal{C}, w}$ is the following:
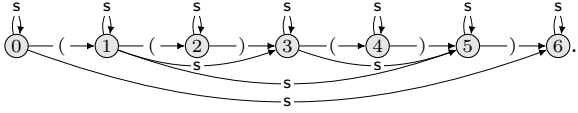
\[\begin{tikzpicture}[baseline = -.5ex]
        \graph[grow right = 1.2cm, branch down = 6ex, nodes={mysmallnode, font = \scriptsize}]{
        {0/{$0$}[draw, circle]}
        -!-{1/{$1$}[draw, circle]}
        -!-{2/{$2$}[draw, circle]}
        -!-{3/{$3$}[draw, circle]}
        -!-{4/{$4$}[draw, circle]}
        -!-{5/{$5$}[draw, circle]}
        -!-{6/{$6$}[draw, circle]}
        };
        \path (0) edge [draw = white, opacity = 0] node[pos= 0.5, inner sep = 1.5pt, font = \scriptsize, opacity = 1](a01){$($}(1);
        \path (1) edge [draw = white, opacity = 0] node[pos= 0.5, inner sep = 1.5pt, font = \scriptsize, opacity = 1](a12){$($}(2);
        \path (2) edge [draw = white, opacity = 0] node[pos= 0.5, inner sep = .5pt, font = \scriptsize, opacity = 1](a23){$)$}(3);
        \path (3) edge [draw = white, opacity = 0] node[pos= 0.5, inner sep = .5pt, font = \scriptsize, opacity = 1](a34){$($}(4);
        \path (4) edge [draw = white, opacity = 0] node[pos= 0.5, inner sep = .5pt, font = \scriptsize, opacity = 1](a45){$)$}(5);
        \path (5) edge [draw = white, opacity = 0] node[pos= 0.5, inner sep = .5pt, font = \scriptsize, opacity = 1](a56){$)$}(6);
        \node(a00)[above = 1.5ex of 0,  inner sep = 1.5pt, font = \scriptsize,]{$\const{s}$};
        \node(a11)[above = 1.5ex of 1,  inner sep = 1.5pt, font = \scriptsize,]{$\const{s}$};
        \node(a22)[above = 1.5ex of 2,  inner sep = 1.5pt, font = \scriptsize,]{$\const{s}$};
        \node(a33)[above = 1.5ex of 3,  inner sep = 1.5pt, font = \scriptsize,]{$\const{s}$};
        \node(a44)[above = 1.5ex of 4,  inner sep = 1.5pt, font = \scriptsize,]{$\const{s}$};
        \node(a55)[above = 1.5ex of 5,  inner sep = 1.5pt, font = \scriptsize,]{$\const{s}$};
        \node(a66)[above = 1.5ex of 6,  inner sep = 1.5pt, font = \scriptsize,]{$\const{s}$};
        \graph[use existing nodes, edges={color=black, pos = .5, earrow}, edge quotes={fill=white, inner sep=1pt,font= \scriptsize}]{
        0 -- a01 -> 1;
        1 -- a12 -> 2;
        2 -- a23 -> 3;
        3 -- a34 -> 4;
        4 -- a45 -> 5;
        5 -- a56 -> 6;
        0 --[bend left = 15] a00 ->[bend left = 15] 0;
        1 --[bend left = 15] a11 ->[bend left = 15] 1;
        2 --[bend left = 15] a22 ->[bend left = 15] 2;
        3 --[bend left = 15] a33 ->[bend left = 15] 3;
        4 --[bend left = 15] a44 ->[bend left = 15] 4;
        5 --[bend left = 15] a55 ->[bend left = 15] 5;
        6 --[bend left = 15] a66 ->[bend left = 15] 6;
        1 ->["$\const{s}$", font=\scriptsize, bend right = 20] 3;
        3 ->["$\const{s}$", font=\scriptsize, bend right = 20] 5;
        1 ->["$\const{s}$", font=\scriptsize, bend right = 20] 5;
        0 ->["$\const{s}$", font=\scriptsize, bend right = 20] 6;
        };
    \end{tikzpicture}.\]
Then we have the following:
\begin{lemma}\label{lemma: CFG encoding}
    Let $\mathcal{C} = \tuple{X, \mathcal{R}, \const{s}}$ be a CFG.
    For every $x \in X$ and $w \in \CSIG^*$, the following are equivalent:
    \begin{enumerate}
        \item \label{lemma: CFG encoding 1} $x \vdash_{\mathcal{C}} w$;
        \item \label{lemma: CFG encoding 2} $\ppmodel_{\mathcal{C}, w} \models x$;
        \item \label{lemma: CFG encoding 3} $\Gamma_{\mathcal{C}} \models_{\REL} w \le x$.
    \end{enumerate}
\end{lemma}
\begin{proof}
    \ref{lemma: CFG encoding 1} $\Longrightarrow$ \ref{lemma: CFG encoding 3}:
    By induction on the derivation tree of $x \vdash_{\mathcal{C}} w$.
    This derivation tree is of the following form:
    \begin{prooftree}
        \hypo{x_1 \vdash_{\mathcal{C}} v_1}
        \hypo{\dots}
        \hypo{x_n \vdash_{\mathcal{C}} v_n}
        \infer3{x \vdash_{\mathcal{C}} w_0 v_1 w_1 \dots v_n w_n}
    \end{prooftree},
    where $n \in \nat$, $x_1, \dots, x_n \in X$, $w_0, \dots, w_n \in \CSIG^*$, $v_1, \dots, v_n \in \CSIG^*$ s.t.\
    $w = w_0 v_1 w_1 \dots v_n w_n$ and $x \leftarrow w_0 x_1 w_1 \dots x_n w_n \in \mathcal{R}$.
    W.r.t.\ $\Gamma_{\mathcal{C}} \models_{\REL}$, we have
    \begin{align*}
        w & =  w_0 v_1 w_1 \dots v_n w_n                                                                                                \\
          & \le w_0 x_1 w_1 \dots x_n w_n                  \tag{$\Gamma_{\mathcal{C}} \models_{\REL} v_i \le x_i$ for $i \in [n]$ (IH)} \\
          & \le x \tag{$(w_0 x_1 w_1 \dots x_n w_n \le x) \in \Gamma_{\mathcal{C}}$}.
    \end{align*}
    Hence $\Gamma_{\mathcal{C}} \models_{\REL} w \le x$.
    \ref{lemma: CFG encoding 3} $\Longrightarrow$ \ref{lemma: CFG encoding 2}:
    Because $\model_{\mathcal{C}, w} \models \Gamma_{\mathcal{C}}$ and $\ppmodel_{\mathcal{C}, w} \models w$ hold.
    \ref{lemma: CFG encoding 2} $\Longrightarrow$ \ref{lemma: CFG encoding 1}:
    We show the following:
    \begin{sublemma*}
        For every $i, j \in [0, m]$ and $x \in X$,
        if $\model_{\mathcal{C}, w}[i,j] \models x$, then $i \le j$ and $x \vdash_{\mathcal{C}} a_i \dots a_{j-1}$.
    \end{sublemma*}
    \begin{proof}
        By induction on the derivation tree from the definition of $\set{x^{\model_{\mathcal{C}, w}}}_{x \in X}$.
        This derivation tree is of the following form
        $\begin{prooftree}
                \hypo{\stackanchor{$\tuple{i_0, j_0} \in \jump{w_0}_{\model_{\mathcal{C}, w}}$ \quad $\tuple{j_0, i_1} \in x_1^{\model_{\mathcal{C}, w}}$ \quad \dots}{
                $\tuple{j_{n-1}, i_n} \in x_n^{\model_{\mathcal{C}, w}}$ \quad $\tuple{i_n, j_n} \in \jump{w_n}_{\model_{\mathcal{C}, w}}$}}
                \infer1{\tuple{i_0,j_n} \in x^{\model_{\mathcal{C}, w}}}
            \end{prooftree}$,
        where $x_1, \dots, x_n \in X$, $w_0, \dots, w_n \in \CSIG^*$,
        $i_0, j_0, \dots, i_n, j_n \in |\model_{\mathcal{C}, w}|$,
        and $\tuple{i_0, j_n} = \tuple{i, j}$.
        By the definition of $a^{\model_{\mathcal{C}, w}}$ (where $a \in \CSIG$), we have $i_k \le j_k$ and $w_k = a_{i_k} \dots a_{j_k - 1}$.
        By IH, we have $j_{k-1} \le i_k$ and $x_k \vdash_{\mathcal{C}} a_{j_{k-1}} \dots a_{i_{k}-1}$.
        Combining them yields $i \le j$ and $x \vdash_{\mathcal{C}}  a_i \dots a_{j-1}$ (because $x \leftarrow w_0 x_1 w_1 \dots x_n w_n \in \mathcal{R}$).
    \end{proof}
    \noindent By specializing this sub-lemma with $\tuple{i,j} = \tuple{0,m}$, this completes the proof.
\end{proof}
\begin{lemma}\label{lemma: CFG encoding lang}
    For every CFG $\mathcal{C} = \tuple{X, \mathcal{R}, \const{s}}$, we have
    \[ \lang{\mathcal{C}} = \CSIG^* \quad \iff \quad \Gamma_{\mathcal{C}} \models_{\REL} \CSIG^* \le \const{s}.\]
    (Here, $\CSIG^*$ denotes the term $(a_1 \cup \dots \cup a_n)^*$ in the right-hand side, where $\CSIG = \set{a_1, \dots, a_n}$.)
\end{lemma}
\begin{proof}
    We have
    \begin{align*}
        \lang{\mathcal{C}} = \CSIG^* & \iff \forall w \in \CSIG^*.\ \const{s} \vdash_{\mathcal{C}} w                    \tag{Def.\ of $\lang{}$}         \\
                                     & \iff \forall w \in \CSIG^*.\ \Gamma_{\mathcal{C}} \models_{\REL} w \le \const{s} \tag{\Cref{lemma: CFG encoding}} \\
                                     & \iff \Gamma_{\mathcal{C}} \models_{\REL} \CSIG^* \le \const{s}. \tag*{\qedhere}
    \end{align*}
\end{proof}
Additionally, we prepare the following deduction lemma for \emph{Hoare hypotheses} $\term[3] = \bot$ (cf.\ \cite[Thm.\ 4.1]{Kozen2000} for KAT):
\begin{lemma}\label{lemma: deduction theorem}
    For every $\ECoRTC$ terms $\term[1], \term[2], \term[3]$ and every set $\Gamma$ of equations,
    we have
    \[\Gamma \cup \set{\term[3] = \bot} \models_{\REL} \term[1] \le \term[2] \iff \Gamma \models_{\REL} \term[1] \le \term[2] \cup (\top \term[3] \top).\]
\end{lemma}
\begin{proof}
    We have
    \begin{align*}
                                           & \Gamma \cup \set{\term[3] = \bot} \models_{\REL} \term[1] \le \term[2]                                                   \\
                                           & \iff \mbox{for every ${\struc}[x,y]$ s.t.\ $\struc \models \Gamma \cup \set{\term[3] = \bot}$,}                          \\
                                           & \hspace{10em} \mbox{${\struc}[x,y] \models \term[1] \le \term[2]$} \tag{Def.\ of $\models_{\REL}$}                       \\
                                           & \iff \mbox{for every ${\struc}[x,y]$ s.t.\  $\struc \models \Gamma$ and ${\struc}[x,y] \not\models \top \term[3] \top$,} \\
        \label{lemma: deduction theorem 1} & \hspace{10em} \mbox{${\struc}[x,y] \models \term[1] \le \term[2]$}  \tag{$\clubsuit$1}                                   \\
        \label{lemma: deduction theorem 2} & \iff \mbox{for every ${\struc}[x,y]$ s.t.\  $\struc \models \Gamma$,}                                                    \\
                                           & \hspace{10em} {\struc}[x,y] \models \term[1] \le \term[2] \cup (\top \term[3] \top)\tag{$\clubsuit$2}                    \\
                                           & \iff \Gamma \models_{\REL} \term[1] \le \term[2] \cup (\top \term[3] \top).  \tag{Def.\ of $\models_{\REL}$}
    \end{align*}
    Here, (\ref{lemma: deduction theorem 1}) is because $\model \models \term[3] = \bot \iff \jump{\term[3]}_{\model} = \emptyset \iff {\model}[x, y] \not\models \top \term[3] \top$.
    (\ref{lemma: deduction theorem 2}) is because:
    \begin{align*}
         & ({\struc}[x,y] \not\models \top \term[3] \top) \mbox{ implies } ({\struc}[x,y] \models \term[1] \le \term[2])                      \\
         & \iff {\struc}[x,y] \models \top \term[3] \top \ \lor\  {\struc}[x,y] \not\models \term[1] \ \lor \  {\struc}[x,y] \models \term[2] \\
         & \iff {\struc}[x,y] \not\models \term[1] \ \lor \  {\struc}[x,y] \models \term[2] \cup (\top \term[3] \top)                         \\
         & \iff {\struc}[x,y] \models \term[1] \le \term[2] \cup (\top \term[3] \top).  \tag*{\qedhere}
    \end{align*}
\end{proof}
By \Cref{lemma: CFG encoding lang,lemma: deduction theorem}, we have the following:
\begin{theorem}\label{theorem: PCoRTC complexity}
    The equational theory of $\ECoRTC$ is $\Pi_{1}^{0}$-complete.
\end{theorem}
\begin{proof}
    For upper bound:
    By \Cref{lemma: PCoRTC upper bound}.
    For hardness:
    Let $\mathcal{C} = \tuple{X, \set{x_i \leftarrow w_i \mid i \in [n]}, \const{s}}$ be a CFG.
    Then we have
    \begin{align*}
         & \lang{\mathcal{C}} = A^*                                                                                                                                                                                                                      \\
         & \iff \set{w_i \le x_i \mid i \in [n]} \models_{\REL} A^* \le \const{s}                                                                     \tag{\Cref{lemma: CFG encoding lang}}                                                              \\
         & \iff \set{w_i \cap x_i^{-} = \bot \mid i \in [n]} \models_{\REL} A^* \le \const{s}                                                         \tag{For every $\struc$, $\model \models w_i \le x_i \iff \model \models w_i \cap x_i^{-} = \bot$} \\
         & \iff {} \models_{\REL} A^* \le \const{s} \cup \left(\bigcup_{i = 1}^{n} \top (w_i \cap x_i^{-}) \top\right) \tag{\Cref{lemma: deduction theorem}}                                                                                             \\
         & \iff {} \models_{\REL} A^* \le \const{s} \cup \left(\bigcup_{i = 1}^{n} (a \cup a^{-}) (w_i \cap x_i^{-}) (a \cup a^{-})\right). \tag{${} \models_{\REL} \top = a \cup a^{-}$, where $a$ is some element in $A$.}
    \end{align*}
    Thus we can reduce the universality problem for CFGs, which is $\Pi_{1}^{0}$-hard, to the equational theory of $\ECoRTC$ (precisely, $\Termset_{\set{\cdot, \cup, \cap, \bl^{*}}}$ with the complement of term variables).
\end{proof}

\section{Conclusion and future work}\label{section: conclusion}
We have studied the computational complexity of existential calculi of relations with transitive closure, using edge saturations.
A natural interest is to extend our complexity results for more general syntaxes.
We believe that the upper bound results for intersection-free fragments hold even if we extend them with \emph{tests} in KAT (by considering \emph{guarded strings} \cite{Kozen1996} instead of words (strings), in saturable paths); e.g., \emph{KAT with top} (w.r.t.\ binary relations), which are recently studied for modeling \emph{incorrectness logic} \cite{ohearnIncorrectnessLogic2019, zhangIncorrectnessLogicKleene2022, Pous2022, pousCompletenessTheoremsKleene2023}.

Another future work is to study the axiomatizability of them.
Unfortunately, the equational theory of (full) $\ECoRTC$ is not finitely axiomatizable because it is not recursively enumerable (\Cref{theorem: PCoRTC complexity});
but we leave it open to finding some complete (finite) axiomatization for its fragments, including KA terms with complements of term variables.
(The equation (\ref{equation: I- top 0}) indicates that, at least, we need axioms of KA with top w.r.t.\ binary relations \cite{Pous2022}.)

\paragraph*{Acknowledgments}
We would like to thank the anonymous reviewers for their useful comments.

\bibliographystyle{IEEEtran}
\bibliography{IEEEabrv,main-pand.bib}

\ifiscameraready
\else
    \appendices
\crefalias{section}{appendix}
\section{Proof of the equations in the Introduction}\label{section: equations}
(The notations in this section depend on \Cref{section: ECoR def}.)

For \Cref{equation: I- top 0}: ${}  \models_{\REL} \top = a \cup a^{-}$.
For every structure $\struc$ and $x, y \in |\struc|$,
we have
\begin{align*}
    \tuple{x, y} \in \jump{\top}_{\struc}
     & \iff \mathsf{true}                                                                  \\
     & \iff \tuple{x, y} \in \jump{a}_{\struc} \lor \tuple{x, y} \not\in \jump{a}_{\struc} \\
     & \iff \tuple{x, y} \in \jump{a \cup a^{-}}_{\struc}.                    \tag*{\qed}
\end{align*}

For \Cref{equation: I- top}: ${}  \models_{\REL} a b a^{-}         \le \const{I}^{-} a^{-} \cup a \const{I}^{-}$.
Assume $\tuple{x_0, x_3} \in \jump{a b a^{-}}_{\struc}$.
Let $x_1$ and $x_2$ be s.t.\ $\tuple{x_0, x_1} \in \jump{a}_{\struc}$, $\tuple{x_1, x_2} \in \jump{b}_{\struc}$, and $\tuple{x_2, x_3} \in \jump{a^{-}}_{\struc}$.
Because $\tuple{x_0, x_1} \in a^{\struc}$ and $\tuple{x_2, x_3} \not\in a^{\struc}$, $\tuple{x_0, x_1} \neq \tuple{x_2, x_3}$ should hold.
We distinguish the following cases:
\begin{itemize}
    \item If $x_0 \neq x_2$, then $\tuple{x_0, x_2} \in \jump{\const{I}^{-}}_{\struc}$ and $\tuple{x_2, x_3} \in \jump{a^{-}}_{\struc}$.
          Thus, $\tuple{x_0, x_3} \in \jump{\const{I}^{-} a^{-} \cup a \const{I}^{-}}_{\struc}$.
    \item Otherwise ($x_1 \neq x_3$), $\tuple{x_0, x_1} \in \jump{a}_{\struc}$ and $\tuple{x_1, x_3} \in \jump{\const{I}^{-}}_{\struc}$.
          Thus, $\tuple{x_0, x_3} \in \jump{\const{I}^{-} a^{-} \cup a \const{I}^{-}}_{\struc}$.
\end{itemize}
This completes the proof.
\qed

For \Cref{equation: I- ID}: ${} \models_{\REL} a               \le \const{I}^{-} \cup a a$.
Assume $\tuple{x, y} \in \jump{a}_{\struc}$.
We distinguish the following cases:
\begin{itemize}
    \item If $x \neq y$, then $\tuple{x,y} \in \jump{\const{I}^{-}}_{\struc}$.
          Thus, $\tuple{x,y} \in \jump{\const{I}^{-} \cup a a}_{\struc}$.
    \item Otherwise ($x = y$),
          $\tuple{x, y} \in \jump{a}_{\struc}$, $y = x$, and $\tuple{x, y} \in \jump{a}_{\struc}$.
          Thus, $\tuple{x,y} \in \jump{\const{I}^{-} \cup a a}_{\struc}$.
\end{itemize}
This completes the proof.
\qed

For \Cref{equation: I- De Morgan}: ${} \models_{\REL} a^{\smile} a^{-} \le \const{I}^{-}$.
Assume $\tuple{x_0, x_2} \in \jump{a^{\smile} a^{-}}_{\struc}$.
Let $x_1$ be s.t.\ $\tuple{x_0, x_1} \in \jump{a^{\smile}}_{\struc}$ and $\tuple{x_1, x_2} \in \jump{a^{-}}_{\struc}$.
Because $\tuple{x_1, x_0} \in a^{\struc}$ and $\tuple{x_1, x_2} \not\in a^{\struc}$, $x_0 \neq x_2$ should hold.
Thus, we have $\tuple{x_0, x_2} \in \jump{\const{I}^{-}}_{\struc}$.
This completes the proof.
\qed

\newcommand{\genTermset}{\Termset_{\mathrm{gen}}}
\section{Note: Supplement of \Cref{footnote: term set}}\label{section: connf}
Let $\genTermset$ be the set of terms defined as follows:
\begin{align*}
    \genTermset \ni \term[1], \term[2], \term[3] & ::= \aterm \mid \aterm^{-} \mid \const{I} \mid \const{I}^{-} \mid \bot \mid \bot^{-} \mid \top \mid \top^{-}          \\
                                                 & \mid \term[1] \cdot \term[2] \mid \term[1] \cup \term[2] \mid \term[1]^{*} \mid \term^{\smile} \tag*{($a \in \SIG$).}
\end{align*}
The \emph{binary relation} $\jump{\term}_{\model} \subseteq |\model|^2$ of an $\genTermset$ term $\term$ on a structure $\model$ is defined as follows, where $a \in \SIG$ and $b \in \SIG \cup \set{\const{I}, \bot, \top}$:
\begin{flalign*}
    \jump{a}_{\model}                       & \defeq a^{\model}                                              & 
    \jump{b^{-}}_{\model}                   & \defeq (b^{-})^{\model}                                           \\
    \jump{\bot}_{\model}                    & \defeq \emptyset                                               & 
    \jump{\top}_{\model}                    & \defeq |\model|^2                                                 \\
    \jump{\term[1] \cup \term[2]}_{\model}  & \defeq \jump{\term[1]}_{\model} \cup \jump{\term[2]}_{\model}  & 
    \jump{\term[1] \cap \term[2]}_{\model}  & \defeq \jump{\term[1]}_{\model} \cap \jump{\term[2]}_{\model}     \\
    \jump{\const{I}}_{\model}               & \defeq \const{I}^{\model}                                      & 
    \\
    \jump{\term[1] \cdot \term[2]}_{\model} & \defeq \jump{\term[1]}_{\model} \cdot \jump{\term[2]}_{\model}    \\
    \jump{\term[1]^{*}}_{\model}            & \defeq \bigcup_{n \in \nat} \jump{\term[1]^{n}}_{\model}       & 
    \jump{\term^{\smile}}_{\model}          & \defeq \jump{\term}_{\model}^{\smile}
\end{flalign*}
We write $\models_{\REL} \term[1] = \term[2]$ if $\jump{\term[1]}_{\model} = \jump{\term[2]}_{\model}$ for every structure $\model$.
The \emph{equational theory over $\sig$ w.r.t.\ binary relations} is defined as the set of every pairs $\term[1] = \term[2]$ of terms in $\Termset_{\sig}$ s.t.\ $\models_{\REL} \term[1] = \term[2]$.
We use the notation $\models$ (\Cref{notation: models}) also for $\genTermset$.

For term $\term \in \genTermset$,
the \emph{converse normal form} $\widebreve{\term}$ (e.g.,\ \cite[Sect.\ 2.1]{Brunet2016}) of the term $\term^{\smile}$ is defined by:
\begin{align*}
    \widebreve{a}                            & \defeq a^{\smile}                                      & \widebreve{a^{\smile}}                         & \defeq a     \tag{$a \in \SIG^{(-)}$}                                                 \\
    \widebreve{a}                            & \defeq a                                               & \widebreve{a^{\smile}}                         & \defeq a \tag{$a \in \set{\const{I}, \const{I}^{-}, \top, \top^{-}, \bot, \bot^{-}}$} \\
    \widebreve{\term[1] \cdot \term[2]}      & \defeq \widebreve{\term[2]} \cdot \widebreve{\term[1]} & \widebreve{(\term[1] \cdot \term[2])^{\smile}} & \defeq \widebreve{\term[1]^{\smile}} \cdot \widebreve{\term[2]^{\smile}}              \\
    \widebreve{\term[1]^*}                   & \defeq (\widebreve{\term[1]})^*                        & \widebreve{(\term[1]^*)^{\smile}}              & \defeq (\widebreve{\term[1]^{\smile}})^*                                              \\
    \widebreve{\term[1] \cup \term[2]}       & \defeq \widebreve{\term[1]} \cup \widebreve{\term[2]}  & \widebreve{(\term[1] \cup \term[2])^{\smile}}  & \defeq \widebreve{\term[1]^{\smile}} \cup \widebreve{\term[2]^{\smile}}               \\
    \widebreve{(\term[1]^{\smile})^{\smile}} & \defeq \widebreve{\term[1]}.
\end{align*}
\begin{proposition}\label{proposition: cnf}
    For every term $\term \in \genTermset$, $\models_{\REL} \widebreve{\term} = \term^{\smile}$.
\end{proposition}
\begin{proof}[Proof Sketch]
    By straightforward induction on $\term$.
    For example, if $\term = \term[2] \cdot \term[3]$, 
    for every ${\model}[x,y]$, we have
    \begin{align*}
        {\model}[x,y] \models \widebreve{\term[2] \cdot \term[3]} & \iff  {\model}[x,y] \models \widebreve{\term[3]} \cdot \widebreve{\term[2]} \tag{Def.\ of $\widebreve{\bullet}$}                             \\
                                                                  & \iff \exists z.\   {\model}[x,z] \models \widebreve{\term[3]} \ \land \  {\model}[z,y] \models \widebreve{\term[2]} \tag{Def.\ of $\jump{}$} \\
                                                                  & \iff \exists z.\  {\model}[x,z] \models \term[3]^{\smile} \  \land \  {\model}[z,y] \models \term[2]^{\smile} \tag{IH}                       \\
                                                                  & \iff {\model}[x,y] \models \term[3]^{\smile} \cdot \term[2]^{\smile}  \tag{Def.\ of $\jump{}$}                                               \\
                                                                  & \iff  {\model}[x,y] \models (\term[3] \cdot \term[2])^{\smile}. \tag{Def.\ of $\jump{}$}
    \end{align*}
\end{proof}
\begin{proposition}\label{proposition: KACC}
    The equational theory of $\genTermset$ can be reduced to that of $\ECoRTC$ in polynomial time.
\end{proposition}
\begin{proof}
    Let $\term[1]$ and $\term[2]$ be terms in $\genTermset$.
    Let $\term[1]'$ (resp.\ $\term[2]'$) be the term $ \widebreve{\widebreve{\term[1]}}$ (resp.\ $\widebreve{\widebreve{\term[2]}}$)
    in which each $\top^{-}$ has been replaced with $\bot$ and each $\bot^{-}$ has been replaced with $\top$.
    Then, ${} \models_{\REL} \term[1] = \term[2] \iff {} \models_{\REL} \term[1]' = \term[2]'$ by
    \Cref{proposition: cnf} with $\models_{\REL} \top^{-} = \bot$ and $\models_{\REL} \bot^{-} = \top$.
    By the construction, $\term[1]'$ and $\term[2]'$ can be obtained from $\term[1]$ and $\term[2]$ in polynomial time, respectively.
    Also, $\term[1]'$ and $\term[2]'$ are $\genTermset$ terms, because $\smile$ only applies to $a \in \SIG^{(-)}$ and $\top^{-}$ and $\bot^{-}$ does not occur.
\end{proof}

\section{Note: Proof of \Cref{proposition: semantics for graphs,proposition: graph characterization}}\label{section: proposition: semantics for graphs}
\begin{proposition}\label{proposition: graph operation}
    For every structure $\model$ and graphs $G, H$,
    \begin{align*}
        \label{proposition: graph operation: cap}\jump{G \cap H}_{\model}      & = \jump{G}_{\model} \cap \jump{H}_{\model} \tag{\Cref{proposition: graph operation}$\cap$}   \\
        \label{proposition: graph operation: cdot} \jump{G \cdot H}_{\model}   & = \jump{G}_{\model} \cdot \jump{H}_{\model} \tag{\Cref{proposition: graph operation}$\cdot$} \\
        \label{proposition: graph operation: smile} \jump{G^{\smile}}_{\model} & = \jump{G}_{\model}^{\smile}. \tag{\Cref{proposition: graph operation}$\smile$}
    \end{align*}
\end{proposition}
\begin{proof}
    \ref{proposition: graph operation: cap}.
    By the definition of $\jump{}_{\model}$ (\Cref{definition: semantics for graphs}), it suffices to show that for every $x, y \in |\model|$:
    \begin{align*}
         & \exists f.\ f \colon (G \cap H) \homo {\model}[x,y]                                                      \\
         & \iff \exists f_{G}, f_{H}.\ f_{G} \colon G \homo {\model}[x,y] \land f_{H} \colon H \homo {\model}[x,y].
    \end{align*}
    $\Longrightarrow$:
    By letting $f_{G} = \set{\tuple{x', f(x')} \mid x' \in |G|}$ and $f_{H} = \set{\tuple{x', f(x')} \mid x' \in |H|}$.
    $\Longleftarrow$:
    By letting $f = f_{G} \cup f_{H}$.
    Note that $f_{G}(\src^{G}) = x = f_{H}(\src^{H})$ and $f_{G}(\tgt^{G}) = y = f_{H}(\tgt^{H})$ by \Cref{definition: semantics for graphs}; so $f$ is indeed a map.
    
    \ref{proposition: graph operation: cdot}.
    By the definition of $\jump{}_{\model}$, it suffices to show that for every $x, y \in |\model|$:
    \begin{align*}
         & \exists f.\ f \colon (G \cdot H) \homo {\model}[x,y]                                                         \\
         & \iff \exists z, f_{G}, f_{H}.\  f_{G} \colon G \homo {\model}[x,z] \land f_{H} \colon H \homo {\model}[z,y].
    \end{align*}
    $\Longrightarrow$:
    By letting $z = f(\tgt^{G})$, $f_{G} = \set{\tuple{x', f(x')} \mid x' \in |G|}$, and $f_{H} = \set{\tuple{x', f(x')} \mid x' \in |H|}$.
    $\Longleftarrow$:
    By letting $f = f_{G} \cup f_{H}$.
    Note that $f_{G}(\tgt^{G}) = z = f_{H}(\src^{H})$ by \Cref{definition: semantics for graphs}; so $f$ is indeed a map.
    
    \ref{proposition: graph operation: smile}.
    By the definition of $\jump{}_{\model}$, it suffices to show that for every $x, y \in |\model|$:
    \begin{align*}
        \exists f.\ f \colon G^{\smile} \homo {\model}[x,y] & \ \iff\  \exists f_{G}.\ f_{G} \colon G \homo {\model}[y,x].
    \end{align*}
    This is clear by using the same map.
\end{proof}

\begin{proposition}[restatement of \Cref{proposition: semantics for graphs}]
    For every structure $\model$ and $\PCoRTC$ term $\term$, we have
    $\jump{\term}_{\model} = \jump{\glang(\term)}_{\model}$.
\end{proposition}
\begin{proof}
    By induction on $\term$.
    
    Case $\term = a$ where $a \in \SIG$:
    For every $x, y \in |\model|$, we have
    \begin{align*}
        \tuple{x, y} \in \jump{a}_{\model} & \iff \tuple{x, y} \in a^{\model}                                                         \tag{Def.\ of $\jump{}$}                                     \\
                                           & \iff \begin{tikzpicture}[baseline = -.5ex]
                                                      \graph[grow right = 1.2cm, branch down = 6ex, nodes={mynode}]{
                                                      {0/{}[draw, circle]}-!-{1/{}[draw, circle]}
                                                      };
                                                      \node[left = .5em of 0](l){};
                                                      \node[right = .5em of 1](r){};
                                                      \path (0) edge [draw = white, opacity = 0] node[pos= 0.5, inner sep = 1.5pt, font = \scriptsize, opacity = 1](a1){$a$}(1);
                                                      \graph[use existing nodes, edges={color=black, pos = .5, earrow}, edge quotes={fill=white, inner sep=1pt,font= \scriptsize}]{
                                                          0 -- a1 -> 1;
                                                          l -> 0; 1 -> r;
                                                      };
                                                  \end{tikzpicture} \homo {\model}[x,y]                        \\
                                           & \iff \tuple{x, y} \in \jump{\begin{tikzpicture}[baseline = -.5ex]
                                                                                 \graph[grow right = 1.2cm, branch down = 6ex, nodes={mynode}]{
                                                                                 {0/{}[draw, circle]}-!-{1/{}[draw, circle]}
                                                                                 };
                                                                                 \node[left = .5em of 0](l){};
                                                                                 \node[right = .5em of 1](r){};
                                                                                 \path (0) edge [draw = white, opacity = 0] node[pos= 0.5, inner sep = 1.5pt, font = \scriptsize, opacity = 1](a1){$a$}(1);
                                                                                 \graph[use existing nodes, edges={color=black, pos = .5, earrow}, edge quotes={fill=white, inner sep=1pt,font= \scriptsize}]{
                                                                                     0 -- a1 -> 1;
                                                                                     l -> 0; 1 -> r;
                                                                                 };
                                                                             \end{tikzpicture}}_{\model} \tag{\Cref{definition: semantics for graphs}} \\
                                           & \iff \tuple{x, y} \in \jump{\glang(a)}_{\model}. \tag{Def.\ of $\glang$}
    \end{align*}
    
    Case $\term = \const{I}$:
    For every $x, y \in |\model|$, we have
    \begin{align*}
        \tuple{x, y} \in \jump{\const{I}}_{\model} & \iff \tuple{x, y} \in \const{I}^{\model}                                                         \tag{Def.\ of $\jump{}$}                                                   \\
                                                   & \iff \begin{tikzpicture}[baseline = -.5ex]
                                                              \graph[grow right = 1.2cm, branch down = 6ex, nodes={mynode}]{
                                                              {0/{}[draw, circle]}-!-{1/{}[draw, circle]}
                                                              };
                                                              \node[left = .5em of 0](l){};
                                                              \node[right = .5em of 1](r){};
                                                              \path (0) edge [draw = white, opacity = 0] node[pos= 0.5, inner sep = 1.5pt, font = \scriptsize, opacity = 1](a1){$\const{I}$}(1);
                                                              \graph[use existing nodes, edges={color=black, pos = .5, earrow}, edge quotes={fill=white, inner sep=1pt,font= \scriptsize}]{
                                                                  0 -- a1 -> 1;
                                                                  l -> 0; 1 -> r;
                                                              };
                                                          \end{tikzpicture} \homo {\model}[x,y]                        \\
                                                   & \iff \tuple{x, y} \in \jump{\begin{tikzpicture}[baseline = -.5ex]
                                                                                         \graph[grow right = 1.2cm, branch down = 6ex, nodes={mynode}]{
                                                                                         {0/{}[draw, circle]}-!-{1/{}[draw, circle]}
                                                                                         };
                                                                                         \node[left = .5em of 0](l){};
                                                                                         \node[right = .5em of 1](r){};
                                                                                         \path (0) edge [draw = white, opacity = 0] node[pos= 0.5, inner sep = 1.5pt, font = \scriptsize, opacity = 1](a1){$\const{I}$}(1);
                                                                                         \graph[use existing nodes, edges={color=black, pos = .5, earrow}, edge quotes={fill=white, inner sep=1pt,font= \scriptsize}]{
                                                                                             0 -- a1 -> 1;
                                                                                             l -> 0; 1 -> r;
                                                                                         };
                                                                                     \end{tikzpicture}}_{\model} \tag{\Cref{definition: semantics for graphs}} \\
                                                   & \iff \tuple{x, y} \in \jump{\begin{tikzpicture}[baseline = -.5ex]
                                                                                         \graph[grow right = 1.2cm, branch down = 6ex, nodes={mynode}]{
                                                                                         {0/{}[draw, circle]}
                                                                                         };
                                                                                         \node[left = .5em of 0](l){};
                                                                                         \node[right = .5em of 0](r){};
                                                                                         \graph[use existing nodes, edges={color=black, pos = .5, earrow}, edge quotes={fill=white, inner sep=1pt,font= \scriptsize}]{
                                                                                             l -> 0; 0 -> r;
                                                                                         };
                                                                                     \end{tikzpicture}}_{\model} \tag{$\const{I}^{\model}$ is the identity relation}                                                    \\
                                                   & \iff \tuple{x, y} \in \jump{\glang(\const{I})}_{\model}. \tag{Def.\ of $\glang$}
    \end{align*}
    
    Case $\term = \top$:
    For every $x, y \in |\model|$, we have
    \begin{align*}
        \tuple{x, y} \in \jump{\top}_{\model} & \iff \const{true}                                                  \tag{Def.\ of $\jump{}$}                              \\
                                              & \iff \tuple{x, y} \in \jump{\begin{tikzpicture}[baseline = -.5ex]
                                                                                    \graph[grow right = 1.2cm, branch down = 6ex, nodes={mynode}]{
                                                                                    {0/{}[draw, circle]}-!-{1/{}[draw, circle]}
                                                                                    };
                                                                                    \node[left = .5em of 0](l){};
                                                                                    \node[right = .5em of 1](r){};
                                                                                    \graph[use existing nodes, edges={color=black, pos = .5, earrow}, edge quotes={fill=white, inner sep=1pt,font= \scriptsize}]{
                                                                                        l -> 0; 1 -> r;
                                                                                    };
                                                                                \end{tikzpicture}}_{\model} \tag{\Cref{definition: semantics for graphs}} \\
                                              & \iff \tuple{x, y} \in \jump{\glang(\top)}_{\model}. \tag{Def.\ of $\glang$}
    \end{align*}
    
    Case $\term = \bot$:
    For every $x, y \in |\model|$, we have
    \begin{align*}
        \tuple{x, y} \in \jump{\bot}_{\model} & \iff \const{false}                                                  \tag{Def.\ of $\jump{}$}  \\
                                              & \iff \tuple{x, y} \in \jump{\emptyset}_{\model} \tag{\Cref{definition: semantics for graphs}} \\
                                              & \iff \tuple{x, y} \in \jump{\glang(\bot)}_{\model}. \tag{Def.\ of $\glang$}
    \end{align*}
    
    Case $\term = \term[2] \cdot \term[3]$:
    \begin{align*}
         & \jump{\term[2] \cdot \term[3]}_{\model}  = \jump{\term[2]}_{\model} \cdot \jump{\term[3]}_{\model}  \tag{Def.\ of $\jump{}$}                                               \\
         & = \jump{\glang(\term[2])}_{\model} \cdot \jump{\glang(\term[3])}_{\model}          \tag{IH}                                                                                \\
         & = \bigcup_{G \in \glang(\term[2])} \bigcup_{H \in \glang(\term[3])} \jump{G}_{\model} \cdot \jump{H}_{\model}                 \tag{$\cdot$ is distributive w.r.t.\ $\cup$} \\
         & = \bigcup_{G \in \glang(\term[2])} \bigcup_{H \in \glang(\term[3])} \jump{G \cdot H}_{\model} \tag{\Cref{proposition: graph operation: cdot}}                              \\
         & = \jump{\glang(\term[2] \cdot \term[3])}_{\model}. \tag{Def.\ of $\glang$}
    \end{align*}
    
    Case $\term = \term[2] \cap \term[3]$:
    \begin{align*}
         & \jump{\term[2] \cap \term[3]}_{\model}   = \jump{\term[2]} \cap \jump{\term[3]}                     \tag{Def.\ of $\jump{}$}                                              \\
         & = \jump{\glang(\term[2])}_{\model} \cap \jump{\glang(\term[3])}_{\model}            \tag{IH}                                                                              \\
         & =  \bigcup_{G \in \glang(\term[2])} \bigcup_{H \in \glang(\term[3])} \jump{G}_{\model} \cap \jump{H}_{\model}                 \tag{$\cap$ is distributive w.r.t.\ $\cup$} \\
         & =  \bigcup_{G \in \glang(\term[2])} \bigcup_{H \in \glang(\term[3])} \jump{G \cap H}_{\model} \tag{\Cref{proposition: graph operation: cap}}                              \\
         & = \jump{\glang(\term[2] \cap \term[3])}_{\model}. \tag{Def.\ of $\glang$}
    \end{align*}
    
    Case $\term = \term[2] \cup \term[3]$:
    \begin{align*}
        \jump{\term[2] \cup \term[3]}_{\model} & = \jump{\term[2]}_{\model} \cup \jump{\term[3]}_{\model}                     \tag{Def.\ of $\jump{}$} \\
                                               & = \jump{\glang(\term[2])}_{\model} \cup \jump{G(\term[3])}_{\model}            \tag{IH}               \\
                                               & = \jump{\glang(\term[2]) \cup \glang(\term[3])}_{\model}                                              \\
                                               & = \jump{\glang(\term[2] \cup \term[3])}_{\model}. \tag{Def.\ of $\glang$}
    \end{align*}
    
    Case $\term = \term[2]^*$:
    \begin{align*}
         & \jump{\term[2]^{*}}_{\model} = \bigcup_{n \in \nat} \jump{\term[2]^n}_{\model}                     \tag{Def.\ of $\jump{}$}                                                                                 \\
         & = \bigcup_{n \in \nat} \underbrace{\jump{\term[2]}_{\struc} \cdot \ldots \cdot \jump{\term[2]}_{\struc}}_{n \mbox{ times}}  \tag{Def.\ of $\jump{}$}                                                        \\
         & = \bigcup_{n \in \nat} \underbrace{\jump{\glang(\term[2])}_{\struc} \cdot \ldots \cdot \jump{\glang(\term[2])}_{\struc}}_{n \mbox{ times}}       \tag{IH}                                                   \\
         & = \bigcup_{n \in \nat} \bigcup_{G_1, \dots, G_n \in  \mathcal{G}(\term[2])} \jump{G_1}_{\model} \cdot \ldots \cdot \jump{G_n}_{\model}                         \tag{$\cdot$ is distributive w.r.t.\ $\cup$} \\
         & = \bigcup_{n \in \nat} \bigcup_{G_1, \dots, G_n \in  \mathcal{G}(\term[2])} \jump{G_1 \cdot \ldots \cdot G_n}_{\model}  \tag{\Cref{proposition: graph operation: cdot}}                                     \\
         & = \bigcup_{n \in \nat} \jump{\mathcal{G}(\term[2]^n)}_{\model}  \tag{Def.\ of $\term[2]^{n}$}                                                                                                               \\
         & = \jump{\mathcal{G}(\term[2]^*)}_{\model}. \tag{Def.\ of $\mathcal{G}$}
    \end{align*}
    
    Case $\term = \term[2]^{\smile}$:
    \begin{align*}
        \jump{\term[2]^{\smile}}_{\model} & = \jump{\term[2]}_{\model}^{\smile} \tag{Def.\ of $\jump{}$}                                                    \\
                                          & = \jump{\glang(\term[2])}_{\model}^{\smile}        \tag{IH}                                                     \\
                                          & = \bigcup_{G \in \glang(\term[2])} \jump{G}_{\model}^{\smile}  \tag{$\smile$ is distributive w.r.t.\ $\cup$}    \\
                                          & = \bigcup_{G \in \glang(\term[2])} \jump{G^{\smile}}_{\model}  \tag{\Cref{proposition: graph operation: smile}} \\
                                          & = \jump{\glang(\term[2]^{\smile})}_{\model}.       \tag{Def.\ of $\mathcal{G}$}
    \end{align*}
\end{proof}

\begin{proposition}[restatement of \Cref{proposition: graph characterization}]
    For every $\PCoRTC$ terms, $\term[1]$ and $\term[2]$, we have
    \[{} \models_{\REL} \term[1] \le \term[2] \iff \forall G \in \glang(\term[1]). \exists H \in \glang(\term[2]).\  H \homo G. \]
\end{proposition}
\begin{proof}
    By the following formula transformation:
    \begin{align*}
         & {}\models_{\REL} \term[1] \le \term[2] \iff \forall \model.\ \jump{\term[1]}_{\model} \subseteq  \jump{\term[2]}_{\model} \tag{Def.\ of $\models_{\REL}$}      \\
         & \iff \forall \model.\ \jump{\glang(\term[1])}_{\model} \subseteq  \jump{\glang(\term[2])}_{\model} \tag{\Cref{proposition: semantics for graphs}}              \\
         & \iff \forall \model. \forall G \in \glang(\term[1]).\  \jump{G}_{\model} \subseteq \bigcup_{H \in \glang(\term[2])} \jump{H}_{\model} \tag{Def.\ of $\jump{}$} \\
         & \iff \forall G \in \glang(\term[1]). \forall \model.\  \jump{G}_{\model} \subseteq \bigcup_{H \in \glang(\term[2])} \jump{H}_{\model}                          \\
         & \iff \forall G \in \glang(\term[1]). \forall \ppmodel.                                                                                                         \\
         & \quad (G \homo \ppmodel) \mbox{ implies } (\exists H \in \glang(\term[2]). H \homo \ppmodel)       \tag{Def.\ of $\jump{}$}                                    \\
         & \iff \forall G \in \glang(\term[1]). \exists H \in \glang(\term[2]).\ H \homo G. \tag{$\heartsuit$}
    \end{align*}
    
    Here, for ($\heartsuit$),
    $\Longrightarrow$:
    Let $\ppstruc$ be the saturation (see \Cref{section: graph characterization}) of $G$ s.t.\
    \begin{itemize}
        \item $a^{\ppstruc} = a^{G}$ for $a \in \SIG$;
        \item $\const{I}^{\ppstruc} = \set{\tuple{x, x} \mid x \in |\ppstruc|}$.
    \end{itemize}
    ($\ppstruc$ is a $2$-pointed structure, because $\comnf{a}^{\ppstruc} = |\ppstruc|^2 \setminus a^{\ppstruc}$ and $\const{I}^{\ppstruc}$ is the identity relation.)
    Since $G \homo \ppstruc$ by the identity map,
    there is some $H \in \glang(\term[2])$ s.t.\ $H \homo \ppmodel$ (by the assumption).
    Because $a^{H} = \emptyset$ for every $a \in \SIG_{\const{I}}^{(-)} \setminus \SIG$ (note that $\term[2]$ is \emph{positive}) and $a^{\ppstruc} = a^{G}$ for every $a \in \SIG$,
    we have $H \homo G$ by the same homomorphism for $H \homo \ppmodel$.
    $\Longleftarrow$:
    Let $H \in \glang(\term[2])$ be s.t.\ $H \homo G$.
    Let $\ppmodel$ be any $2$-pointed structure s.t.\ $G \homo \ppmodel$.
    Then, by transitivity of $\homo$, we have $H \homo \ppmodel$.
    Hence, $\exists H \in \glang(\term[2]).\  H \homo \ppmodel$.
\end{proof}

\section{Proof of \Cref{proposition: semantics for graphs ECoR}}\label{section: proof: proposition: semantics for graphs ECoR}

\begin{proposition}[restatement of \Cref{proposition: semantics for graphs ECoR}]
    For every structure $\model$ and $\ECoRTC$ term $\term$, we have
    $\jump{\term}_{\model} = \jump{\glang(\term)}_{\model}$.
\end{proposition}
\begin{proof}
    By induction on $\term[1]$, as with \Cref{proposition: semantics for graphs}.
    
    Case $\term = a$ where $a \in \SIG_{\const{I}}^{(-)} \setminus \SIG_{\const{I}}$:
    \begin{align*}
        \tuple{x, y} \in \jump{a}_{\model} & \iff \tuple{x, y} \in a^{\model}                                                         \tag{Def.\ of $\jump{}$}                                     \\
                                           & \iff \begin{tikzpicture}[baseline = -.5ex]
                                                      \graph[grow right = 1.2cm, branch down = 6ex, nodes={mynode}]{
                                                      {0/{}[draw, circle]}-!-{1/{}[draw, circle]}
                                                      };
                                                      \node[left = .5em of 0](l){};
                                                      \node[right = .5em of 1](r){};
                                                      \path (0) edge [draw = white, opacity = 0] node[pos= 0.5, inner sep = 1.5pt, font = \scriptsize, opacity = 1](a1){$a$}(1);
                                                      \graph[use existing nodes, edges={color=black, pos = .5, earrow}, edge quotes={fill=white, inner sep=1pt,font= \scriptsize}]{
                                                          0 -- a1 -> 1;
                                                          l -> 0; 1 -> r;
                                                      };
                                                  \end{tikzpicture} \homo {\model}[x,y]                        \\
                                           & \iff \tuple{x, y} \in \jump{\begin{tikzpicture}[baseline = -.5ex]
                                                                                 \graph[grow right = 1.2cm, branch down = 6ex, nodes={mynode}]{
                                                                                 {0/{}[draw, circle]}-!-{1/{}[draw, circle]}
                                                                                 };
                                                                                 \node[left = .5em of 0](l){};
                                                                                 \node[right = .5em of 1](r){};
                                                                                 \path (0) edge [draw = white, opacity = 0] node[pos= 0.5, inner sep = 1.5pt, font = \scriptsize, opacity = 1](a1){$a$}(1);
                                                                                 \graph[use existing nodes, edges={color=black, pos = .5, earrow}, edge quotes={fill=white, inner sep=1pt,font= \scriptsize}]{
                                                                                     0 -- a1 -> 1;
                                                                                     l -> 0; 1 -> r;
                                                                                 };
                                                                             \end{tikzpicture}}_{\model} \tag{\Cref{definition: semantics for graphs}} \\
                                           & \iff \tuple{x, y} \in \jump{\glang(a)}_{\model}. \tag{Def.\ of $\glang$}
    \end{align*}
    
    For the other cases, they are in the same way as the proof of \Cref{proposition: semantics for graphs}.
\end{proof}

\section{Proof completion of \Cref{proposition: Thompson}}\label{section: proposition: Thompson}
\begin{proposition}\label{proposition: Thompson 2}
    For every $\ExKA$ term $\term[3]$ and structure $\struc$ over $\SIG$, we have
    \[\jump{\term[3]}_{\struc} = \bigcup_{w \in [\term[3]]} \jump{w}_{\struc}.\]
    Here, $[\term[3]] \subseteq (A_{\const{I}}^{(-, \smile)} \setminus \set{\const{I}})^*$ is the language of $\term[3]$ over $A_{\const{I}}^{(-, \smile)} \setminus \set{\const{I}}$ obtained by viewing $\term[3]$ as a term over $A_{\const{I}}^{(-, \smile)} \setminus \set{\const{I}}$.
\end{proposition}
\begin{proof}
    By induction on $\term[3]$.
    
    Case $\term[3] = a$ for $a \in A_{\const{I}}^{(-, \smile)}$ (including $a = \const{I}$):
    Since $[a] = \set{a}$,
    we have $\jump{a}_{\struc} = \bigcup_{w \in [a]} \jump{w}_{\struc}$.
    
    Case $\term[3] = \bot$:
    Since $[\bot] = \emptyset$,
    we have 
    $\jump{\bot}_{\struc} = \emptyset = \bigcup_{w \in \emptyset} \jump{w}_{\struc} = \bigcup_{w \in [\bot]} \jump{w}_{\struc}$.
    
    Case $\term[3] = \term[1] \cup \term[2]$:
    \begin{align*}
        \jump{\term[1] \cup \term[2]}_{\struc}
         & = \jump{\term[1]}_{\struc} \cup \jump{\term[2]}_{\struc}                                                        \\
         & = (\bigcup_{w \in [\term[1]]} \jump{w}_{\struc}) \cup (\bigcup_{w \in [\term[2]]} \jump{w}_{\struc})   \tag{IH} \\
         & = \bigcup_{w \in [\term[1]] \cup [\term[2]]} \jump{w}_{\struc}                                                  \\
         & = \bigcup_{w \in [\term[1] \cup \term[2]]} \jump{w}_{\struc}.
    \end{align*}
    
    Case $\term[3] = \term[1] \cdot \term[2]$:
    \begin{align*}
        \jump{\term[1] \cdot \term[2]}_{\struc}
         & = \jump{\term[1]}_{\struc} \cdot \jump{\term[2]}_{\struc}                                                          \\
         & = (\bigcup_{w \in [\term[1]]}\jump{w}_{\struc}) \cdot (\bigcup_{v \in [\term[2]]}\jump{v}_{\struc}) \tag{IH}       \\
         & = \bigcup_{w \in [\term[1]], v \in [\term[2]]} (\jump{w}_{\struc} \cdot \jump{v}_{\struc}) \tag{By distributivity} \\
         & = \bigcup_{w \in [\term[1]], v \in [\term[2]]} \jump{w v}_{\struc}                                                 \\
         & = \bigcup_{w \in [\term[1] \cdot \term[2]]} \jump{w}_{\struc}.
    \end{align*}
    
    Case $\term[3] = \term[1]^*$:
    We have
    \begin{align*}
        \jump{\term[1]^*}_{\struc}
         & = \bigcup_{n \in \nat} \jump{\term[1]^n}_{\struc} = \bigcup_{n \in \nat} \jump{\term[1]}_{\struc}^n                                                  \\
         & = \bigcup_{n \in \nat} (\bigcup_{w \in [\term[1]]}\jump{w}_{\struc})^n \tag{IH}                                                                      \\
         & = \bigcup_{n \in \nat} \bigcup_{w_1, \dots, w_n \in [\term[1]]} (\jump{w_1}_{\struc} \cdot \ldots \cdot \jump{w_n}_{\struc}) \tag{By distributivity} \\
         & = \bigcup_{n \in \nat} \bigcup_{w_1, \dots, w_n \in [\term[1]]} \jump{w_1 w_2 \dots w_n}_{\struc}                                                    \\
         & = \bigcup_{n \in \nat} \bigcup_{w \in [\term[1]^n]} \jump{w}_{\struc} = \bigcup_{w \in [\term[1]^*]} \jump{w}_{\struc}. \tag*{\qedhere}
    \end{align*}
\end{proof}

\section{Proof completion of \Cref{lemma: Hintikka bound}} \label{section: lemma: Hintikka bound}
\begin{proof}[Proof completion of \Cref{lemma: Hintikka bound} ($P'$ is an saturable path)]
    Let
    \[\tuple{l_0, \dots, l_{x}, l_{x+1}, \dots, l_{n'}} = \tuple{0, \dots, x, y + 1, \dots, n}.\]
    (Here, ``$y + 1, \dots, n$'' is the empty sequence if $y = n$.)
    Then, $w' = a_{l_0} \dots a_{l_{n'}}$ holds and $P' = \tuple{G', \set{U_{i}'}_{i \in [0, n']}}$ is as follows:
    \begin{itemize}
        \item $U_{i}' = U_{l_{i}}$, for each $i \in [0, n']$;
        \item $G'$ is the $\const{I}$-saturation of $G(w')$ such that $\const{I}^{G'} = \set{\tuple{i, j} \in [0, n']^2 \mid {i = j} \lor \lnot \Con_{\comnf{\const{I}}}^{\automaton[2]}(U_{i}', U_{j}')}$.
    \end{itemize}
    If $\lnot \Con_{\comnf{\const{I}}}^{\automaton[2]}(U_{i}', U_{j}')$, then $\Con_{\const{I}}^{\automaton[2]}(U_{i}', U_{j}')$ (by \ref{definition: Hintikka: saturate} for $P$); thus $U_{i}' = U_{j}'$.
    Therefore, the binary relation $\set{\tuple{i, j} \in [0, n]^2 \mid \lnot \Con_{\comnf{\const{I}}}^{\automaton[2]}(U_{i}', U_{j}')}$ is symmetric and transitive;
    thus $\const{I}^{\graph[1]'}$ is an equivalence relation.
    $G'$ is consistent because $G$ is consistent.
    (In particular, when $\tuple{x, y+1} \in \const{I}^{G}$ and $a_{y+1} = \comnf{\const{I}}$,
    we have $U_{x} = U_{y}$ (by the definition of $x$ and $y$) and $U_{x} = U_{y+1}$ (by $\tuple{x, y+1} \in \const{I}^{G}$),
    thus $\Con_{\comnf{\const{I}}}^{\automaton[2]}(U_{x}', U_{x+1}')$ (since $a_{y+1} = \comnf{\const{I}}$ with \ref{definition: Hintikka: consistent} for $P$).
    Hence, $\tuple{x, x+1} \in \comnf{\const{I}}^{G'}$. Thus, since $a_{y+1} = \comnf{\const{I}}$, $G'$ is consistent.)
    $\graph[1]'$ is an edge-extension of $G(w)$, because
    $\const{I}^{\graph[1]'} \supseteq \const{I}^{G(w)} = \emptyset$ and $\comnf{\const{I}}^{\graph[1]'} \supseteq \comnf{\const{I}}^{G(w)}$ (by the definition of $\const{I}^{\graph[1]'}$).
    Hence, $\graph[1]'$ is indeed an $\const{I}$-saturation of $G(w)$.
    
    We show that $P'$ is an saturable path for $\not\models_{\REL} w' \le \automaton[2]$.
    \begin{itemize}
        \item For \ref{definition: Hintikka: u src tgt}:
              By $U_{0}' = U_{0}$ and $U_{n'}' = U_{n}$ with \ref{definition: Hintikka: u src tgt} for $P$.
              
        \item For \ref{definition: Hintikka: consistent}:
              \begin{itemize}
                  \item for $a = \const{I}$:
                        Let $\tuple{i, j} \in \const{I}^{\graph[1]'}$.
                        If $i = j$, then we have $\Con_{\const{I}}^{\automaton[2]}(U_{i}', U_{j}')$ by \ref{definition: Hintikka: consistent} for $P$.
                        If $\lnot \Con_{\comnf{\const{I}}}^{\automaton[2]}(U_{i}', U_{j}')$, then we have $\Con_{\const{I}}^{\automaton[2]}(U_{i}', U_{j}')$ by \ref{definition: Hintikka: saturate} for $P$.
                        
                  \item for $a = \comnf{\const{I}}$:
                        By the definition of $\const{I}^{\graph[1]'}$.
                        
                  \item For $\Con_{a_{l_{i}}}^{\automaton[2]}(U_{i-1}', U_{i}')$ where $i \in [n']$:
                        
                        \noindent Case $i \neq x + 1$:
                        By $\Con_{a_{l_{i}}}^{\automaton[2]}(U_{i-1}', U_{i}')$ (since $\tuple{l_{i-1}, l_{i}} \in a_{l_{i}}^{G}$ with \ref{definition: Hintikka: consistent} for $P$).
                        
                        \noindent Case $i = x + 1$:
                        Then, $U_{x}' = U_{l_{x}} = U_{x} = U_{y}$ and $U_{x+1}' = U_{l_{x+1}} = U_{y+1}$.
                        By $\Con_{a_{y}}^{\automaton[2]}(U_{y}, U_{y+1})$ (since $\tuple{y, y+1} \in a_{y}^{G}$ with \ref{definition: Hintikka: consistent} for $P$),
                        we have $\Con_{a_{l_{x+1}}}^{\automaton[2]}(U_{x}', U_{x+1}')$.
              \end{itemize}
              
        \item For \ref{definition: Hintikka: saturate}:
              By $\set{U_{i}' \mid i \in [0, n']} \subseteq \set{U_{i} \mid i \in [0, n]}$ with \ref{definition: Hintikka: saturate} for $P$.
    \end{itemize}
    Hence, $P'$ is an saturable path for $\not\models_{\REL} w' \le \automaton[2]$.
\end{proof} \fi
\end{document}